\documentclass[twocolumn]{svjour3}          
\usepackage[english]{babel}
\usepackage{kbordermatrix}
\usepackage{color,soul}
\usepackage{blkarray, colortbl,etoolbox,mathrsfs,amsmath,epsfig,multirow,amssymb,amsfonts,graphicx,subfig,cite,paralist,wrapfig,setspace,bm}
\usepackage{enumitem}
\usepackage{mathtools}
\usepackage[ruled,vlined,nokwfunc]{algorithm2e}
\usepackage[dvipsnames]{xcolor}
\usepackage[titletoc,title]{appendix}
\RequirePackage{amsmath, amssymb}

\makeatletter
\newcommand{\removelatexerror}{\let\@latex@error\@gobble}
\makeatother

\newcommand{\specialcell}[2][c]{%
  {\renewcommand{\arraystretch}{1.2}%
   \begin{tabular}[#1]{@{}l@{}}#2\end{tabular}}}

\definecolor{LightGray}{rgb}{.8,.8,.8}
\def\Big#1{\makebox(0,0){{\Large#1}}}

\makeatletter
\def\mathcolor#1#{\@mathcolor{#1}}
\def\@mathcolor#1#2#3{%
  \protect\leavevmode
  \begingroup
    \color#1{#2}#3%
  \endgroup
}
\makeatother

\newcommand{\ie}{\textit{i.e.,~}}

\newcommand{\eg}{\textit{e.g.,~}}
\newcommand{\Resp}{\textit{resp.~}}

\newcommand{\wrt}{\textit{w.r.t.~}}
\newcommand{\etal}{\textit{et~al.~}}

\newcommand{\IncUSR}{\mbox{\textsf{Inc-uSR}}}
\newcommand{\IncUSRone}{\mbox{\textsf{Inc-uSR}}}
\newcommand{\IncUSRtwo}{\mbox{\textsf{Inc-uSR-C1}}}
\newcommand{\IncUSRthree}{\mbox{\textsf{Inc-uSR-C2}}}
\newcommand{\IncUSRfour}{\mbox{\textsf{Inc-uSR-C3}}}
\newcommand{\IncSRAll}{\mbox{\textsf{Inc-SR-All}}}
\newcommand{\IncSRAllP}{\mbox{\textsf{Inc-SR-All-P}}}
\newcommand{\LTSF}{\mbox{\textsf{L-TSF}}}
\newcommand{\IncBSR}{\mbox{\textsf{Inc-bSR}}}

\newcommand{\PartialSim}{\mbox{\textsf{PartialSim}}}

\newcommand{\IncSR}{\mbox{\textsf{Inc-SR}}}
\newcommand{\IncSVD}{\mbox{\textsf{Inc-SVD}}}
\newcommand{\Batch}{\mbox{\textsf{Batch}}}

\newcommand{\YOUTU}{\mbox{\textsc{YouTu}}}
\newcommand{\CITH}{\mbox{\textsc{CitH}}}
\newcommand{\DBLP}{\mbox{\textsc{DBLP}}}

\newcommand{\WEBB}{\mbox{\textsc{WebB}}}
\newcommand{\WEBG}{\mbox{\textsc{WebG}}}
\newcommand{\CITP}{\mbox{\textsc{CitP}}}
\newcommand{\SOCL}{\mbox{\textsc{SocL}}}
\newcommand{\UK}{\mbox{\textsc{UK05}}}
\newcommand{\IT}{\mbox{\textsc{IT04}}}

\newcommand{\AFF}{\mbox{\textsf{AFF}}}
\newcommand{\op}{\mbox{\textsf{op}}}

\begin{document}

\title{Dynamical SimRank Search on Time-Varying Networks
}


\author{}
\institute{}

\author{Weiren Yu     \and
        Xuemin Lin    \and
        Wenjie Zhang  \and
        Julie A. McCann
}


\institute{W. Yu \at
              School of Engineering and Applied Science,
              Aston University, \\
              \email{w.yu3@aston.ac.uk}           
           \and
           X. Lin \and W. Zhang \at
              School of Computer Science and Engineering,\\
              University of New South Wales, \\
             \email{\{lxue, zhangw\}@cse.unsw.edu.au}
           \and
           J. A. McCann \at
              Department of Computing,
              Imperial College London, \\
              \email{j.mccann@imperial.ac.uk}           
}

\date{Received: date / Accepted: date}

\maketitle

\begin{abstract}
SimRank is an appealing pair-wise similarity measure based on graph structure.
It iteratively follows the intuition that two nodes are assessed as similar if they are pointed to by similar nodes.
Many real graphs are large, and links are constantly subject to minor changes.
In this article, we study the efficient dynamical computation of all-pairs SimRanks on time-varying graphs.
Existing methods for the dynamical SimRank computation (\eg L-TSF \cite{Shao2015} and READS \cite{Zhang2017}) mainly focus on top-$k$ search with respect to a given query.
For all-pairs dynamical SimRank search, Li \etal\!\!'s approach \cite{Li2010} was proposed for this problem.
It first factorizes the graph via a singular value decomposition (SVD),
and then incrementally maintains such a factorization in response to link updates at the expense of exactness.
As a result, all pairs of SimRanks are updated approximately,
yielding $O({r}^{4}n^2)$ time and $O({r}^{2}n^2)$ memory in a graph with $n$ nodes,
where $r$ is the target rank of the low-rank SVD.

Our solution to the dynamical computation of SimRank comprises of five ingredients:
(1) We first consider edge update that does not accompany new node insertions.
We show that the SimRank update $\mathbf{\Delta S}$ in response to every link update is expressible as a rank-one Sylvester matrix equation.
This provides an incremental method requiring $O(Kn^2)$ time and $O(n^2)$ memory in the worst case to update $n^2$ pairs of similarities for $K$ iterations.
(2) To speed up the computation further,
we propose a lossless pruning strategy that captures the ``affected areas'' of $\mathbf{\Delta S}$ to eliminate unnecessary retrieval.
This reduces the time of the incremental SimRank to $O(K(m+|\AFF|))$,
where $m$ is the number of edges in the old graph, and $|\AFF| \ (\le n^2)$ is the size of ``affected areas'' in $\mathbf{\Delta S}$,
and in practice, $|\AFF| \ll n^2$.
(3) We also consider edge updates that accompany node insertions, and categorize them into three cases,
according to which end of the inserted edge is a new node.
For each case, we devise an efficient incremental algorithm that can support new node insertions and accurately update the affected SimRanks.
(4) We next study batch updates for dynamical SimRank computation,
and design an efficient batch incremental method that handles ``similar sink edges'' simultaneously and eliminates redundant edge updates.
(5) To achieve linear memory,
we 
devise a memory-efficient strategy that dynamically updates all pairs of SimRanks column by column in just $O(Kn+m)$ memory,
without the need to store all $(n^2)$ pairs of old SimRank scores.
%
%
%
%
%
%
Experimental studies on various datasets demonstrate that our solution substantially outperforms the existing incremental SimRank methods,
and is faster and more memory-efficient than its competitors on million-scale graphs. 
%
%
%
\end{abstract}
\begin{figure*}[h!t] \centering
  \includegraphics[width=\linewidth]{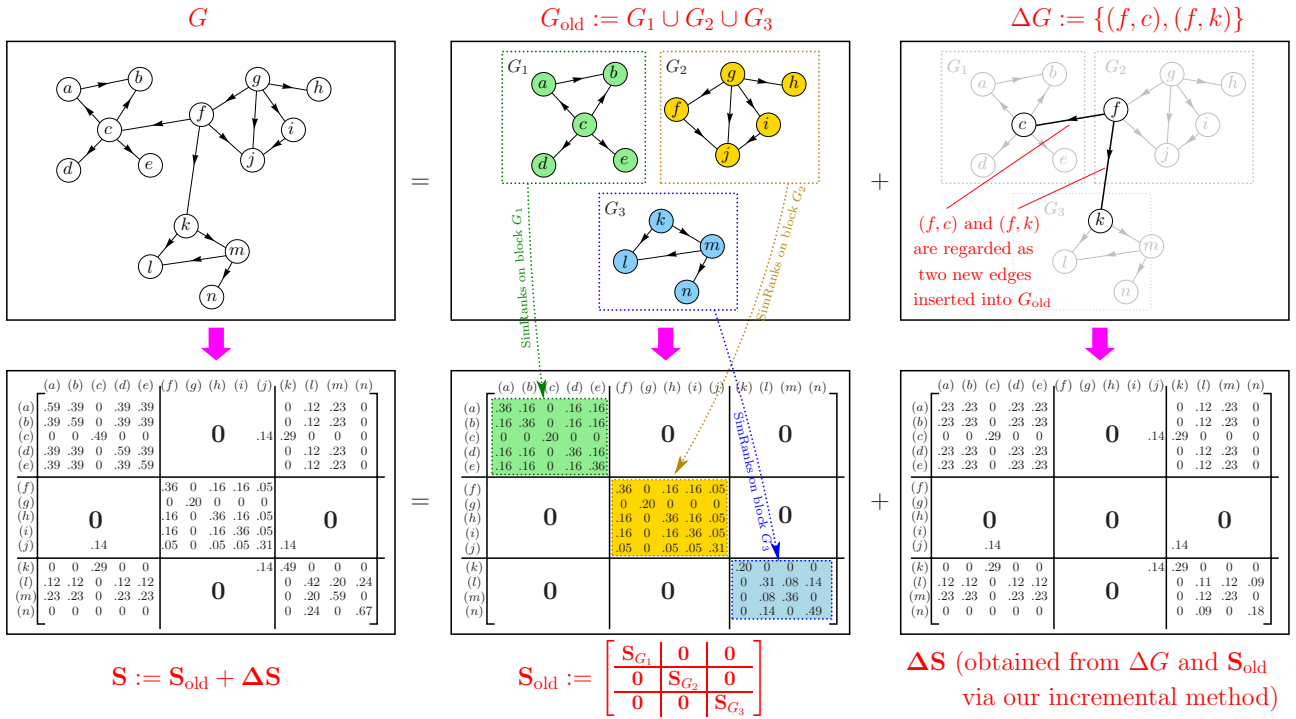}
  \caption{Incremental SimRank problem can decentralise large-scale SimRank retrieval over $G$} \label{fig:10} 
\end{figure*}
\section{Introduction}
Recent rapid advances in web data management reveal that
link analysis is becoming an important tool for similarity assessment.
Due to the growing number of applications in
\eg social networks, recommender systems, citation analysis, and link prediction \cite{Jeh2002},
a surge of graph-based similarity measures have surfaced over the past decade.
For instance, Brin and Page \cite{Berkhin2005} proposed a very successful relevance measure, called Google PageRank, to rank web pages.
Jeh and Widom \cite{Jeh2002} devised SimRank, an appealing pair-wise similarity measure that quantifies the structural equivalence of two nodes based on link structure.
Recently, Sun \etal \cite{Sun2011} invented PathSim to retrieve nodes proximities in a heterogeneous graph.
Among these emerging link based measures,
SimRank has stood out as an attractive one in recent years,
due to its simple and iterative philosophy that ``two nodes are similar if they are pointed to by similar nodes'',
coupled with the base case that ``every node is most similar to itself''.
This recursion 
not only allows SimRank to capture the global structure of a graph,
but also equips SimRank with mathematical insights that attract many researchers. 
For example, Fogaras and R{\'a}cz \cite{Fogaras2007} interpreted SimRank as the meeting time of the coalescing pair-wise random walks.
Li \etal \cite{Li2010} harnessed an elegant matrix equation to formulate the closed form of SimRank.

Nevertheless, the batch computation of SimRank is costly:
$O(Kd'n^2)$ time for all node-pairs \cite{Yu2013},
where $K$ is the total number of iterations,
and $d' \le d$ ($d$ is the average in-degree of a graph).
Generally, many real graphs are large, with links constantly evolving with minor changes.
This is especially apparent in \eg co-citation networks, web graphs, and social networks.
As a statistical example \cite{Ntoulas2004}, there are 5\%--10\% links updated every week in a web graph.
It is rather expensive to recompute similarities for all pairs of nodes from scratch when a graph is updated.
Fortunately, we observe that when link updates are small, the affected areas for SimRank updates are often small as well.
With this comes the need for incremental algorithms that compute changes to SimRank in response to link updates,
to discard unnecessary recomputations.
In this article, we investigate the following problem for SimRank evaluation:
\begin{description}
  \item[\textbf{Problem}] (\textsc{Incremental SimRank Computation})
  \item[\textbf{Given}] an old digraph $G$, old similarities in $G$, link changes $\Delta G$
\footnote{$\Delta G$ consists of a sequence of edges to be inserted/deleted.} to $G$, and a damping factor $C \in (0,1)$.
  \item[\textbf{Retrieve}] the changes to the old similarities.
\end{description}
\begin{figure*}[t] \centering
  \includegraphics[width=.93\linewidth]{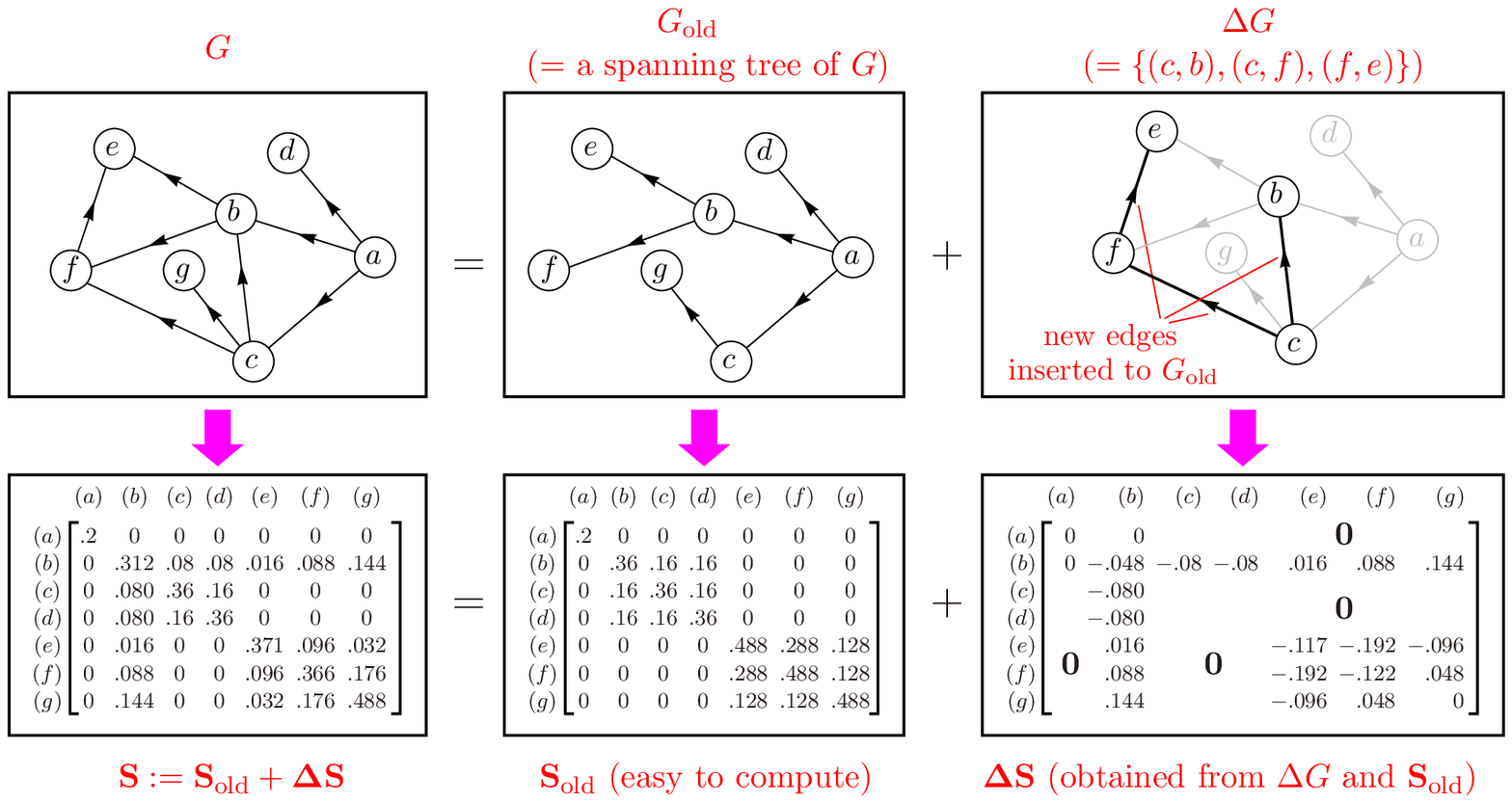}
  \caption{{Incremental SimRank problem can be applied to accelerate the batch computation of SimRank on $G$}} \label{fig:11} 
\end{figure*}
{{
Our research for the above SimRank problem is motivated by the following real applications:
\begin{example}[Decentralise Large-Scale SimRank Retrieval]  \label{eg:11}
Consider the web graph $G$ in Figure~\ref{fig:10}. There are $n=14$ nodes (web pages) in $G$, and each edge is a hyperlink.
To evaluate the SimRank scores of all $({n} \times n)$ pairs of web pages in $G$,
existing all-pairs SimRank algorithms need iteratively compute the SimRank matrix $\mathbf{S}$ of size $({n} \times n)$ in a centralised way (by using a single machine).
In contrast, our incremental approach can significantly improve the computational efficiency of all pairs of SimRanks by retrieving $\mathbf{S}$ in a decentralised way as follows:

We first employ a graph partitioning algorithm (\textit{e.g.,} METIS\footnote{http://glaros.dtc.umn.edu/gkhome/views/metis}) that can decompose the large graph $G$ into several small blocks such that the number of the edges with endpoints in different blocks is minimised. In this example, we partition $G$ into 3 blocks, $G_1 \cup G_2 \cup G_3$, along with 2 edges $\{(f,c),(f,k)\}$ across the blocks, as depicted in the first row of Figure~\ref{fig:10}.

Let $G_{\textrm{old}}:= G_1 \cup G_2 \cup G_3$ and $\Delta G:=\{(f,c),(f,k)\}$.
Then, $G$ can be viewed as ``$G_{\textrm{old}}$ perturbed by $\Delta G$ edge insertions''. That is,
\[
    G =  \overbrace{( G_1 \cup G_2 \cup G_3 )}^{:= G_{\textrm{old}}} \cup \overbrace{\{(f,c),(f,k)\}}^{:= \Delta G} = G_{\textrm{old}} \cup \Delta G
\]

Based on this decomposition,
we can efficiently compute $\mathbf{S}$ over $G$ by dividing $\mathbf{S}$ into two parts:
\[
\mathbf{S} = \mathbf{S}_{\textrm{old}} + \mathbf{\Delta S}
\]
where $\mathbf{S}_{\textrm{old}}$ is obtained by using a batch SimRank algorithm over $G_{\textrm{old}}$,
and $\mathbf{\Delta S}$ is derived from our proposed incremental method which takes $\mathbf{S}_{\textrm{old}}$ and $\Delta G$ as input.

It is worth mentioning that this way of retrieving $\mathbf{S}$ is far more efficient than directly computing $\mathbf{S}$ over $G$ via a batch algorithm.
There are two reasons:

Firstly, $\mathbf{S}_{\textrm{old}}$ can be efficiently computed in a decentralised way. It is a block diagonal matrix with no need of $n \times n$ space to store $\mathbf{S}_{\textrm{old}}$.
This is because $G_{\textrm{old}}$ is only comprised of several connected components $(G_1, G_2, G_3)$, and there are no edges across distinct components.
Thus, $\mathbf{S}_{\textrm{old}}$ exhibits a block diagonal structure:
\[
\mathbf{S}_{\textrm{old}} := diag(\mathbf{S}_{G_1}, \mathbf{S}_{G_2},\mathbf{S}_{G_3}) :=  \left[
\scalebox{1}{$
\begin{array}{c|c|c}
  {\mathbf{S}}_{G_1} & \mathbf{0} & \mathbf{0}   \\ \hline
  \mathbf{0} & {\mathbf{S}}_{G_2} & \mathbf{0} \\ \hline
  \mathbf{0} & \mathbf{0} & {\mathbf{S}}_{G_3}
\end{array}$} \right]
\]
To obtain $\mathbf{S}_{\textrm{old}}$, instead of applying the batch SimRank algorithm over the entire $G_{\textrm{old}}$,
we can apply the batch SimRank algorithm over each component $G_i \ (i=1,2,3)$ independently to obtain the $i$-th diagonal block of $\mathbf{S}_{\textrm{old}}$, $\mathbf{S}_{G_i}$.
Indeed, each $\mathbf{S}_{G_i}$ is computable in parallel.
Even if $\mathbf{S}_{\textrm{old}}$ is computed using a single machine,
only $O(n_1^2+n_2^2+n_3^2)$ space is required to store its diagonal blocks,
where $n_i$ is the number of nodes in each $G_i$,
rather than $O(n^2)$ space to store the entire $\mathbf{S}_{\textrm{old}}$ (see Figure~\ref{fig:10}).

Secondly, after graph partitioning, there are not many edges across components.
Small size of $\Delta G$ often leads to sparseness of $\mathbf{\Delta S}$ in general.
Hence, $\mathbf{\Delta S}$ is stored in a sparse format.
In addition, our incremental SimRank method will greatly accelerate the computation of $\mathbf{\Delta S}$.

Hence, along with graph partitioning, our incremental SimRank research will significantly enhance the computational efficiency of SimRank on large graphs, using a decentralised fashion. \qed
\end{example}

Our research on the incremental SimRank problem not only can decentralise large-scale SimRank retrieval, but also will enable a substantial speedup on the batch computation of SimRank, as indicated below.
\begin{example}[Accelerate Batch Computation of SimRank]  \label{eg:12}
Consider the citation network $G$ in Figure~\ref{fig:11},
where each node denotes a paper, and an edge a reference from one paper to another.
We wish to assess all pairs of similarities between papers.
Unlike existing batch computation that iteratively retrieves all-pairs SimRanks over the entire $G$,
our incremental method significantly accelerates the batch computation of SimRank as follows:

We first utilise BFS or DFS search to find a spanning tree (or arborescence) of $G$, denoted as $G_\textrm{old}$.
We observe that, due to the tree structure, all-pairs SimRank scores in $G_\textrm{old}$ are relatively easier to compute.
For example, each entry of Li \etal\!\!'s SimRank matrix $\mathbf{S}_{\textrm{old}}$ over $G_\textrm{old}$ can be obtained from a lightweight formula:
\[
\begin{split}
\mathbf{S}_{\textrm{old}}(a,b) = \scalebox{0.85}{$ \left\{
                                   \begin{array}{ll}
                                     0,  & \textrm{if $a$ and $b$ are not on the} \\
                                        &  \textrm{same level of the tree $G_\textrm{old}$;} \\
                                     C^{\lambda_{a,b}}(1-C^{H-\lambda_{a,b}+1}),  &\hbox{otherwise.}
                                   \end{array}
                                 \right. $}
\end{split}
\]
where $C$ is a damping factor, $\lambda_{a,b}$ is the number of edges from the lowest common ancestor of $(a,b)$ to node $a$ (or equivalently, to $b$) in the tree $G_\textrm{old}$, and $H$ is the number of edges from the root to node $a$ (or equivalently, to $b$) in the tree $G_\textrm{old}$.

Given $\mathbf{S}_{\textrm{old}}$, we next apply our incremental SimRank method that can significantly speed up the computation of new SimRank scores $\mathbf{S}$ in $G$.
Specifically, we denote by $\Delta G$ the set of edges in $G$ but not in $G_\textrm{old}$.
In Figure~\ref{fig:11},
\[
\Delta G := G - G_\textrm{old} = \{ (c,b), (c,f), (f,e) \}.
\]
Based on $\mathbf{S}_{\textrm{old}}$ and $\Delta G$, our incremental SimRank method can dynamically retrieve only the changes to $\mathbf{S}_{\textrm{old}}$ in response to $\Delta G$,
whose result $\mathbf{\Delta S}$ is a sparse matrix, as illustrated in Figure~\ref{fig:11}.

It is important to note that it does not require $n \times n$ memory to store $\mathbf{S}_{\textrm{old}}$ because $G_{\textrm{old}}$ is a tree structure.
If there are $H$ levels in the tree $G_{\textrm{old}}$ with $n_l$ nodes on level $l \ (l=1,\cdots,H)$, then we only need the space
\[ \textstyle
O \big( \sum_{l=1}^H ({{n_l}^2}) \big) \ll O \big(({\sum_{l=1}^H {n_l}})^2 \big) = O(n^2)
\]
to store the nonzero diagonal blocks of $\mathbf{S}_{\textrm{old}}$. \qed
\end{example}

These examples show that our incremental SimRank is useful to (i) decentralise large-scale SimRank retrieval, and (ii) accelerate the batch computation of SimRank.
Despite its usefulness,
existing work on incremental SimRank computation is rather limited.
To the best of our knowledge,
there is a relative paucity of work \cite{Jiang2017,Shao2015,Yu2014,Li2010} on incremental SimRank problems.
Shao \etal \cite{Shao2015} proposed a novel two-stage random-walk sampling scheme, named TSF,
which can support top-$k$ SimRank search over dynamic graphs.
In the preprocessing stage, TSF samples $R_g$ one-way graphs that serve as an index for querying process.
At query stage, for each one-way graph, $R_q$ new random walks of node $u$ are sampled.
However, the dynamic SimRank problems studied in \cite{Shao2015} and this work are different:
This work focuses on \emph{all $(n^2)$ pairs} of SimRank retrieval,
which requires $O(K(m+|\AFF|))$ time to compute the \emph{entire matrix} $\mathbf{S}$ in a deterministic style.
In Section~\ref{sec:07}, we have proposed a memory-efficient version of our incremental method that updates all pairs of similarities in a column-by-column fashion within only $O(Kn+m)$ memory.
In comparison,
\cite{Shao2015} focuses on top-$k$ SimRank dynamic search \wrt a given query $u$.
It incrementally retrieves \emph{only $k \ (\le n)$ nodes} with highest SimRank scores in \emph{a single column} $\mathbf{S}_{\star,u}$,
which requires $O(K^2 R_q R_g)$ \emph{average} query time\footnote{Recently, Jiang \etal \cite{Jiang2017} has argued that, to retrieve $\mathbf{S}_{\star,u}$, the querying time of Shao \etal\!\!'s TSF \cite{Shao2015} is $O(K n R_q R_g)$. The $n$ factor is due to the time to traverse the reversed one-way graph; in the worst case, all $n$ nodes are visited.}
 to retrieve $\mathbf{S}_{\star,u}$ along with $O(n \log k)$ time to return top-$k$ results from $\mathbf{S}_{\star,u}$.
Recently, Jiang \etal \cite{Jiang2017} pointed out that the probabilistic error guarantee of Shao \etal\!\!'s method is based on the assumption that no cycle in the given graph has a length shorter than $K$,
and they proposed READS, a new efficient indexing scheme that improves precision and indexing space for dynamic SimRank search.
The querying time of READS is $O(r n)$ to retrieve one column $\mathbf{S}_{\star,u}$, where $r$ is the number of sets of random walks.
Hence, TSF and READS are highly efficient for \emph{top-$k$ single-source} SimRank search.
Moreover, optimization methods in this work are based on a rank-one Sylvester
matrix equation characterising changes to the entire SimRank matrix $\mathbf{S}$ for all-pairs dynamical search,
which is fundamentally different from \cite{Shao2015,Jiang2017}'s methods that maintain one-way graphs (or SA forests) updating.
\begin{figure}[t] \centering
  \includegraphics[width=1\linewidth]{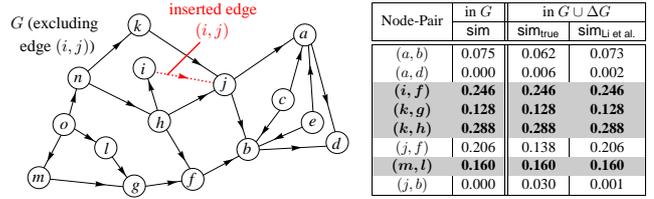}
  \caption{Incrementally update SimRanks when a new edge $(i,j)$ (with $\{i, j\} \subseteq V$) is inserted into $G=(V,E)$} \label{fig:01} 
\end{figure}
It is important to note that, for large-scale graphs, our incremental methods do not need to memoize all $(n^2)$ pairs of old SimRank scores.
Instead, $\mathbf{S}$ can be dynamically updated column-wisely in $O(Kn+m)$ memory.
For updating each column of $\mathbf{S}$,
our experiments in Section~\ref{sec:09} verify that
our memory-efficient incremental method is scalable on large real graphs while running 4--7x faster than the dynamical TSF \cite{Shao2015} per edge update,
due to the high cost of \cite{Shao2015} merging one-way graph's log buffers for TSF indexing.}}

Among the existing studies \cite{Jiang2017,Shao2015,Li2010} on dynamical SimRank retrieval,
the problem setting of Li \etal\!\!'s \cite{Li2010} on all-pairs dynamic search is exactly the same as ours:
the goal is to retrieve changes $\mathbf{\Delta S}$ to all-pairs SimRank scores $\mathbf{S}$,
given old graph $G$, link changes $\Delta G$ to $G$. 
To address this problem,
the central idea of \cite{Li2010} is to factorize the backward transition matrix $\mathbf{Q}$
of the original graph into $\mathbf{U} \cdot \mathbf{\Sigma} \cdot {\mathbf{V}}^T$
via a singular value decomposition (SVD) first,
and then incrementally estimate the updated matrices of $\mathbf{U}$, $\mathbf{\Sigma}$, ${\mathbf{V}}^T$ for link changes at the expense of exactness.
Consequently, updating all pairs of similarities entails $O({r}^{4}n^2)$ time and $O({r}^{2}n^2)$ memory yet without guaranteed accuracy,
where 
$r \ (\le n)$ is the target rank of the low-rank SVD approximation\footnote{According to \cite{Li2010}, using our notation,
$r \le \textrm{rank}(\mathbf{\Sigma} + \mathbf{U}^T \cdot \mathbf{\Delta Q} \cdot \mathbf{V})$,
where $\mathbf{\Delta Q}$ is the changes to $\mathbf{Q}$ for link updates.}.
This method is efficient to graphs when $r$ is extremely small, \eg a star graph $(r=1)$.
However, in general, $r$ is not always negligibly small.

(Please refer to Appendix~\ref{app:01} for a discussion in detail, and Appendix~\ref{app:03a} for an example.)
\subsection{Main Contributions}
Motivated by this,
we propose an efficient and accurate scheme for incrementally computing all-pairs SimRanks on link-evolving graphs.
Our main contributions consist of the following five ingredients:
\begin{itemize}[itemsep=3pt] 
\item
We first focus on unit edge update that does not accompany new node insertions.
By characterizing the SimRank update matrix $\mathbf{\Delta S}$ \wrt every link update as a rank-one Sylvester matrix equation,
we devise a fast incremental SimRank algorithm,
which entails $O(Kn^2)$ time in the worst case to update $n^2$ pairs of similarities for $K$ iterations. (Section~\ref{sec:04})
%
%
\item
To speed up the computation further,
we also propose an effective pruning strategy that captures the ``affected areas'' of $\mathbf{\Delta S}$ to discard unnecessary retrieval,
without loss of accuracy.
This reduces the time of incremental SimRank to $O(K(m+|\AFF|))$,
where
$|\AFF| \ (\le n^2)$ is the size of ``affected areas'' in $\mathbf{\Delta S}$,
and in practice, $|\AFF| \ll n^2$. (Section~\ref{sec:05}) 
\item
We also consider edge updates that accompany new node insertions, and distinguish them into three categories, according to which end of the inserted edge is a new node.
For each case, we devise an efficient incremental SimRank algorithm that can support new nodes insertion and accurately update affected SimRank scores. (Section~\ref{sec:06})
\item
We next investigate the batch updates of dynamical SimRank computation.
Instead of dealing with each edge update one by one,
we devise an efficient algorithm that can handle a sequence of edge insertions and deletions simultaneously,
by merging ``similar sink edges'' and minimizing unnecessary updates. (Section~\ref{sec:08}) 
\item
To achieve linear memory efficiency,
we also express $\mathbf{\Delta S}$ as the sum of many rank-one tensor products,
and devise a memory-efficient technique that updates all-pairs SimRanks in a column-by-column style in $O(Kn+m)$ memory,
without loss of exactness. (Section~\ref{sec:07})
%
%
\item
  We conduct extensive experiments on real and synthetic datasets to demonstrate that our algorithm
(a) is consistently faster than the existing incremental methods from several times to over one order of magnitude;
(b) is faster than its batch counterparts especially when link updates are small;
(c) for batch updates, runs faster than the repeated unit update algorithms; 
(d) entails linear memory and scales well on billion-edge graphs for all-pairs SimRank update; 
(e) is faster than {\LTSF} and its memory space is less than {\LTSF};
(f) entails more time on Cases (C0) and (C2) than Cases (C1) and (C3) for four edge types, and Case (C3) runs the fastest. (Section~\ref{sec:09})
\end{itemize}

This article is a substantial extension of our previous work~\cite{Yu2014}.
We have made the following new updates:
(1) In Section~\ref{sec:06}, we study three types of edge updates that accompany new node insertions.
This solidly extends \cite{Yu2014} and Li \etal's incremental method \cite{Li2010} whose edge updates disallow node changes.
%
(2) In Section~\ref{sec:08}, we
also investigate batch updates for dynamic SimRank computation, and devise an efficient algorithm that can handle ``similar sink edges'' simultaneously and discard unnecessary unit updates further.
(3) In Section~\ref{sec:07},  we propose a memory-efficient strategy that significantly reduces the memory from $O(n^2)$ to $O(Kn+m)$ for incrementally updating all pairs of SimRanks on million-scale graphs, without compromising running time and accuracy.
%
(4) In Section~\ref{sec:09}, we conduct additional experiments on real and synthetic datasets to verify the high scalability and fast computational time of our memory-efficient methods, as compared with the L-TSF method.
(5) In Section~\ref{sec:02}, we update the related work section by incorporating state-of-the-art SimRank research.
%
%
%
%
\section{SimRank Background} \label{sec:03}
In this section,
we give a broad overview of SimRank.
Intuitively,
the central theme behind SimRank is that
``two nodes are considered as similar if their incoming neighbors are themselves similar''.
Based on this idea,
there have emerged two widely-used SimRank models:
(1) Li \etal\!\!'s model (\eg \cite{Rothe2014,Yu2015a,Li2010,He2010,Fujiwara2013}) and
(2) Jeh and Widom's model (\eg \cite{Jeh2002,Lizorkin2008,Fogaras2005,Shao2015,Kusumoto2014}).

Throughout this article, our focus is on Li \etal\!\!'s SimRank model, also known as Co-SimRank in \cite{Rothe2014},
since the recent work \cite{Rothe2014} by Rothe and Sch\"{u}tze has showed that Co-SimRank is more accurate than Jeh and Widom's SimRank model in real applications such as bilingual lexicon extraction.
\noindent (Please refer to Remark~\ref{rem:01} for detailed explanations.)

%
%
\subsection{Li \etal\!\!'s SimRank model}
Given a directed graph ${G}=({V},{E})$ with a node set ${V}$ and an edge set ${E}$,
let $\mathbf{Q}$ be its backward transition matrix (that is, the transpose of the column-normalized adjacency matrix),
whose entry $[\mathbf{Q}]_{i,j}=1/\textrm{in-degree}(i)$ if there is an edge from $j$ to $i$, and 0 otherwise.
Then, Li \etal\!\!'s SimRank matrix, denoted by $\mathbf{S}$, is defined as
\begin{equation}  \label{eq:03a}
   {{\mathbf{S}}}= C\cdot (\mathbf{Q} \cdot {{\mathbf{S}}}\cdot {{\mathbf{Q}}^{T}} ) + (1-C) \cdot {{\mathbf{I}}_{n}},
\end{equation}
where $C \in \left( 0,1 \right)$ is a damping factor, which is generally taken to be 0.6--0.8,
and $\mathbf{I}_n$ is an $n \times n$ identity matrix $(n=|V|)$.
The notation ${(\star)}^{T}$ is the matrix transpose.

Recently, Rothe and Sch\"{u}tze \cite{Rothe2014} have introduced Co-SimRank, whose definition is
\begin{equation}  \label{eq:03b}
   {{\mathbf{\tilde{S}}}}= C\cdot (\mathbf{Q} \cdot {{\mathbf{\tilde{S}}}}\cdot {{\mathbf{Q}}^{T}} ) + {{\mathbf{I}}_{n}},
\end{equation}

Comparing Eqs.\eqref{eq:03a} and \eqref{eq:03b},
we can readily verify that Li \etal\!\!'s SimRank scores equal Co-SimRank scores scaled by a constant factor $(1-C)$, \ie
${{\mathbf{S}}} = (1-C) \cdot   {{\mathbf{\tilde{S}}}}$.
Hence, the relative order of all Co-SimRank scores in ${{\mathbf{\tilde{S}}}}$ is exactly the same as that of Li \etal\!\!'s SimRank scores in ${{\mathbf{S}}}$ even though the entries in ${{\mathbf{\tilde{S}}}}$ can be larger than 1.
That is, the ranking of Co-SimRank ${{\mathbf{\tilde{S}}}}(*,*)$ is identical to the ranking of Li \etal\!\!'s SimRank ${{\mathbf{S}}}(*,*)$.

\subsection{Jeh and Widom's SimRank model}
Jeh and Widom's SimRank model, in matrix notation, can be formulated as
\begin{equation}
   {{\mathbf{S}'}}= \max \{C\cdot (\mathbf{Q} \cdot {{\mathbf{S}'}}\cdot {{\mathbf{Q}}^{T}} ), {{\mathbf{I}}_{n}}\},
\end{equation}
where ${{\mathbf{S}'}}$ is their SimRank similarity matrix;
$\max \{ \mathbf{X}, \mathbf{Y} \}$ is matrix element-wise maximum,
\ie $[\max \{ \mathbf{X}, \mathbf{Y} \}]_{i,j}:=\max\{[\mathbf{X}]_{i,j}, [\mathbf{Y}]_{i,j}\}$.

\begin{remark} \label{rem:01}
The recent work by Kusumoto \etal \cite{Kusumoto2014} has showed that ${{\mathbf{S}}}$ and ${{\mathbf{S}'}}$ do not produce the same results.
Recently, Yu and McCann \cite{Yu2015a} have showed the subtle difference of the two SimRank models from a semantic perspective,
and also justified that Li \etal\!\!'s SimRank ${{\mathbf{S}}}$ can capture more pairs of self-intersecting paths that are neglected by Jeh and Widom's SimRank ${{\mathbf{S}'}}$.

The recent work \cite{Rothe2014} by Rothe and Sch\"{u}tze has demonstrated further that, in real applications such as bilingual lexicon extraction,
the ranking of Co-SimRank ${{\mathbf{\tilde{S}}}}$ (\ie the ranking of Li \etal\!\!'s SimRank ${{\mathbf{S}}}$) is more accurate than that of Jeh and Widom's SimRank ${{\mathbf{S}'}}$ (see \cite[Table~4]{Rothe2014}).

%
%
%
\end{remark}

Despite the high precision of Li \etal\!\!'s SimRank model,
the existing incremental approach of Li \etal \cite{Li2010} for updating SimRank does not always obtain the correct solution $\mathbf{S}$ to Eq.\eqref{eq:03a}.
(Please refer to Appendix~\ref{app:01} for theoretical explanations).

Table~\ref{tab:01} lists the notations used in this article.
\begin{table}\small
\begin{tabular}{c|p{6cm}}
  \hline
\textbf{Symbol}             & \textbf{Description} \\
  \hline
$n$                 & number of nodes in old graph $G$ \\
$m$                 & number of edges in old graph $G$ \\
${d_i}$             & in-degree of node $i$ in old graph $G$ \\
$d$                 & average in-degree of graph $G$ \\
$C$                 & damping factor ($0<C<1$) \\
$K$                 & iteration number \\
$\mathbf{e}_i$      & $n \times 1$ unit vector with a 1 in the $i$-th entry and 0s elsewhere \\
$\mathbf{Q} / \tilde{\mathbf{Q}}$        & old/new (backward) transition matrix \\
$\mathbf{S} / \tilde{\mathbf{S}}$        & old/new SimRank matrix \\
$\mathbf{I}_n$        & $n \times n$ identity matrix \\
${\mathbf{X}}^T$  &  transpose of matrix $\mathbf{X}$ \\
${[\mathbf{X}]}_{i,\star}$  &  $i$-th row of matrix $\mathbf{X}$ \\
${[\mathbf{X}]}_{\star,j}$  &  $j$-th column of matrix $\mathbf{X}$ \\
${[\mathbf{X}]}_{i,j}$  &  $(i,j)$-th entry of matrix $\mathbf{X}$ \\
  \hline
\end{tabular}
\caption{Symbol and Description} \label{tab:01}
\end{table}
%
%
\section{Edge Update without node insertions} \label{sec:04}
In this section, we consider edge update that does not accompany new node insertions,
\ie the insertion
of new edge $(i,j)$ into $G=(V,E)$ with $i \in V$ and $j \in V$.
In this case, the new SimRank matrix $\mathbf{\tilde{S}}$ and the old one $\mathbf{S}$ are of the same size.
As such, it makes sense to denote the SimRank change $\mathbf{\Delta {S}} $ as $\mathbf{\tilde{S}} -\mathbf{S}$.
Below we first introduce the big picture of our main idea, and then present rigorous justifications and proofs. 
%
%
\subsection{The main idea} 
For each edge $(i,j)$ insertion,
we can show that $\mathbf{\Delta Q}$ is a \emph{rank-one} matrix,
\ie there exist two column vectors $\mathbf{u},\mathbf{v} \in \mathbb{R}^{n \times 1} $ such that
$\mathbf{\Delta Q} \in \mathbb{R}^{n \times n}$ can be decomposed into the {outer product} of $\mathbf{u}$ and $\mathbf{v}$ as follows:
%
\begin{equation} \label{eq:12}
  \mathbf{\Delta Q} = \mathbf{u} \cdot \mathbf{v}^T.  
\end{equation}

Based on Eq.\eqref{eq:12},
we then have an opportunity to efficiently compute $\mathbf{\Delta S}$,
by characterizing it as
\begin{equation} \label{eq:13}
  \mathbf{\Delta S} = \mathbf{M} + \mathbf{M}^T, 
\end{equation}
where the auxiliary matrix $\mathbf{M}\in \mathbb{R}^{n \times n}$ satisfies the following \emph{rank-one} Sylvester equation:
\begin{equation} \label{eq:14}
\mathbf{M} = C \cdot \tilde{\mathbf{Q}} \cdot \mathbf{M} \cdot \tilde{\mathbf{Q}}^T + C \cdot \mathbf{u} \cdot \mathbf{w}^T.
\end{equation}
Here, $\mathbf{u}, \mathbf{w}$ are two obtainable column vectors:
$\mathbf{u}$ can be derived from Eq.\eqref{eq:12},
and $\mathbf{w}$ can be described by the old $\mathbf{Q}$ and $\mathbf{S}$
(we will provide their exact expressions later after some discussions);
and $\tilde{\mathbf{Q}}=\mathbf{Q} + \mathbf{\Delta Q}$.

Thus, computing $\mathbf{\Delta S}$ boils down to solving $\mathbf{M}$ in Eq.\eqref{eq:14}.
The main advantage of solving $\mathbf{M}$ via Eq.\eqref{eq:14},
as compared to directly computing the new scores $\tilde{\mathbf{S}}$ via SimRank formula
\begin{equation} \label{eq:15}
\tilde{\mathbf{S}} = C \cdot \tilde{\mathbf{Q}} \cdot \tilde{\mathbf{S}} \cdot \tilde{\mathbf{Q}}^T + (1-C) \cdot \mathbf{I}_n,
\end{equation}
is the high computational efficiency.
More specifically,
solving $\tilde{\mathbf{S}}$ via Eq.\eqref{eq:15} needs expensive \emph{matrix-matrix} multiplications,
whereas solving $\mathbf{M}$ via Eq.\eqref{eq:14} involves only \emph{matrix-vector} and \emph{vector-vector} multiplications,
which is a substantial improvement achieved by our observation that $(C \cdot \mathbf{u} \mathbf{w}^T) \in \mathbb{R}^{n \times n}$ in Eq.\eqref{eq:14} is a \emph{rank-one} matrix,
as opposed to the (full) \emph{rank-$n$} matrix $(1-C) \cdot \mathbf{I}_n$ in Eq.\eqref{eq:15}.
To further elaborate on this,
we can readily convert the recursive forms of Eqs.\eqref{eq:14} and \eqref{eq:15}, respectively, into the series forms:
\setlength\arraycolsep{2pt}
\begin{eqnarray}
\mathbf{M} & = & \sum\nolimits_{k=0}^{\infty }{{{C}^{k+1}}\cdot {{\tilde{\mathbf{Q}}}^{k}}\cdot \mathbf{u}\cdot {{\mathbf{w}}^{T}}\cdot {{({{\tilde{\mathbf{Q}}}^{T}})}^{k}}}, \label{eq:16} \\
\tilde{\mathbf{S}} & = & (1-C)\cdot \sum\nolimits_{k=0}^{\infty }{{{C}^{k}}\cdot {{\tilde{\mathbf{Q}}}^{k}}\cdot \mathbf{I}_n \cdot {{({{\tilde{\mathbf{Q}}}^{T}})}^{k}}}. \label{eq:17}
\end{eqnarray}

To compute the sums in Eq.\eqref{eq:16} for $\mathbf{M}$,
a conventional way is to memoize
$\mathbf{M}_0 \leftarrow C \cdot \mathbf{u}\cdot {{\mathbf{w}}^{T}}$ first
(where the intermediate result $\mathbf{M}_0$ is an $n \times n$ matrix),
and then iterate as
\[\mathbf{M}_{k+1} \leftarrow \mathbf{M}_{0} + C \cdot {\tilde{\mathbf{Q}}} \cdot \mathbf{M}_{k}  \cdot {\tilde{\mathbf{Q}}}^T, \quad (k=0,1,2,\cdots)\]
which involves expensive \emph{matrix-matrix} multiplications (\eg ${\tilde{\mathbf{Q}}} \cdot \mathbf{M}_{k}$).
In contrast,
our method takes advantage of the \emph{rank-one} structure of $\mathbf{u}\cdot {{\mathbf{w}}^{T}}$ to compute the sums in Eq.\eqref{eq:16} for $\mathbf{M}$,
by converting the conventional \emph{matrix-matrix} multiplications ${\tilde{\mathbf{Q}}} \cdot (\mathbf{u} {{\mathbf{w}}^{T}}) \cdot {\tilde{\mathbf{Q}}}^T$
into only \emph{matrix-vector} and \emph{vector-vector} multiplications $({\tilde{\mathbf{Q}}} \mathbf{u}) \cdot (\tilde{\mathbf{Q}} {{\mathbf{w}}})^T$.
To be specific, we leverage two vectors ${\bm{\xi}}_{k}, {\bm{\eta}}_{k}$,
and iteratively compute Eq.\eqref{eq:16} as
\begin{eqnarray}
\nonumber & &   \textrm{initialize } {{\bm{\xi }}_{0}}\leftarrow C\cdot \mathbf{u},\quad {{\bm{\eta }}_{0}}\leftarrow \mathbf{w},\quad {{\mathbf{M}}_{0}}\leftarrow C \cdot \mathbf{u}\cdot {{\mathbf{w}}^{T}}   \\
\nonumber & &   \textbf{for }  k=0,1,2,\cdots  \\
\nonumber & &   \quad {{\bm{\xi }}_{k+1}}\leftarrow C\cdot \mathbf{\tilde{Q}}\cdot {{\bm{\xi }}_{k}},\quad {{\bm{\eta }}_{k+1}}\leftarrow \mathbf{\tilde{Q}}\cdot {{\bm{\eta }}_{k}} \\
\label{eq:16a} & &   \quad {{\mathbf{M}}_{k+1}}\leftarrow {{\bm{\xi }}_{k+1}}\cdot \bm{\eta }_{k+1}^{T}+{{\mathbf{M}}_{k}}
\end{eqnarray}
so that \emph{matrix-matrix} multiplications are safely avoided.
%
%
\subsection{Describing $\mathbf{u}, \mathbf{v},\mathbf{w}$ in Eqs.\eqref{eq:12} and \eqref{eq:14}}
%

To obtain $ \mathbf{u}$ and $\mathbf{v}$ in Eq.\eqref{eq:12} at a low cost,
we have the following theorem.

\begin{theorem} \label{thm:01}
Given an old digraph $G=(V,E)$,
if there is a new edge $(i,j)$ with $i \in V$ and $j \in V$ to be added to $G$,
then the change to $\mathbf{Q}$ is an $n \times n$ rank-one matrix, \ie
$\mathbf{\Delta Q} = \mathbf{u} \cdot \mathbf{v}^T$,
where
\begin{equation} \label{eq:18} \small
\mathbf{u}=\left\{ \begin{matrix}
   {{\mathbf{e}}_{j}} & \left( {{d}_{j}}=0 \right)  \\
   \tfrac{1}{{{d}_{j}}+1}{{\mathbf{e}}_{j}} & \left( {{d}_{j}}>0 \right)  \\
\end{matrix} \right.
, \quad
\mathbf{v}=\left\{ \begin{matrix}
   {{\mathbf{e}}_{i}} & \left( {{d}_{j}}=0 \right)  \\
   {{\mathbf{e}}_{i}}-{{[\mathbf{Q}]}_{j,\star}^T} & \left( {{d}_{j}}>0 \right)  \\
\end{matrix} \right.
\end{equation}
\qed
%
\end{theorem}

(Please refer to Appendix~\ref{app:02a} for the proof of Theorem~\ref{thm:01}, and Appendix~\ref{app:03b} for an example.)

Theorem \ref{thm:01} suggests that the change $\mathbf{\Delta Q}$ is an $n\times n$ \emph{rank-one} matrix,
which can be obtain in only constant time from $d_j$ and ${{[\mathbf{Q}]}_{j,\star}^T}$.
%
In light of this, 
%
%
we next describe $\mathbf{w}$ in Eq.\eqref{eq:14} in terms of the old $\mathbf{Q}$ and $\mathbf{S}$ such that Eq.\eqref{eq:14} is a \emph{rank-one} Sylvester equation.

\begin{theorem} \label{thm:02}
Let $(i,j)_{i \in V, \ j \in V}$ be a new edge to be added to $G$ (\Resp an existing edge to be deleted from $G$).
Let $\mathbf{u}$ and $\mathbf{v}$ be the rank-one decomposition of $\mathbf{\Delta Q} = \mathbf{u} \cdot \mathbf{v}^T$. 
Then, (i) there exists a vector $\mathbf{w}=\mathbf{y}+\tfrac{\lambda }{2}\mathbf{u}$
with
\begin{equation} \label{eq:20}
  \mathbf{y}=\mathbf{Q}\cdot \mathbf{z},\quad \lambda ={{\mathbf{v}}^{T}}\cdot \mathbf{z},\quad \mathbf{z}=\mathbf{S}\cdot \mathbf{v}
\end{equation}
such that Eq.\eqref{eq:14} is the \emph{rank-one} Sylvester equation.

(ii) Utilizing the solution $\mathbf{M}$ to Eq.\eqref{eq:14},
the SimRank update matrix $\mathbf{\Delta S}$ can be represented by Eq.\eqref{eq:13}. \qed
\end{theorem}

(The proof of Theorem~\ref{thm:02} is in Appendix~\ref{app:02b}.)

Theorem~\ref{thm:02} provides an elegant expression of $\mathbf{w}$ in Eq.\eqref{eq:14}.
To be precise,
given $\mathbf{Q}$ and $\mathbf{S}$ in the old graph $G$, and an edge $(i,j)$ inserted to $G$,
one can find $\mathbf{u}$ and $\mathbf{v}$ via Theorem~\ref{thm:01} first,
and then resort to Theorem~\ref{thm:02} to compute $\mathbf{w}$ from $\mathbf{u},\mathbf{v},\mathbf{Q},\mathbf{S}$.
Due to the existence of the vector $\mathbf{w}$,
it can be guaranteed that the Sylvester equation \eqref{eq:14} is \emph{rank-one}.
Henceforth, our aforementioned method can be employed to iteratively compute $\mathbf{M}$ in Eq.\eqref{eq:16},
requiring no \emph{matrix-matrix} multiplications.
%
%
\subsection{Characterizing $\mathbf{\Delta S}$}
%

Leveraging Theorems \ref{thm:01} and \ref{thm:02},
we next characterize the SimRank change $\mathbf{\Delta S}$. 

\begin{theorem} \label{thm:03}
If there is a new edge $(i,j)$ with $i \in V$ and $j \in V$ to be inserted to $G$,
then the SimRank change $\mathbf{\Delta S}$ can be characterized as
\[
\mathbf{\Delta S}=\mathbf{M}+{{\mathbf{M}}^{T}} \quad \textrm{ with }
\]
\begin{equation}  \label{eq:29c}
\mathbf{M}=\sum\nolimits_{k=0}^{\infty }{{{C}^{k+1}}\cdot {{{\mathbf{\tilde{Q}}}}^{k}}\cdot {{\mathbf{e}}_{j}}\cdot {{\bm{\gamma }}^{T}}\cdot {{({{{\mathbf{\tilde{Q}}}}^{T}})}^{k}}},
\end{equation}
where the auxiliary vector $\bm{\gamma }$ is obtained as follows:

(i) when ${{d}_{j}}=0$,
\begin{equation}\label{eq:29aa}
\bm{\gamma } = \mathbf{Q}\cdot {{[\mathbf{S}]}_{\star,i}}+\tfrac{1}{2}{{[\mathbf{S}]}_{i,i}}\cdot {{\mathbf{e}}_{j}}
\end{equation}

(ii) when ${{d}_{j}}>0$,
\begin{equation}\label{eq:29bb}
\scalebox{0.92}{$
\bm{\gamma } = \tfrac{1}{({{d}_{j}}+1)} \left( \mathbf{Q} {{[\mathbf{S}]}_{\star,i}}-\tfrac{1}{C} {{[\mathbf{S}]}_{\star,j}}+( \tfrac{\lambda }{2\left( {{d}_{j}}+1 \right)}+ \tfrac{1}{C}-1 ) {{\mathbf{e}}_{j}} \right)$}
\end{equation}
%
%
and the scalar $\lambda$ can be derived from
\begin{equation} \label{eq:29b}
\lambda = {{[\mathbf{S}]}_{i,i}}+\tfrac{1}{C} \cdot {[\mathbf{S}]}_{j,j}-2\cdot {{[\mathbf{Q}]}_{j,\star}}\cdot {{[\mathbf{S}]}_{\star,i}} - \tfrac{1}{C} +1.
\end{equation}
\qed
\end{theorem}

(The proof of Theorem~\ref{thm:03} is in Appendix~\ref{app:02c}.)

Theorem \ref{thm:03} provides an efficient method to compute the incremental SimRank matrix $\mathbf{\Delta S}$,
by utilizing the previous information of $\mathbf{Q}$ and $\mathbf{S}$, 
as opposed to \cite{Li2010} that requires to maintain the incremental SVD.

%
\subsection{Deleting an edge $(i,j)_{i \in V, \ j \in V}$ from $G=(V,E)$} \label{sec:04f}
For an edge deletion, we next propose a Theorem~\ref{thm:03}-like technique that can efficiently update SimRanks.
\begin{theorem} \label{thm:09}
%
When an edge $(i,j)_{i \in V, \ j \in V}$ is deleted from $G=(V,E)$,
the changes to $\mathbf{Q}$ is a rank-one matrix, which can be described as $\mathbf{\Delta Q} = \mathbf{u} \cdot \mathbf{v}^T$,
where
\begin{equation*} \label{eq:19} 
\mathbf{u}=\left\{ \begin{matrix}
   {{\mathbf{e}}_{j}} & \left( {{d}_{j}}=1 \right)  \\
   \tfrac{1}{{{d}_{j}}-1}{{\mathbf{e}}_{j}} & \left( {{d}_{j}}>1 \right)  \\
\end{matrix} \right.
, \quad
\mathbf{v}=\left\{ \begin{matrix}
   {-{\mathbf{e}}_{i}} & \left( {{d}_{j}}=1 \right)  \\
   {{[\mathbf{Q}]}_{j,\star}^T}-{{\mathbf{e}}_{i}} & \left( {{d}_{j}}>1 \right)  \\
\end{matrix} \right.
\end{equation*}
The changes $\mathbf{\Delta S}$ to SimRank can be characterized as
%
%
\[
\mathbf{\Delta S}=\mathbf{M}+{{\mathbf{M}}^{T}} \ \textrm{ with }
\mathbf{M}=\sum\nolimits_{k=0}^{\infty }{{{C}^{k+1}} {{{\mathbf{\tilde{Q}}}}^{k}} {{\mathbf{e}}_{j}} {{\bm{\gamma }}^{T}} {{({{{\mathbf{\tilde{Q}}}}^{T}})}^{k}}},
\]
%
where the auxiliary vector $\bm{\gamma }:=$ 
%
%
\begin{equation*} 
\scalebox{0.86}{$
\left\{ \begin{array}{lc}
    -\mathbf{Q}\cdot {{[\mathbf{S}]}_{\star,i}}+\frac{1}{2}{{[\mathbf{S}]}_{i,i}}\cdot {{\mathbf{e}}_{j}}  & ({{d}_{j}}=1)  \\
   \tfrac{1}{({{d}_{j}}-1)} \left(\frac{1}{C}\cdot {{[\mathbf{S}]}_{\star,j}} - \mathbf{Q}\cdot {{[\mathbf{S}]}_{\star,i}} +( \frac{\lambda }{2\left( {{d}_{j}}-1 \right)}- \frac{1}{C}+1 )\cdot {{\mathbf{e}}_{j}} \right) & ({{d}_{j}}>1)  \\
\end{array} \right.
$}
\end{equation*}
and $\lambda:={{[\mathbf{S}]}_{i,i}}+\tfrac{1}{C} \cdot {[\mathbf{S}]}_{j,j}-2\cdot {{[\mathbf{Q}]}_{j,\star}}\cdot {{[\mathbf{S}]}_{\star,i}} - \tfrac{1}{C} +1$. \qed 
\end{theorem}

(The proof of Theorem~\ref{thm:09} is in Appendix~\ref{app:02d}.)

%
%
\subsection{{\IncUSRone} Algorithm}
%
We present our efficient incremental approach, denoted as {\IncUSRone} (in Appendix~\ref{app:04a}), that supports the edge insertion without accompanying new node insertions. 
The complexity of \IncUSRone~is bounded by $O(Kn^2)$ time and $O(n^2)$ memory\footnote{In the next sections, we shall substantially reduce its time and memory complexity further.}
in the worst case for updating all $n^2$ pairs of similarities.

(Please refer to Appendix~\ref{app:04a} for a detailed description of {\IncUSRone}, and Appendix~\ref{app:03c} for an example.)

\section{Pruning Unnecessary Node-Pairs in $\mathbf{\Delta S}$} \label{sec:05}
After the SimRank update matrix $\mathbf{\Delta S}$ has been characterized as a rank-one Sylvester equation,
pruning techniques can further skip node-pairs with unchanged SimRanks in $\mathbf{\Delta S}$ (called ``unaffected areas'').
\subsection{Affected Areas in $\mathbf{\Delta S}$}
%
We next reinterpret the series $\mathbf{M}$ in Theorem \ref{thm:03},
aiming to identify ``affected areas'' in $\mathbf{\Delta S}$.
Due to space limitations,
we mainly focus on the edge insertion case of $d_j>0$.
Other cases have the similar results.

By substituting Eq.\eqref{eq:29bb} back into Eq.\eqref{eq:29c},
we can readily split the series form of $\mathbf{M}$ into three parts:

\vspace{-8pt} \begin{small}
\begin{eqnarray*}
{[\mathbf{M}]}_{a,b}= && \tfrac{1}{{{d}_{j}}+1} \bigg(\underbrace{\sum\nolimits_{k=0}^{\infty }{{{C}^{k+1}} \cdot {{[{{{\mathbf{\tilde{Q}}}}^{k}}]}_{a,j}}  {{[\mathbf{S}]}_{i,\star}}  {{\mathbf{Q}}^{T}}\cdot {[{{({{{\mathbf{\tilde{Q}}}}^{T}})}^{k}}]}_{\star,b}}}_{\text{Part 1}} - \\
&& - \underbrace{\sum\nolimits_{k=0}^{\infty }{{{C}^{k}} {{[{{{\mathbf{\tilde{Q}}}}^{k}}]}_{a,j}}  {{[\mathbf{S}]}_{j,\star}}  {[{{({{{\mathbf{\tilde{Q}}}}^{T}})}^{k}}]}_{\star,b}}}_{\text{Part 2}} + \\
&&  + \mu  \underbrace{\sum\nolimits_{k=0}^{\infty }{{{C}^{k+1}}  {{[{{{\mathbf{\tilde{Q}}}}^{k}}]}_{a,j}}  {{[{{({{{\mathbf{\tilde{Q}}}}^{T}})}^{k}}]}_{j,b}}}}_{\text{Part 3}} \bigg)
\end{eqnarray*}
\end{small}
with the scalar $\mu :=\frac{\lambda }{2\left( {{d}_{j}}+1 \right)}+\frac{1}{C}-1$.

Intuitively, when edge $(i,j)$ is inserted and $d_j>0$,
Part~1 of ${[\mathbf{M}]}_{a,b}$ tallies the weighted sum of the following new paths for node-pair $(a,b)$: 
\begin{equation} \label{eq:39a}
\scalebox{0.74}{$
\underbrace{\overbrace{a\leftarrow \circ \cdots \circ \leftarrow j}^{{{[{{{\mathbf{\tilde{Q}}}}^{k}}]}_{a,j}}}}_{\text{length }k} \Leftarrow \underbrace{\overbrace{i\leftarrow \circ \cdots \circ \leftarrow \bullet \to \circ \cdots \circ \to \star}^{{{[\mathbf{S}]}_{i,\star}}}}_{\mathclap{\text{all symmetric in-link paths for node-pair }(i,\star)}}\overbrace{\to }^{{{\mathbf{Q}}^{T}}}\underbrace{\overbrace{\blacktriangle \to \cdots \circ \to b}^{{{[{{({{{\mathbf{\tilde{Q}}}}^{T}})}^{k}}]}_{\blacktriangle,b}}}}_{\text{length }k}
$}
\end{equation}

Such paths are the concatenation of four types of sub-paths (as depicted above)
associated with four matrices, respectively, ${{[{{{\mathbf{\tilde{Q}}}}^{k}}]}_{a,j}}, {{[\mathbf{S}]}_{i,\star}}, {{\mathbf{Q}}^{T}},{{[{{({{{\mathbf{\tilde{Q}}}}^{T}})}^{k}}]}_{\blacktriangle,b}} $, plus the inserted edge $j \Leftarrow i$.
When such entire concatenated paths exist in the new graph,
they should be accommodated for assessing the new SimRank ${[\tilde{\mathbf{S}}]}_{a,b}$ in response to the edge insertion $(i,j)$
because our reinterpretation of SimRank indicates that SimRank counts \emph{all} the symmetric in-link paths,
and the entire concatenated paths can prove to be symmetric in-link paths.

Likewise,  Parts 2 and 3 of ${[\mathbf{M}]}_{a,b}$, respectively,
tally the weighted sum of the following paths for pair $(a,b)$:
\begin{equation} \label{eq:39b}
\scalebox{0.9}{$
\underbrace{\overbrace{a\leftarrow \circ \cdots \circ \leftarrow}^{{{[{{{\mathbf{\tilde{Q}}}}^{k}}]}_{a,j}}}}_{\text{length }k} j  \underbrace{\overbrace{\leftarrow \circ \cdots \circ \leftarrow \bullet \to \circ \cdots \circ \to }^{{{[\mathbf{S}]}_{j,\star}}}}_{\mathclap{\text{all symmetric in-link paths for }(j,\star)}} \star \underbrace{\overbrace{\to \cdots \circ \to b}^{{{[{{({{{\mathbf{\tilde{Q}}}}^{T}})}^{k}}]}_{\star,b}}}}_{\text{length }k}
$}
\end{equation}
\begin{equation} \label{eq:39c}
\scalebox{0.95}{$
\underbrace{\overbrace{a\leftarrow \circ \cdots \circ \leftarrow}^{{{[{{{\mathbf{\tilde{Q}}}}^{k}}]}_{a,j}}}}_{\text{length }k} j \underbrace{\overbrace{ \to \circ \cdots \circ \to b}^{{{[{{({{{\mathbf{\tilde{Q}}}}^{T}})}^{k}}]}_{j,b}}}}_{\text{length }k}
$}
\end{equation}

Indeed, when edge $(i,j)$ is inserted,
only these three kinds of paths have extra contributions for $\mathbf{M}$ (therefore for $\mathbf{\Delta S}$).
As incremental updates in SimRank merely tally these paths,
node-pairs without having such paths could be safely pruned.
In other words,
for those pruned node-pairs,
the three kinds of paths will have ``zero contributions'' to the changes in $\mathbf{M}$ in response to edge insertion.
Thus, after pruning, the remaining node-pairs in $G$ constitute the ``affected areas'' of $\mathbf{M}$.

We next identify ``affected areas'' of $\mathbf{M}$,
by pruning redundant node-pairs in $G$,
based on the following. 
\begin{theorem} \label{thm:04}
For the edge $(i,j)$ insertion,
let $\mathsf{\mathcal{O}}(a)$ and $\mathsf{\tilde{\mathcal{O}}}(a)$ be the out-neighbors of node $a$ in old $G$ and new $G\cup \{(i,j)\}$, respectively.
Let $\mathbf{M}_k$ be the $k$-th iterative matrix in Eq.\eqref{eq:16a}, and let
\begin{eqnarray}
 {{\mathsf{\mathcal{F}}}_{1}}&:=&\{ b \ | \ b\in \mathsf{\mathcal{O}}(y),\ \exists y,\ s.t.\ {{[\mathbf{S}]}_{i,y}}\ne 0\} \label{eq:39} \\
{{\mathsf{\mathcal{F}}}_{2}}&:=&\left\{\begin{array}{lc}
   \varnothing  & \ \ ({{d}_{j}}=0)  \\
   \{  y  \ | \ {{[\mathbf{S}]}_{j,y}}\ne 0\} & \ \  ({{d}_{j}}>0)  \\
\end{array} \right.  \label{eq:40}
\end{eqnarray}
\begin{align}
\label{eq:41} & {{\mathsf{\mathcal{A}}}_{k}}\times {{\mathsf{\mathcal{B}}}_{k}}:= \\
\nonumber & \scalebox{0.82}{$\left\{ \begin{array}{lc}
   \{j\}\times \left( {{\mathsf{\mathcal{F}}}_{1}}\cup {{\mathsf{\mathcal{F}}}_{2}}\cup \{j\} \right) & (k=0)  \\
   \{(a,\left. b) \right|a\in \mathsf{\tilde{\mathcal{O}}}(x),\ b\in \mathsf{\tilde{\mathcal{O}}}(y),\ \exists x,\ \exists y,\ s.t.\ {{[{{\mathbf{M}}_{k-1}}]}_{x,y}}\ne 0\} & (k>0)  \\
\end{array} \right.  $}
\end{align}
%

Then, for every iteration $k=0,1,\cdots$,
the matrix ${{\mathbf{M}}_{k}}$ has the following sparse property:
\[
{{[{{\mathbf{M}}_{k}}]}_{a,b}}=0 \quad \textrm{for all } (a,b)\notin ({{\mathsf{\mathcal{A}}}_{k}}\times {{\mathsf{\mathcal{B}}}_{k}}) \cup ({{\mathsf{\mathcal{A}}}_{0}}\times {{\mathsf{\mathcal{B}}}_{0}}).
\]

For the edge $(i,j)$ deletion case,
all the above results hold except that, in Eq.\eqref{eq:40},
the conditions $d_j=0$ and $d_j>0$ are, respectively, replaced by $d_j=1$ and $d_j>1$. \qed
\end{theorem}

(Please refer to Appendix~\ref{app:02e} for the proof and intuition of Theorem~\ref{thm:04}, and Appendix~\ref{app:03d} for an example.)

Theorem~\ref{thm:04} provides a pruning strategy to iteratively eliminate node-pairs with a-priori zero values in $\mathbf{M}_k$ (thus in $\mathbf{\Delta S}$).
Hence, by Theorem~\ref{thm:04},
when edge $(i,j)$ is updated,
we just need to consider node-pairs in $({{\mathsf{\mathcal{A}}}_{k}}\times {{\mathsf{\mathcal{B}}}_{k}}) \cup ({{\mathsf{\mathcal{A}}}_{0}}\times {{\mathsf{\mathcal{B}}}_{0}})$ for incrementally updating $\mathbf{\Delta S}$.

%
%
%
\subsection{{\IncSR} Algorithm with Pruning}
Based on Theorem~\ref{thm:04}, we provide a complete incremental algorithm, referred to as \IncSR, by incorporating our pruning strategy into \IncUSR.
The total time of \IncSR~is $O(K(m+|\AFF|))$ for $K$ iterations,
where $|\AFF|:= \textrm{avg}_{k \in[0,K]} ( |{\cal A}_k| \cdot |{\cal B}_k|)$
with ${\cal A}_k, {\cal B}_k$ in Eq.\eqref{eq:41},
being the average size of ``affected areas'' in $\mathbf{M}_k$ for $K$ iterations.

(Please refer to Appendix~\ref{app:04b} for {\IncSR} algorithm description and its complexity analysis.)
\section{Edge Update with node insertions} \label{sec:06}
%
In this section, we focus on the edge update that accompanies new node insertions.
Specifically, given a new edge $(i,j)$ to be inserted into the old graph $G=(V,E)$, we consider the following cases when
\begin{center}
   (C1) $i \in V$ and $j \notin V$; \qquad (in Subsection~\ref{sec:04c}) \\
   (C2) $i \notin V$ and $j \in V$; \qquad (in Subsection~\ref{sec:04d})  \\
   (C3) $i \notin V$ and $j \notin V$.  \qquad (in Subsection~\ref{sec:04e})  \\
\end{center}

For each case, we devise an efficient incremental algorithm that can support new node insertions and can accurately update only ``affected areas'' of SimRanks.
\begin{remark}
 Let $n=|V|$, without loss of generality, it can be tacitly assumed that \\
a) in case (C1), new node $j \notin V$ is indexed by $(n+1)$; \\
b) in case (C2), new node $i \notin V$ is indexed by $(n+1)$; \\
c) in case (C3), new nodes $i \notin V$ and $j \notin V$ are indexed by $(n+1)$ and $(n+2)$, respectively.
\end{remark}

%
%
%
\subsection{Inserting an edge $(i,j)$ with $i \in V$ and $j \notin V$} \label{sec:04c}
%
In this case, the inserted new edge $(i,j)$ accompanies the insertion of a new node $j$.
Thus, the size of the new SimRank matrix $\mathbf{\tilde{S}}$ is different from that of the old $\mathbf{{S}} $.
As a result, we cannot simply evaluate the changes to $\mathbf{{S}} $ by adopting $\mathbf{\tilde{S}} -\mathbf{S}$ as we did in Section~\ref{sec:04}.

To resolve this problem, we introduce the block matrix representation of new matrices for edge insertion.
Firstly, when a new edge $(i,j)_{i \in V, j \notin V}$ is inserted to $G$,
the new transition matrix $\mathbf{\tilde{Q}}$ can be described as
\begin{equation} \label{eq:61}
\renewcommand\arraystretch{1.2}
\mathbf{\tilde{Q}}=\left[ \begin{array}{c|c}
   \mathbf{Q} & \mathbf{0}  \\  \hline
   {{\mathbf{e}}_{i}^{T}} & 0  \\
\end{array} \right] \begin{array}{l}
   \  \} \ n \textrm{ rows}  \\
   \rightarrow \textrm{row }j  \\
\end{array} \in \mathbb{R}^{(n+1)\times (n+1)}
\end{equation}
Intuitively, $\mathbf{\tilde{Q}}$ is formed by bordering the old $\mathbf{{Q}}$ by 0s except $[\mathbf{\tilde{Q}}]_{j,i}=1$.
Utilizing this block structure of $\mathbf{\tilde{Q}}$,
we can obtain the new SimRank matrix, which exhibits a similar block structure, as shown below:
\begin{theorem} \label{thm:08}
Given an old digraph $G=(V,E)$,
if there is a new edge $(i,j)$ with $i \in V$ and $j \notin V$ to be inserted,
then the new SimRank matrix becomes
\begin{equation} \label{eq:62}
\scalebox{0.82}{$
\renewcommand\arraystretch{1.2}
\mathbf{\tilde{S}}=\left[ \begin{array}{c|c}
   \mathbf{S} & \mathbf{y}  \\ \hline
   {{\mathbf{y}}^{T}} & C{{[\mathbf{S}]}_{i,i}}+(1-C)  \\
\end{array} \right] \begin{array}{l}
   \  \} \ n \textrm{ rows}  \\
   \rightarrow \textrm{row }j  \\
\end{array} \ \textrm{ with } \ \mathbf{y}=C\mathbf{Q}{{[\mathbf{S}]}_{\star,i}}
$}
\end{equation}
where $\mathbf{S} \in \mathbb{R}^{n \times n}$ is the old SimRank matrix of $G$. \qed
\end{theorem}
\begin{proof}
We substitute the new $\mathbf{\tilde{Q}}$ in Eq.\eqref{eq:61} back into the SimRank equation $\mathbf{\tilde{S}}=C \cdot \mathbf{\tilde{Q}} \cdot \mathbf{\tilde{S}}  \cdot {{\mathbf{\tilde{Q}}}^{T}}+(1-C) \cdot {{\mathbf{I}}_{n+1}}$:
\renewcommand\arraystretch{1.2}
\begin{eqnarray*}
\mathbf{S}
:= \left[ \begin{array}{c|c}
   {{{\mathbf{\tilde{S}}}}_{\mathbf{11}}} & {{{\mathbf{\tilde{S}}}}_{\mathbf{12}}}  \\ \hline
   {{{\mathbf{\tilde{S}}}}_{\mathbf{21}}} & {{{\mathbf{\tilde{S}}}}_{\mathbf{22}}}  \\
\end{array} \right]
= && C\left[ \begin{array}{c|c}
   \mathbf{Q} & \mathbf{0}  \\ \hline
   {{\mathbf{e}}_{i}^{T}} & 0  \\
\end{array} \right]\left[ \begin{array}{c|c}
   {{{\mathbf{\tilde{S}}}}_{\mathbf{11}}} & {{{\mathbf{\tilde{S}}}}_{\mathbf{12}}}  \\ \hline
   {{{\mathbf{\tilde{S}}}}_{\mathbf{21}}} & {{{\mathbf{\tilde{S}}}}_{\mathbf{22}}}  \\
\end{array} \right]\left[ \begin{array}{c|c}
   {{\mathbf{Q}}^{T}} & {{\mathbf{e}}_{i}}  \\ \hline
   \mathbf{0} & 0  \\
\end{array} \right] \\
&& +(1-C)\left[ \begin{array}{c|c}
   {{\mathbf{I}}_{n}} & \mathbf{0}  \\ \hline
   \mathbf{0} & 1  \\
\end{array} \right]
\end{eqnarray*}
By expanding the right-hand side, we can obtain
\[
\scalebox{0.85}{$
\left[ \begin{array}{c|c}
   {{{\mathbf{\tilde{S}}}}_{\mathbf{11}}} & {{{\mathbf{\tilde{S}}}}_{\mathbf{12}}}  \\ \hline
   {{{\mathbf{\tilde{S}}}}_{\mathbf{21}}} & {{{\mathbf{\tilde{S}}}}_{\mathbf{22}}}  \\
\end{array} \right]=\left[ \begin{array}{c|c}
   C\mathbf{Q}{{{\mathbf{\tilde{S}}}}_{\mathbf{11}}}{{\mathbf{Q}}^{T}}+(1-C){{\mathbf{I}}_{n}} & C\mathbf{Q}{{{\mathbf{\tilde{S}}}}_{\mathbf{11}}}{{\mathbf{e}}_{i}}  \\ \hline
   C{{\mathbf{e}}_{i}}^{T}{{{\mathbf{\tilde{S}}}}_{\mathbf{11}}}{{\mathbf{Q}}^{T}} & C{{\mathbf{e}}_{i}}^{T}{{{\mathbf{\tilde{S}}}}_{\mathbf{11}}}{{\mathbf{e}}_{i}}+(1-C)  \\
\end{array} \right]$}
\]
The above block matrix equation implies that
\[
{{\mathbf{\tilde{S}}}_{\mathbf{11}}}=C\mathbf{Q}{{\mathbf{\tilde{S}}}_{\mathbf{11}}}{{\mathbf{Q}}^{T}}+(1-C){{\mathbf{I}}_{n}}
\]
Due to the uniqueness of $\mathbf{S}$ in Eq.\eqref{eq:03a}, it follows that
\[{{\mathbf{\tilde{S}}}_{\mathbf{11}}}=\mathbf{S}\]
Thus, we have
\[
\begin{split}
{{\mathbf{\tilde{S}}}_{\mathbf{12}}}&={{\mathbf{\tilde{S}}}_{\mathbf{21}}}^{T}=C\mathbf{Q}{{\mathbf{\tilde{S}}}_{\mathbf{11}}}{{\mathbf{e}}_{i}}=C\mathbf{Q}{{[\mathbf{S}]}_{\star,i}} \\
{{\mathbf{\tilde{S}}}_{\mathbf{22}}}&=C{{\mathbf{e}}_{i}}^{T}{{\mathbf{\tilde{S}}}_{\mathbf{11}}}{{\mathbf{e}}_{i}}+(1-C)=C{{[\mathbf{S}]}_{i,i}}+(1-C)
\end{split}
\]
Combining all blocks of ${{\mathbf{\tilde{S}}}}$ together yields Eq.\eqref{eq:62}. \qed
\end{proof}

Theorem~\ref{thm:08} provides an efficient incremental way of computing the new SimRank matrix $\mathbf{\tilde{S}}$ for unit insertion of the case (C1).
Precisely, the new $\mathbf{\tilde{S}}$ is formed by bordering the old $\mathbf{{S}}$ by the auxiliary vector $\mathbf{y}$.
To obtain $\mathbf{y}$ (and thereby $\mathbf{\tilde{S}}$),
we just need use the $i$-th column of $\mathbf{{S}}$ with one matrix-vector multiplication $(\mathbf{Q}{{[\mathbf{S}]}_{\star,i}})$.
Thus, the total cost of computing new $\mathbf{\tilde{S}}$ requires $O(m)$ time,
as illustrated in Algorithm~\ref{alg:03}.
\begin{figure}[t] \centering
  \includegraphics[width=1\linewidth]{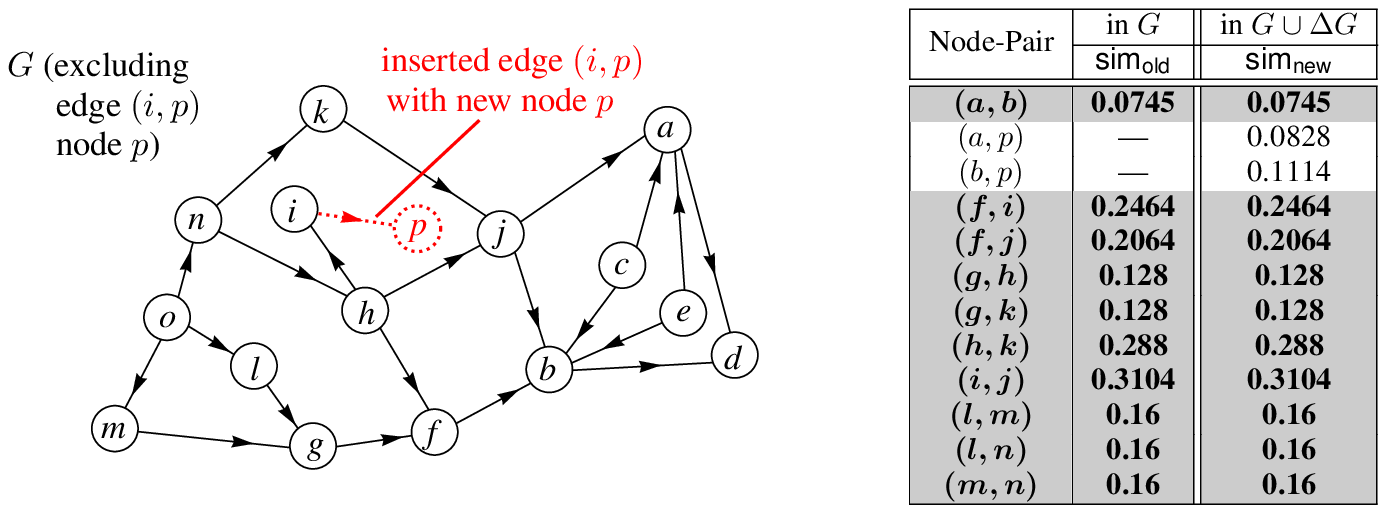}
  \caption{Incrementally updating SimRank when an edge $(i,p)$ with $i \in V$ and $p \notin V$ is inserted into $G=(V,E)$} \label{fig:02} 
\end{figure}
\begin{example}
Consider the citation digraph $G$ in Fig.~\ref{fig:02}.
If the new edge $(i,p)$ with new node $p$ is inserted to $G$, the new $\mathbf{\tilde{S}}$ can be updated from the old $\mathbf{{S}}$ as follows:

According to Theorem~\ref{thm:08},
since $C=0.8$ and  \\[-10pt]
\[\scalebox{0.86}{$
{{[\mathbf{S}]}_{\star,i}} = \kbordermatrix{
 \hspace*{-0.5em} & (a) & \cdots  & (e)  &   (f)   &  (g) & (h) & (i)     & (j)     & (k) & \cdots  & (o) \\
 \hspace*{-0.5em} & 0,  & \cdots, & 0,   & 0.2464, &   0, & 0,  & 0.5904, & 0.3104, & 0,  & \cdots, &  0  \\
}{}^T $} 
\]
it follows that
\begin{equation*}
\renewcommand\arraystretch{1.2}
\mathbf{\tilde{S}}=\left[ \begin{array}{c|c}
   \mathbf{S} & \mathbf{y}  \\ \hline
   {{\mathbf{y}}^{T}} & z  \\
\end{array} \right]  \ \textrm{ with } z=0.8{{[\mathbf{S}]}_{i,i}}+(1-0.8) = 0.6723
\end{equation*} \\[-20pt]
\[\scalebox{0.9}{$
 \quad \mathbf{y} = 0.8 \mathbf{Q} {{[\mathbf{S}]}_{\star,i}}
= \kbordermatrix{
 \hspace*{-0.5em} &  (a)   &  (b)     & (c) & \cdots  & (o)  \\
 \hspace*{-0.5em} & 0.0828, & 0.1114, & 0,  & \cdots, & 0    \\
}{}^T \in \mathbb{R}^{15 \times 1}$} \ \ \qed
\]
\end{example}
\begin{algorithm}[t]
\small
\DontPrintSemicolon
\SetKwInOut{Input}{Input}
\SetKwInOut{Output}{Output}
\Input{a directed graph $G=(V,E)$, \\
       a new edge $(i,j)_{i \in V, \ j \notin V}$ inserted to $G$, \\
       the old similarities $\mathbf{S}$ in $G$, \\
       the damping factor $C$.}
\Output{the new similarities ${\mathbf{\tilde{S}}}$ in $G \cup\{(i,j)\}$.}
\nl \label{ln:a03-01} initialize the transition matrix $\mathbf{Q}$ in $G$ ; \;
\nl \label{ln:a03-02} compute $\mathbf{y} := C \cdot \mathbf{Q}\cdot {{[\mathbf{S}]}_{\star,i}}$ ; \;
\nl \label{ln:a03-03} compute $z:= C \cdot {{[\mathbf{S}]}_{i,i}}+(1-C)$ ; \;
\nl \label{ln:a03-04} \Return $\tilde{\mathbf{S}} := \left[ \renewcommand\arraystretch{1.2} \begin{array}{c|c}
   \mathbf{S} & \mathbf{y}  \\ \hline
   {{\mathbf{y}}^{T}} & z  \\
\end{array} \right]$ ; \;
\caption{\IncUSRtwo~($G, (i,j), \mathbf{S}, C$)}  \label{alg:03}
\end{algorithm}
\subsection{Inserting an edge $(i,j)$ with $i \notin V$ and $j \in V$} \label{sec:04d}
We now focus on the case (C2), the insertion of an edge $(i,j)$ with $i \notin V$ and $j \in V$.
Similar to the case (C1), the new edge accompanies the insertion of a new node $i$.
Hence, $\mathbf{\tilde{S}} -\mathbf{S}$ makes no sense. 

However, in this case, the dynamic computation of SimRank is far more complicated than that of the case (C1),
in that such an edge insertion not only increases the dimension of the old transition matrix $\mathbf{{Q}}$ by one,
but also changes several original elements of $\mathbf{{Q}}$, which may recursively influence SimRank similarities.
Specifically, the following theorem shows, in the case (C2), how $\mathbf{{Q}}$ changes with the insertion of an edge $(i,j)_{i \notin V, j \in V}$.
\begin{theorem} \label{thm:05}
Given an old digraph $G=(V,E)$,
if there is a new edge $(i,j)$ with $i \notin V$ and $j \in V$ to be added to $G$,
then the new transition matrix can be expressed as
\begin{equation} \label{eq:50}
\scalebox{0.8}{$
\mathbf{\tilde{Q}}=\left[ \begin{array}{c|c}
   {\mathbf{\hat{Q}}} & \tfrac{1}{{{d}_{j}}+1}{{\mathbf{e}}_{j}}  \\ \hline
   \mathbf{0} & 0  \\
\end{array} \right] \begin{array}{l}
   \  \} \ n \textrm{ rows}  \\
   \rightarrow \textrm{row }i  \\
\end{array} 
\ \textrm{  with } \mathbf{\hat{Q}}:=\mathbf{Q}-\tfrac{1}{{{d}_{j}}+1}{{\mathbf{e}}_{j}} {{[\mathbf{Q}]}_{j,\star}}
$}
\end{equation}
where
$\mathbf{Q}$ is the old transition matrix of $G$. \qed
\end{theorem}
\begin{proof}
When edge $(i,j)$ with $i \notin V$ and $j \in V$ is added,
there will be two changes to the old $\mathbf{Q}$:

\noindent (i) All nonzeros in ${{[\mathbf{Q}]}_{j,\star}}$ are updated from $\tfrac{1}{d_j}$ to $\tfrac{1}{d_j+1}$:
\begin{equation} \label{eq:51}
{{[\mathbf{\hat{Q}}]}_{j,\star}}
= \tfrac{{{d}_{j}}}{{{d}_{j}}+1}  {{[\mathbf{Q}]}_{j,\star}}
=  {{[\mathbf{Q}]}_{j,\star}} - \tfrac{1}{{{d}_{j}}+1}  {{[\mathbf{Q}]}_{j,\star}}
\end{equation}
(ii) The size of the old $\mathbf{{Q}}$ is added by 1,
with new entry ${{[\tilde{\mathbf{Q}}]}_{j,i}} = \tfrac{1}{d_j+1}$ in the bordered areas and 0s elsewhere:
\begin{equation} \label{eq:52}
\mathbf{\tilde{Q}}=\left[ \begin{array}{c|c}
   {\mathbf{\hat{Q}}} & \tfrac{1}{{{d}_{j}}+1}{{\mathbf{e}}_{j}}  \\ \hline
   \mathbf{0} & 0  \\
\end{array} \right] 
\end{equation}
Combining Eqs.\eqref{eq:51} and \eqref{eq:52} yields \eqref{eq:50}. \qed
\end{proof}

Theorem~\ref{thm:05} exhibits a special structure of the new ${\mathbf{\tilde{Q}}}$:
it is formed by bordering ${\mathbf{\hat{Q}}}$ by 0s except $[{\mathbf{\tilde{Q}}}]_{j,i}=\tfrac{1}{d_j+1}$,
where ${\mathbf{\hat{Q}}}$ is a rank-one update of the old ${\mathbf{{Q}}}$.
The block structure of ${\mathbf{\tilde{Q}}}$ inspires us to partition the new SimRank matrix ${\mathbf{\tilde{S}}}$ conformably into the similar block structure:
\[
\renewcommand\arraystretch{1.2}
\mathbf{\tilde{S}}= \left[ \begin{array}{c|c}
   {{{\mathbf{\tilde{S}}}}_{\mathbf{11}}} & {{{\mathbf{\tilde{S}}}}_{\mathbf{12}}}  \\ \hline
   {{{\mathbf{\tilde{S}}}}_{\mathbf{21}}} & {{{\mathbf{\tilde{S}}}}_{\mathbf{22}}}  \\
\end{array} \right]
\ \ \textrm{ where }  \
 \begin{array}{ll}
{{{\mathbf{\tilde{S}}}}_{\mathbf{11}}} \in \mathbb{R}^{n \times n}, &
{{{\mathbf{\tilde{S}}}}_{\mathbf{12}}} \in \mathbb{R}^{n \times 1}, \\[3pt]
{{{\mathbf{\tilde{S}}}}_{\mathbf{21}}} \in \mathbb{R}^{1 \times n}, &
{{{\mathbf{\tilde{S}}}}_{\mathbf{22}}} \in \mathbb{R}. 
\end{array}
\]
To determine each block of $\mathbf{\tilde{S}}$ with respect to the old $\mathbf{S}$,
we next present the following theorem.
\begin{theorem} \label{thm:06}
If there is a new edge $(i,j)$ with $i \notin V$ and $j \in V$ to be added to the old digraph $G=(V,E)$,
then there exists a vector
\begin{equation} \label{eq:53}
  \mathbf{z}= \alpha {{\mathbf{e}}_{j}}-\mathbf{y} \textrm{ with }
  \mathbf{y}:=\mathbf{Q}  \mathbf{S}  {{[\mathbf{Q}]}_{j,\star}^{T}}   \textrm{ and }
  \alpha : =\tfrac{{{\mathbf{y}}_{j}}+1-C}{2\left( {{d}_{j}}+1 \right)}
\end{equation}
such that the new SimRank matrix $\mathbf{\tilde{S}}$ is expressible as
\begin{equation} \label{eq:54}
\mathbf{\tilde{S}}=\left[ \begin{array}{c|c}
   \mathbf{S}+\mathbf{\Delta }{{{\mathbf{\tilde{S}}}}_{\mathbf{11}}} & \mathbf{0}  \\  \hline
   \mathbf{0} & 1-C  \\
\end{array}  \right] \begin{array}{l}
   \  \} \ n \textrm{ rows}  \\
   \rightarrow \textrm{row }i  \\
\end{array}
\end{equation}
where $\mathbf{S}$ is the old SimRank of $G$, and $\mathbf{\Delta }{{{\mathbf{\tilde{S}}}}_{\mathbf{11}}}$ satisfies the rank-two Sylvester equation:
\begin{equation} \label{eq:55}
\mathbf{\Delta }{{\mathbf{\tilde{S}}}_{\mathbf{11}}}=C\mathbf{\hat{Q}\Delta }{{\mathbf{\tilde{S}}}_{\mathbf{11}}}{{\mathbf{\hat{Q}}}^{T}}+\tfrac{C}{{{d}_{j}}+1}\left( {{\mathbf{e}}_{j}}{{\mathbf{z}}^{T}}+\mathbf{z}{{\mathbf{e}}_{j}}^{T} \right)
\end{equation}
with $\mathbf{\hat{Q}}$ being defined by Theorem~\ref{thm:05}. \qed
\end{theorem}
\begin{proof}
We plug $\mathbf{\tilde{Q}}$ of Eq.\eqref{eq:50} into the SimRank formula:
\[\mathbf{\tilde{S}}=C\cdot \mathbf{\tilde{Q}}\cdot \mathbf{\tilde{S}}\cdot {{\mathbf{\tilde{Q}}}^{T}}+(1-C) \cdot {{\mathbf{I}}_{n+1}}, \]
which produces
\begin{equation*}
\small \def\arraystretch{1.2}
\begin{split}
\mathbf{\tilde{S}}= \left[ \begin{array}{c|c}
   {{{\mathbf{\tilde{S}}}}_{\mathbf{11}}} & {{{\mathbf{\tilde{S}}}}_{\mathbf{12}}}  \\ \hline
   {{{\mathbf{\tilde{S}}}}_{\mathbf{21}}} & {{{\mathbf{\tilde{S}}}}_{\mathbf{22}}}  \\
\end{array} \right]= & C\left[ \begin{array}{c|c}
   {\mathbf{\hat{Q}}} & \tfrac{1}{{{d}_{j}}+1}{{\mathbf{e}}_{j}}  \\ \hline
   \mathbf{0} & 0  \\
\end{array} \right]\left[ \begin{array}{c|c}
   {{{\mathbf{\tilde{S}}}}_{\mathbf{11}}} & {{{\mathbf{\tilde{S}}}}_{\mathbf{12}}}  \\ \hline
   {{{\mathbf{\tilde{S}}}}_{\mathbf{21}}} & {{{\mathbf{\tilde{S}}}}_{\mathbf{22}}}  \\
\end{array} \right]{{\left[ \begin{array}{c|c}
   {\mathbf{\hat{Q}}^{T}} & \mathbf{0}  \\ \hline
   \tfrac{1}{{{d}_{j}}+1}{{\mathbf{e}}_{j}^{T}} & 0  \\
\end{array} \right]}} \\
& +(1-C)\left[ \begin{array}{c|c}
   {{\mathbf{I}}_{n}} & \mathbf{0}  \\ \hline
   \mathbf{0} & 1  \\
\end{array} \right]
\end{split}
\end{equation*}
By using block matrix multiplications, the above equation can be simplified as
\begin{equation} \label{eq:56}
\left[ \begin{array}{c|c}
   {{{\mathbf{\tilde{S}}}}_{\mathbf{11}}} & {{{\mathbf{\tilde{S}}}}_{\mathbf{12}}}  \\ \hline
   {{{\mathbf{\tilde{S}}}}_{\mathbf{21}}} & {{{\mathbf{\tilde{S}}}}_{\mathbf{22}}}  \\
\end{array} \right]=C\left[ \begin{array}{c|c}
   \mathbf{P} & \mathbf{0}  \\ \hline
   \mathbf{0} & 0  \\
\end{array} \right]+(1-C)\left[ \begin{array}{c|c}
   {{\mathbf{I}}_{n}} & \mathbf{0}  \\ \hline
   \mathbf{0} & 1  \\
\end{array} \right]
\end{equation}
\begin{equation} \label{eq:57}
\begin{split}
 \textrm{with } \mathbf{P}=& \mathbf{\hat{Q}}{{\mathbf{\tilde{S}}}_{\mathbf{11}}}{{\mathbf{\hat{Q}}}^{T}} +\tfrac{1}{{{\left( {{d}_{j}}+1 \right)}^{2}}}{{\mathbf{e}}_{j}}{{\mathbf{\tilde{S}}}_{\mathbf{22}}}{{\mathbf{e}}_{j}}^{T} \\
 & +\tfrac{1}{{{d}_{j}}+1}{{\mathbf{e}}_{j}}{{\mathbf{\tilde{S}}}_{\mathbf{21}}}{{\mathbf{\hat{Q}}}^{T}}+\tfrac{1}{{{d}_{j}}+1}\mathbf{\hat{Q}}{{\mathbf{\tilde{S}}}_{\mathbf{12}}}{{\mathbf{e}}_{j}}^{T}
\end{split}
\end{equation}
Block-wise comparison of both sides of Eq.\eqref{eq:56} yields
\[
\left\{
\begin{array}{l}
   {{{\mathbf{\tilde{S}}}}_{\mathbf{12}}}={{{\mathbf{\tilde{S}}}}_{\mathbf{21}}}=\mathbf{0} \\
  {{{\mathbf{\tilde{S}}}}_{\mathbf{22}}}=1-C \\
  {{{\mathbf{\tilde{S}}}}_{\mathbf{11}}}=C\cdot \mathbf{P}+(1-C)\cdot {{\mathbf{I}}_{n}}
\end{array}
\right.
\]
Combing the above equations with Eq.\eqref{eq:57} produces
\begin{equation} \label{eq:58}
{{\mathbf{\tilde{S}}}_{\mathbf{11}}}=C\mathbf{\hat{Q}}{{\mathbf{\tilde{S}}}_{\mathbf{11}}}{{\mathbf{\hat{Q}}}^{T}}+\tfrac{\left( 1-C \right)C}{{{\left( {{d}_{j}}+1 \right)}^{2}}}{{\mathbf{e}}_{j}}{{\mathbf{e}}_{j}}^{T}+(1-C){{\mathbf{I}}_{n}}
\end{equation}
Applying ${{\mathbf{\tilde{S}}}_{\mathbf{11}}}=\mathbf{S}+\mathbf{\Delta }{{\mathbf{\tilde{S}}}_{\mathbf{11}}}$ and $\mathbf{S}=C \mathbf{Q} \mathbf{S} {{\mathbf{Q}}^{T}}+(1-C) {{\mathbf{I}}_{n}}$ to Eq.\eqref{eq:58} and rearranging the terms, we have
\[\mathbf{\Delta }{{\mathbf{\tilde{S}}}_{\mathbf{11}}}=C\mathbf{\hat{Q}\Delta }{{\mathbf{\tilde{S}}}_{\mathbf{11}}}{{\mathbf{\hat{Q}}}^{T}}+\tfrac{C}{{{d}_{j}}+1}\left( 2\alpha {{\mathbf{e}}_{j}}{{\mathbf{e}}_{j}}^{T}-{{\mathbf{e}}_{j}}{{\mathbf{y}}^{T}}-\mathbf{y}{{\mathbf{e}}_{j}}^{T} \right)\]
with ${\alpha}$ and $\mathbf{y}$ being defined by Eq.\eqref{eq:53}. \qed
\end{proof}

Theorem~\ref{thm:06} implies that, in the case (C2), after a new edge $(i,j)$ is inserted, the new SimRank matrix $\mathbf{\tilde{S}}$ takes an elegant diagonal block structure:
the upper-left block of $\mathbf{\tilde{S}}$ is perturbed by $\mathbf{\Delta \tilde{S}_{11}}$ which is the solution to the rank-two Sylvester equation~\eqref{eq:55};
the lower-right block of $\mathbf{\tilde{S}}$ is a constant $(1-C)$.
This structure of $\mathbf{\tilde{S}}$ suggests that the inserted edge $(i,j)_{i \notin V, j \in V}$ only has a recursive impact on the SimRanks with pairs $(x,y) \in V \times V$,
but with no impacts on pairs $(x,y) \in (V \times \{i\}) \cup (\{i\} \times V)$.
Thus, our incremental way of computing the new $\mathbf{\tilde{S}}$ will focus on the efficiency of obtaining $\mathbf{\Delta \tilde{S}_{11}}$ from Eq.\eqref{eq:55}.
Fortunately, we notice that $\mathbf{\Delta \tilde{S}_{11}}$ satisfies the rank-two Sylvester equation,
whose algebraic structure is similar to that of $\mathbf{\Delta {S}}$ in Eqs.\eqref{eq:13} and \eqref{eq:14} (in Section~\ref{sec:04}).
Hence, our previous techniques to compute $\mathbf{\Delta {S}}$ in Eqs.\eqref{eq:13} and \eqref{eq:14} can be analogously applied to compute $\mathbf{\Delta \tilde{S}_{11}}$ in Eq.\eqref{eq:55},
thus eliminating costly matrix-matrix multiplications, as will be illustrated in Algorithm~\ref{alg:04}.

One disadvantage of Theorem~\ref{thm:06} is that, in order to get the auxiliary vector $\mathbf{z}$ for evaluating $\mathbf{\tilde{S}}$,
one has to memorize the \emph{entire} old matrix $\mathbf{S}$ in Eq.\eqref{eq:53}.
In fact, we can utilize the technique of rearranging the terms of the SimRank Eq.\eqref{eq:03a} to characterize $\mathbf{Q}  \mathbf{S}  {{[\mathbf{Q}]}_{j,\star}^{T}}$ in terms of only one vector $[\mathbf{S}]_{\star,j}$ so as to avoid memoizing the entire $\mathbf{S}$,
as shown below.
\begin{theorem} \label{thm:07}
The auxiliary matrix $\mathbf{\Delta }{{{\mathbf{\tilde{S}}}}_{\mathbf{11}}}$ in Theorem~\ref{thm:06} can be represented as
\begin{equation}  \label{eq:59}
\begin{split}
\mathbf{\Delta }{{\mathbf{\tilde{S}}}_{\mathbf{11}}}=\tfrac{C}{{{d}_{j}}+1}\left( \mathbf{M}+{{\mathbf{M}}^{T}} \right) \textrm{ with } \\
\mathbf{M}=\sum\nolimits_{k=0}^{\infty }{{{C}^{k}}{{{\mathbf{\hat{Q}}}}^{k}} {{\mathbf{e}}_{j}}{{\mathbf{z}}^{T}} {{\left( {{{\mathbf{\hat{Q}}}}^{T}} \right)}^{k}}}
\end{split}
\end{equation}
where $\mathbf{\hat{Q}}$ is defined by Theorem~\ref{thm:05} and
\begin{equation}  \label{eq:60}
\scalebox{0.88}{$\mathbf{z}:=\left( \tfrac{1}{2C\left( {{d}_{j}}+1 \right)}\left( {{[\mathbf{S}]}_{j,j}}-{{(1-C)}^{2}} \right)+\tfrac{1-C}{C} \right){{\mathbf{e}}_{j}}-\tfrac{1}{C}{{[\mathbf{S}]}_{\star,j}}$}
\end{equation}
and $\mathbf{S}$ is the old SimRank matrix of $G$. \qed
\end{theorem}
\begin{proof}
  We multiply the SimRank equation by ${{\mathbf{e}}_{j}}$ to get
\[
{{[\mathbf{S}]}_{\star,j}}=C\cdot \mathbf{QS}{{[\mathbf{Q}]}_{j,\star}^{T}}+(1-C)\cdot {{\mathbf{e}}_{j}}.
\]
Combining this with $\mathbf{y}=\mathbf{QS}{{[\mathbf{Q}]}_{j,\star}^{T}}$ in Eq.\eqref{eq:53} produces
\[\mathbf{y}=\tfrac{1}{C}{{[\mathbf{S}]}_{\star,j}}-\tfrac{1-C}{C}{{\mathbf{e}}_{j}} \ \textrm{ and } \
{\mathbf{y}}_j = \tfrac{1}{C}{{[\mathbf{S}]}_{j,j}}-\tfrac{1-C}{C}.
\]
Plugging these results into Eq.\eqref{eq:53},
we can get Eq.\eqref{eq:60}.

Also, the recursive form of $\mathbf{\Delta }{{\mathbf{\tilde{S}}}_{\mathbf{11}}}$ in Eq.\eqref{eq:55} can be converted into the following series:
\begin{eqnarray*}
  \mathbf{\Delta }{{\mathbf{\tilde{S}}}_{\mathbf{11}}}
&=&\tfrac{C}{{{d}_{j}}+1}\sum\nolimits_{k=0}^{\infty }{{{C}^{k}}{{{\mathbf{\hat{Q}}}}^{k}}\left( {{\mathbf{e}}_{j}}{{\mathbf{z}}^{T}}+\mathbf{z}{{\mathbf{e}}_{j}}^{T} \right){{\left( {{{\mathbf{\hat{Q}}}}^{T}} \right)}^{k}}} \\
&=&\mathbf{M} + \mathbf{M}^T
\end{eqnarray*}
with $\mathbf{M}$ being defined by Eq.\eqref{eq:59}. \qed
\end{proof}

For edge insertion of the case (C2),
Theorem \ref{thm:07} gives an efficient method to compute the update matrix $\mathbf{\Delta }{{{\mathbf{\tilde{S}}}}_{\mathbf{11}}}$.
We note that the form of $\mathbf{\Delta }{{{\mathbf{\tilde{S}}}}_{\mathbf{11}}}$ in Eq.\eqref{eq:59} is similar to that of $\mathbf{\Delta }{{{\mathbf{\tilde{S}}}}}$ in Eq.\eqref{eq:29c}.
Thus, similar to Theorem~\ref{thm:03},
the follow method can be applied to compute $\mathbf{M}$ so as to avoid matrix-matrix multiplications.

In Algorithm~\ref{alg:04}, we present the edge insertion of our method for the case (C2) to incrementally update new SimRank scores.
The total complexity of Algorithm~\ref{alg:04} is $O(Kn^2)$ time and $O(n^2)$ memory in the worst case for retrieving all $n^2$ pairs of scores, which is dominated by Line~\ref{ln:a04-08}.
To reduce its computational time further, the similar pruning techniques in Section~\ref{sec:05} can be applied to Algorithm~\ref{alg:04}.
This can speed up the computational time to $O(K(m+|\AFF|))$,
where $|\AFF|$ is the size of ``affected areas'' in $\mathbf{\Delta S}_{11}$.
\begin{figure}[t] \centering
  \includegraphics[width=1\linewidth]{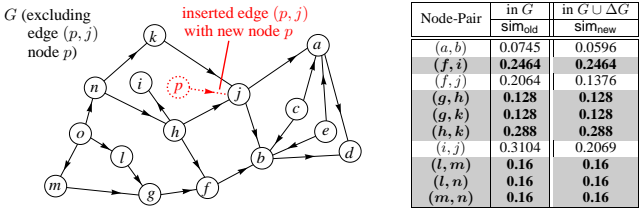}
  \caption{Incrementally update SimRank when a new edge $(p,j)$ with $p \notin V$ and $j \in V$ is inserted into $G=(V,E)$} \label{fig:03} 
\end{figure}

\begin{example}
Consider the citation digraph $G$ in Fig.\ref{fig:03}.
If the new edge $(p,j)$ with new node $p$ is inserted to $G$,
the new $\mathbf{\tilde{S}}$ can be incrementally derived from the old $\mathbf{S}$ as follows:

First, we obtain $\mathbf{\Delta }{{{\mathbf{\tilde{S}}}}_{\mathbf{11}}}$ according to Theorem~\ref{thm:07}.
Note that $C=0.8$, $d_j=2$, and the old SimRank scores \\[-10pt]
\[\scalebox{0.85}{$
{{[\mathbf{S}]}_{\star,j}} = \kbordermatrix{
 \hspace*{-0.5em} & (a) & \cdots  & (e)  &   (f)   &  (g) & (h) & (i)     & (j)     & (k) & \cdots  & (o) \\
 \hspace*{-0.5em} & 0,  & \cdots, & 0,   & 0.2064, &   0, & 0,  & 0.3104, & 0.5104, & 0,  & \cdots, &  0  \\
}{}^T $} 
\]
It follows from Eq.\eqref{eq:60} that the auxiliary vector
\begin{eqnarray*}
\mathbf{z} &= & \scalebox{0.88}{$\left( \tfrac{1}{2 \times 0.8\left( 2+1 \right)}\left( {0.5104}-{{(1-0.8)}^{2}} \right)+\tfrac{1-0.8}{0.8} \right){{\mathbf{e}}_{j}}-\tfrac{1}{0.8}{{[\mathbf{S}]}_{\star,j}}$} \\
&=& \scalebox{0.88}{$\kbordermatrix{
 \hspace*{-0.5em} & (a) & \cdots  & (e)  &   (f)   &  (g) & (h) & (i)     & (j)     & (k) & \cdots  & (o) \\
 \hspace*{-0.5em} & 0,  & \cdots, & 0,   & -0.258, &   0, & 0,  & -0.388, & -0.29,  & 0,  & \cdots, &  0  \\
}{}^T $}
\end{eqnarray*}
Utilizing $\mathbf{z}$, we can obtain $\mathbf{M}$ from Eq.\eqref{eq:59}.
Thus, $\mathbf{\Delta }{{\mathbf{\tilde{S}}}_{\mathbf{11}}}$ can be computed from $\mathbf{M}$ as
\[
\mathbf{\Delta }{{\mathbf{\tilde{S}}}_{\mathbf{11}}}=\tfrac{0.8}{{2}+1}\left( \mathbf{M}+{{\mathbf{M}}^{T}} \right) =
\]
\[
\scalebox{0.52}{$
\begin{blockarray}{ccccc|ccccc|c|c}
       & (a)      & (b)       & (c)& (d)      & (e)       & (f)    &  (g)     & (h)          & (i)     &   (j)  & (k)\cdots(o)  \\[3pt]
 \begin{block}{c(cccc|ccccc|c|c)}
   (a) & -0.0137  & -0.0149   &  0 &  0       &           &        &             &              &         &        &               \\
   (b) & -0.0149  & -0.0146   &  0 &  0       &           &        & \Big{$0$}   &              &         &   \Big{$0$}    &    \Big{$0$}  \\
   (c) & 0        &  0        &  0 &  0       &           &        &             &              &         &        &               \\
   (d) & 0        &  0        &  0 &  -0.0116 &           &        &             &              &         &        &               \\\cline{1-12}
   (e) &          &           &    &          &           &        &             &              &         &   0     &               \\
   (f) &          &           &    &          &           &        &             &              &         & -0.0688 &               \\
   (g) &          & \Big{$0$} &    &          &           &        & \Big{$0$}   &              &         &   0      &   \Big{$0$}   \\
   (h) &          &           &    &          &           &        &             &              &         &   0      &               \\
   (i) &          &           &    &          &           &        &             &              &         & -0.1035 &               \\\cline{1-12}
   (j) &          &     0     &    &          &  0        & -0.0688&    0        &     0        & -0.1035 & -0.1547 &     0          \\\cline{1-12}
\vdots &          & \Big{$0$}&    &          &           &        & \Big{$0$}   &                &         &  \Big{$0$}     &    \Big{$0$}  \\
   (o) &          &          &    &          &           &        &             &                &         &        &               \\
 \end{block}
\end{blockarray}$}\]

Next, by Theorem~\ref{thm:06}, we obtain the new SimRank
\begin{equation*}
\mathbf{\tilde{S}}=\left[ \begin{array}{c|c}
   \mathbf{S}+\mathbf{\Delta }{{{\mathbf{\tilde{S}}}}_{\mathbf{11}}} & \mathbf{0}  \\  \hline
   \mathbf{0} & 0.2  \\
\end{array}  \right]
\end{equation*}
which is partially illustrated in Fig.\ref{fig:03}. \qed
\end{example}
\begin{algorithm}[t]
\small
\DontPrintSemicolon
\SetKwInOut{Input}{Input}
\SetKwInOut{Output}{Output}
\Input{a directed graph $G=(V,E)$, \\
       a new edge $(i,j)_{i \notin V, \ j \in V}$ inserted to $G$, \\
       the old similarities $\mathbf{S}$ in $G$, \\
       the number of iterations $K$, \\
       the damping factor $C$.}
\Output{the new similarities ${\mathbf{\tilde{S}}}$ in $G \cup\{(i,j)\}$.}
\nl \label{ln:a04-01} initialize the transition matrix $\mathbf{Q}$ in $G$ ; \;
\nl \label{ln:a04-02}  ${{d}_{j}}:=$ in-degree of node $j$ in $G$ ; \;
\nl \label{ln:a04-03}  $\mathbf{z}:=\big( \tfrac{1}{2C ( {{d}_{j}}+1 )}\big( {{[\mathbf{S}]}_{j,j}}-{{(1-C)}^{2}} \big)+\tfrac{1-C}{C} \big){{\mathbf{e}}_{j}}-\tfrac{1}{C}{{[\mathbf{S}]}_{\star,j}}$ ;\;
\nl \label{ln:a04-04}  initialize ${{\bm{\xi }}_{0}} := \mathbf{e}_j,\quad {{\bm{\eta }}_{0}} := {\mathbf{z}},\quad {{\mathbf{M}}_{0}} := \mathbf{e}_j  {\mathbf{z}}^T$ ; \;
\nl \label{ln:a04-05}  \For {$k=0,1,\cdots, K-1$} {
\nl \label{ln:a04-06}   ${{\bm{\xi }}_{k+1}} :=  C \cdot \mathbf{Q} \cdot {{\bm{\xi  }}_{k}} - \tfrac{C}{{{d}_{j}}+1}({{[\mathbf{Q}]}_{j,\star}} \cdot \bm{\xi }_{k}) \cdot  {{\mathbf{e}}_{j}}$ ; \;
\nl \label{ln:a04-07}   ${{\bm{\eta }}_{k+1}} := \mathbf{Q}\cdot {{\bm{\eta  }}_{k}} - \tfrac{1}{{{d}_{j}}+1}({{[\mathbf{Q}]}_{j,\star}} \cdot \bm{\eta  }_{k}) \cdot {{\mathbf{e}}_{j}}$ ;\;
\nl \label{ln:a04-08}   ${{\mathbf{M}}_{k+1}} := {{\bm{\xi }}_{k+1}}\cdot \bm{\eta }_{k+1}^{T}+{{\mathbf{M}}_{k}}$ ; \;
}
\nl \label{ln:a04-09} compute $\mathbf{\Delta }{{\mathbf{\tilde{S}}}_{\mathbf{11}}}:=\tfrac{C}{{{d}_{j}}+1}\left( \mathbf{M}_K+{{\mathbf{M}_K^{T}}} \right)$ ; \;
\nl \label{ln:a04-10} \Return $\mathbf{\tilde{S}}:=\left[ \begin{array}{c|c}
   \mathbf{S}+\mathbf{\Delta }{{{\mathbf{\tilde{S}}}}_{\mathbf{11}}} & \mathbf{0}  \\  \hline
   \mathbf{0} & 1-C  \\
\end{array}  \right]$ ; \;
\caption{\IncUSRthree~($G, (i,j), \mathbf{S}, K, C$)}  \label{alg:04}
\end{algorithm}
\subsection{Inserting an edge $(i,j)$ with $i \notin V$ and $j \notin V$} \label{sec:04e}
We next focus on the case (C3), the insertion of an edge $(i,j)$ with $i \notin V$ and $j \notin V$.
Without loss of generality,
it can be tacitly assumed that nodes $i$ and $j$ are indexed by $n+1$ and $n+2$, respectively.
In this case, the inserted edge $(i,j)$ accompanies the insertion of two new nodes,
which can form another independent component in the new graph.

In this case, the new transition matrix $\mathbf{\tilde{Q}}$ can be characterized as a block diagonal matrix
\[
\mathbf{\tilde{Q}}=\left[ \begin{array}{c|c}
   \mathbf{Q} & \mathbf{0}  \\ \hline
   \mathbf{0} & \mathbf{N}  \\
\end{array} \right] \begin{array}{l}
     \} \ n \textrm{ rows}  \\
     \} \ 2 \textrm{ rows}  \\
\end{array} 
\ \textrm{ with } \
\mathbf{N}:=\left[ \begin{array}{cc}
   0 & 0  \\
   1 & 0  \\
\end{array} \right]. 
\]
With this structure, we can infer that the new SimRank matrix $\mathbf{\tilde{S}}$ takes the block diagonal form as
\[
\renewcommand\arraystretch{1.2}
\mathbf{\tilde{S}}=\left[ \begin{array}{c|c}
   \mathbf{S} & \mathbf{0}  \\ \hline
   \mathbf{0} & \mathbf{\hat{S}}  \\
\end{array} \right] \begin{array}{l}
     \} \ n \textrm{ rows}  \\
     \} \ 2 \textrm{ rows}  \\
\end{array} 
 \ \textrm{ with } \
\mathbf{\hat{S}} \in \mathbb{R}^{2 \times 2}. 
\]
This is because, after a new edge $(i,j)_{i \notin V, j \notin V}$ is added,
all node-pairs $(x,y) \in (V \times \{i,j\} \cup \{i,j\} \times V)$ have zero SimRank scores since there are no connections between nodes $x$ and $y$.
Besides, the inserted edge $(i,j)$ is an independent component that has no impact on $s(x,y)$ for $\forall (x,y)\in V \times V$.
Hence, the submatrix $\mathbf{\hat{S}}$ of the new SimRank matrix can be derived by solving the equation:
\[\mathbf{\hat{S}}=C\cdot \mathbf{N}\cdot \mathbf{\hat{S}}\cdot {{\mathbf{N}}^{T}}+(1-C)\cdot {{\mathbf{I}}_{2}}
\quad \Rightarrow \ \mathbf{\hat{S}}=\left[ \begin{matrix}
   1-C & 0  \\
   0 & 1-C^2  \\
\end{matrix} \right]\]
This suggests that, for unit insertion of the case (C3),
the new SimRank matrix becomes
\[
\renewcommand\arraystretch{1.2}
\mathbf{\tilde{S}}=\left[ \begin{array}{c|c}
   \mathbf{S} & \mathbf{0}  \\ \hline
   \mathbf{0} & \mathbf{\hat{S}}  \\
\end{array} \right] \in \mathbb{R}^{(n+2) \times (n+2)} \ \textrm{ with } \
\mathbf{\hat{S}}=\left[ \begin{matrix}
   1-C & 0  \\
   0 & 1-C^2  \\
\end{matrix} \right]. 
\]

Algorithm~\ref{alg:05} presents our incremental method to obtain the new SimRank matrix $\mathbf{\tilde{S}}$ for edge insertion of the case (C3),
which requires just $O(1)$ time.

\begin{algorithm}[t]
\small
\DontPrintSemicolon
\SetKwInOut{Input}{Input}
\SetKwInOut{Output}{Output}
\Input{a directed graph $G=(V,E)$, \\
       a new edge $(i,j)_{i \notin V, \ j \notin V}$ inserted to $G$, \\
       the old similarities $\mathbf{S}$ in $G$, \\
       the damping factor $C$.}
\Output{the new similarities ${\mathbf{\tilde{S}}}$ in $G \cup\{(i,j)\}$.}
\nl \label{ln:a05-01} compute $\mathbf{\hat{S}}:=\left[ \begin{matrix}
   1-C & 0  \\
   0 & 1-C^2  \\
\end{matrix} \right]$ ; \;
\nl \label{ln:a05-02} \Return $\renewcommand\arraystretch{1.2}
\mathbf{\tilde{S}}:=\left[ \begin{array}{c|c}
   \mathbf{S} & \mathbf{0}  \\  \hline
   \mathbf{0} & \mathbf{\hat{S}}  \\
\end{array}  \right]$ ; \;
\caption{\IncUSRfour~($G, (i,j), \mathbf{S}, C$)}  \label{alg:05}
\end{algorithm}
\section{Batch Updates} \label{sec:08}
\begin{figure*}[!t] \centering
  \includegraphics[width=0.8\linewidth]{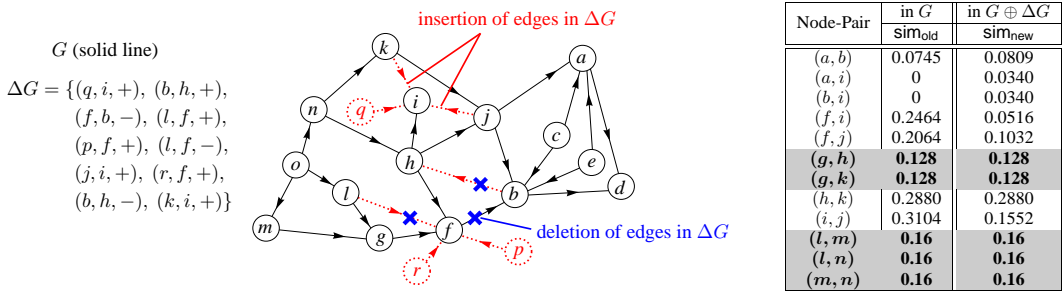}
  \caption{Batch updates for incremental SimRank when a sequence of edges $\Delta G$ are updated to $G=(V,E)$} \label{fig:07} 
\end{figure*}
\begin{table*} \renewcommand\arraystretch{1}
  \centering \small
  \begin{tabular}{c|c|l|p{10cm}}
    \hline
&    when  & new transition matrix $\mathbf{\tilde{Q}}$ & new SimRank matrix $\mathbf{\tilde{S}}$ \\ \hline
\parbox[t]{2mm}{\multirow{4}{*}{\rotatebox[origin=c]{90}{without new node insertions}}}
&    \specialcell[c]{(C0) \\[1pt]
                    insert \\
                    $i_1 \in V$ \\
                    $ \cdots$ \\
                    $i_{\delta} \in V$ \\
                    $j \in V$}
&
    \specialcell[c]{
            $\mathbf{\tilde{Q}} = \mathbf{Q}+\mathbf{u} \cdot \mathbf{v}^T$  with  \\[5pt]
            $\quad \mathbf{u}:=\left\{ \begin{matrix}
               {{\mathbf{e}}_{j}} & \left( {{d}_{j}}=0 \right)  \\
               \tfrac{\delta}{{{d}_{j}}+\delta}{{\mathbf{e}}_{j}} & \left( {{d}_{j}}>0 \right)  \\
            \end{matrix} \right. ,$ \\[15pt]
            $\quad \mathbf{v}:=\left\{ \begin{matrix}
               \tfrac{1}{\delta}{{\mathbf{e}}_{I}} & \left( {{d}_{j}}=0 \right)  \\
               \tfrac{1}{\delta}{{\mathbf{e}}_{I}} - {{[\mathbf{Q}]}_{j,\star}^T} & \left( {{d}_{j}}>0 \right)  \\
            \end{matrix} \right.$
    }
&
    \specialcell[c]{
            $\mathbf{\Delta S}=\mathbf{M}+{{\mathbf{M}}^{T}}$ with \\[5pt]
            $\quad \mathbf{M}:=\sum\nolimits_{k=0}^{\infty }{{{C}^{k+1}} {{{\mathbf{\tilde{Q}}}}^{k}} {{\mathbf{e}}_{j}} {{\bm{\gamma }}^{T}} {{({{{\mathbf{\tilde{Q}}}}^{T}})}^{k}}},$ \\[5pt]
            $\quad \bm{\gamma }:= \left\{ \begin{array}{lc}
               \tfrac{1}{\delta} \mathbf{Q}\cdot {{[\mathbf{S}]}_{\star,I}}+\frac{1}{2 \delta^2}{{[\mathbf{S}]}_{I,I}}\cdot {{\mathbf{e}}_{j}}  & ({{d}_{j}}=0)  \\
               \tfrac{\delta}{({{d}_{j}}+\delta)} \left( \tfrac{1}{\delta} \mathbf{Q}\cdot {{[\mathbf{S}]}_{\star,I}} - \frac{1}{C}\cdot {{[\mathbf{S}]}_{\star,j}}  +( \frac{\lambda \delta }{2\left( {{d}_{j}} + \delta \right)} + \frac{1}{C}-1 )\cdot {{\mathbf{e}}_{j}} \right) & ({{d}_{j}}>0)  \\
            \end{array} \right.$  \\[15pt]
            $\quad \lambda:=\tfrac{1}{\delta^2}{{[\mathbf{S}]}_{I,I}}+\tfrac{1}{C} \cdot {[\mathbf{S}]}_{j,j}-\tfrac{2}{\delta}\cdot {{[\mathbf{Q}]}_{j,\star}}\cdot {{[\mathbf{S}]}_{\star,I}} - \tfrac{1}{C} +1$ \\[5pt]
    }
 \\
    \cline{2-4}
&    \specialcell[c]{(C0) \\[1pt]
                    delete \\
                    $i_1 \in V$ \\
                    $ \cdots$ \\
                    $i_{\delta} \in V$ \\
                    $j \in V$}
&
    \specialcell[c]{
            $\mathbf{\tilde{Q}} = \mathbf{Q}+\mathbf{u} \cdot \mathbf{v}^T$  with  \\[5pt]
            $\quad \mathbf{u}:=\left\{ \begin{matrix}
               {{\mathbf{e}}_{j}} & \left( {{d}_{j}}=1 \right)  \\
               \tfrac{\delta}{{{d}_{j}}-\delta}{{\mathbf{e}}_{j}} & \left( {{d}_{j}}>1 \right)  \\
            \end{matrix} \right. ,$ \\[15pt]
            $\quad \mathbf{v}:=\left\{ \begin{matrix}
               {-\tfrac{1}{\delta} {\mathbf{e}}_{I}} & \left( {{d}_{j}}=1 \right)  \\
               {{[\mathbf{Q}]}_{j,\star}^T}- \tfrac{1}{\delta} {{\mathbf{e}}_{I}} & \left( {{d}_{j}}>1 \right)  \\
            \end{matrix} \right.$
    }
&
    \specialcell[c]{
            $\mathbf{\Delta S}=\mathbf{M}+{{\mathbf{M}}^{T}}$ with \\[5pt]
            $\quad \mathbf{M}:=\sum\nolimits_{k=0}^{\infty }{{{C}^{k+1}} {{{\mathbf{\tilde{Q}}}}^{k}} {{\mathbf{e}}_{j}} {{\bm{\gamma }}^{T}} {{({{{\mathbf{\tilde{Q}}}}^{T}})}^{k}}},$ \\[5pt]
            $\quad \bm{\gamma }:= \left\{ \begin{array}{lc}
                - \tfrac{1}{\delta} \mathbf{Q}\cdot {{[\mathbf{S}]}_{\star,I}}+\frac{1}{2 \delta^2}{{[\mathbf{S}]}_{I,I}}\cdot {{\mathbf{e}}_{j}}  & ({{d}_{j}}=1)  \\
               \tfrac{\delta}{({{d}_{j}}-\delta)} \left(\frac{1}{C}\cdot {{[\mathbf{S}]}_{\star,j}} - \tfrac{1}{\delta}\mathbf{Q}\cdot {{[\mathbf{S}]}_{\star,I}} +( \frac{\lambda \delta }{2\left( {{d}_{j}}-\delta \right)}- \frac{1}{C}+1 )\cdot {{\mathbf{e}}_{j}} \right) & ({{d}_{j}}>1)  \\
            \end{array} \right.$  \\[15pt]
            $\quad \lambda:=\tfrac{1}{\delta^2} {{[\mathbf{S}]}_{I,I}}+\tfrac{1}{C} \cdot {[\mathbf{S}]}_{j,j}-\tfrac{2}{\delta}\cdot {{[\mathbf{Q}]}_{j,\star}}\cdot {{[\mathbf{S}]}_{\star,I}} - \tfrac{1}{C} +1$ \\[5pt]
    }
 \\
    \hline
\parbox[t]{2mm}{\multirow{13}{*}{\rotatebox[origin=c]{90}{with new node insertions}}}
&    \specialcell[c]{(C1) \\[1pt]
                     insert \\
                    $i_1 \in V$ \\
                    $ \cdots$ \\
                    $i_{\delta} \in V$ \\
                    $j \notin V$}
&
    \specialcell[c]{
            $\renewcommand\arraystretch{1.2}
                    \mathbf{\tilde{Q}}=\left[ \begin{array}{c|c}
                       \mathbf{Q} & \mathbf{0}  \\  \hline
                       \tfrac{1}{\delta}{{\mathbf{e}}_{I}^{T}} & 0  \\
                    \end{array} \right] \begin{array}{l}
                       \  \} \ n \textrm{ rows}  \\
                       \rightarrow \textrm{row }j  \\
                    \end{array} $
    }
&
    \specialcell[c]{
            $\renewcommand\arraystretch{1.2}
                \mathbf{\tilde{S}}=\left[ \begin{array}{c|c}
                   \mathbf{S} & \mathbf{y}  \\ \hline
                   {{\mathbf{y}}^{T}} & \tfrac{C}{\delta^2}{{[\mathbf{S}]}_{I,I}}+(1-C)  \\
                \end{array} \right] \begin{array}{l}
                   \  \} \ n \textrm{ rows}  \\
                   \rightarrow \textrm{row }j  \\
                \end{array}$ \quad with \\[15pt]
            $\quad \mathbf{y}:=\tfrac{C}{\delta}\mathbf{Q}{{[\mathbf{S}]}_{\star,I}}$
    }
 \\
    \cline{2-4}
&    \specialcell[c]{(C2) \\[1pt]
                    insert \\
                    $i_1 \notin V$ \\
                    $ \cdots$ \\
                    $i_{\delta} \notin V$ \\
                    $j \in V$}
&
    \specialcell[c]{
            $\renewcommand\arraystretch{1.2}
                                \mathbf{\tilde{Q}}=\left[ \begin{array}{c|c}
               {\mathbf{\hat{Q}}} & \tfrac{1}{{{d}_{j}}+\delta}{{\mathbf{e}}_{j} \mathbf{1}_{\delta}^T}  \\ \hline
               \mathbf{0} & 0  \\
            \end{array} \right] \begin{array}{l}
               \  \} \ n \textrm{ rows}  \\
               \  \} \ \delta \textrm{ rows}  \\
            \end{array}$ \\[15pt]
            $\quad \textrm{with } \mathbf{\hat{Q}}:=\mathbf{Q}-\tfrac{\delta}{{{d}_{j}}+\delta}{{\mathbf{e}}_{j}} {{[\mathbf{Q}]}_{j,\star}}$
    }
&
    \specialcell[c]{
            $\renewcommand\arraystretch{1.2}
                \mathbf{\tilde{S}}=\left[ \begin{array}{c|c}
                   \mathbf{S}+\tfrac{C \delta}{{{d}_{j}}+\delta}\left( \mathbf{M}+{{\mathbf{M}}^{T}} \right) & \mathbf{0}  \\  \hline
                   \mathbf{0} & (1-C) \mathbf{I}_{\delta}  \\
                \end{array}  \right] \begin{array}{l}
               \  \} \ n \textrm{ rows}  \\
               \  \} \ \delta \textrm{ rows}  \\
            \end{array}$ \quad with \\[15pt]
            $\quad \mathbf{M}:=\sum\nolimits_{k=0}^{\infty }{{{C}^{k}}{{{\mathbf{\hat{Q}}}}^{k}} {{\mathbf{e}}_{j}}{{\mathbf{z}}^{T}} {{\left( {{{\mathbf{\hat{Q}}}}^{T}} \right)}^{k}}},$ \\[5pt]
            $\quad \mathbf{z}:=\left( \tfrac{1}{2C\left( {{d}_{j}}+ \delta \right)}\left( \delta {{[\mathbf{S}]}_{j,j}}-{{(\delta-C)(1-C)}} \right)+\tfrac{1-C}{C} \right){{\mathbf{e}}_{j}}-\tfrac{1}{C}{{[\mathbf{S}]}_{\star,j}}$ \\[10pt]
    }
 \\
    \cline{2-4}
&    \specialcell[c]{(C3) \\[1pt]
                    insert \\
                    $i_1 \notin V$ \\
                    $ \cdots$ \\
                    $i_{\delta} \notin V$ \\
                    $j \notin V$}
&
    \specialcell[c]{
            $\mathbf{\tilde{Q}}=\left[ \begin{array}{c|c}
               \mathbf{Q} & \mathbf{0}  \\ \hline
               \mathbf{0} & \mathbf{N}  \\
            \end{array} \right] \begin{array}{l}
                 \} \ n \textrm{ rows}  \\
                 \} \ \delta+1 \textrm{ rows}  \\
            \end{array}$   \\[15pt]
            $\quad \textrm{with } \mathbf{N}:=\left[ \begin{array}{c|c}
               \mathbf{0} & \mathbf{0}  \\ \hline
               \tfrac{1}{\delta}\mathbf{1}_{\delta}^T & 0  \\
            \end{array} \right] \begin{array}{l}
                 \} \ \delta \textrm{ rows}  \\
                 \rightarrow \textrm{row }j  \\
            \end{array}$
    }
&
    \specialcell[c]{
            $\renewcommand\arraystretch{1.2}
            \mathbf{\tilde{S}}=\left[ \begin{array}{c|c}
               \mathbf{S} & \mathbf{0}  \\ \hline
               \mathbf{0} & \mathbf{\hat{S}}  \\
            \end{array} \right]\begin{array}{l}
                 \} \ n \textrm{ rows}  \\
                 \} \ \delta+1 \textrm{ rows}  \\
            \end{array}$  \\[15pt]
            $\quad \textrm{with } \mathbf{\hat{S}}:=\left[ \begin{array}{c|c}
               (1-C)\mathbf{I}_{\delta} & \mathbf{0}  \\ \hline
               \mathbf{0} & (1-C)(1+\tfrac{C}{\delta})  \\
            \end{array} \right]. \begin{array}{l}
                 \} \ \delta \textrm{ rows}  \\
                 \rightarrow \textrm{row }j   \\
            \end{array}$
    } \\  \hline
  \end{tabular}
  \caption{{Batch updates for a sequence of edges $\{(i_1,j), \cdots, (i_{\delta},j)\}$ to the old graph $G=(V,E)$, \\
 where $[\mathbf{S}]_{\star,I} := \sum_{i \in I} [\mathbf{S}]_{\star,i}, \quad [\mathbf{S}]_{I,I} := \sum_{i \in I} [\mathbf{S}]_{i,I}, \quad \mathbf{1}_{\delta} := (1,1,\cdots, 1)^T \in \mathbb{R}^{\delta \times 1}$}}  \label{tab:02}
\end{table*}
In this section, we consider the batch updates problem for incremental SimRank, \ie
given an old graph $G=(V,E)$ and a sequence of edges $\Delta G$ to be updated to $G$, the retrieval of new SimRank scores in $G\oplus \Delta G$.
Here, the set $\Delta G$ can be mixed with insertions and deletions: 
\[
\scalebox{0.95}{$
\Delta G := \{(i_1, j_1, {\op}_1), (i_2, j_2, {\op_2}), \cdots, (i_{|\Delta G|}, j_{|\Delta G|}, {\op_{|\Delta G|}}) \}
$}
\]
where $(i_q, j_q)$ is the $q$-th edge in $\Delta G$ to be inserted into (if $\op_q =$``$+$'') or deleted from (if $\op_q =$``$-$'') $G$.

The straightforward approach to this problem is to update each edge of $\Delta G$ one by one, by running a unit update algorithm for $|\Delta G|$ times.
However, this would produce many unnecessary intermediate results and redundant updates that may cancel out each other.
\begin{example} \label{eg:09}
  Consider the old citation graph $G$ in Fig.~\ref{fig:07}, and a sequence of edge updates $\Delta G$ to $G$:
\[\scalebox{0.95}{$
\begin{split}
\Delta G =\{& (q,i,+), \ \bm{(b,h,+)}, \ (f,b,-), \ \bm{(l,f,+)}, \  (p,f,+), \\
             & \bm{(l,f,-)}, \ (j,i,+), \ (r,f,+), \ \bm{(b,h,-)}, \ (k,i,+)\}
\end{split}$}
\]
We notice that, in $\Delta G$, the edge insertion $(b,h,+)$ can cancel out the edge deletion $(b,h,-)$.
Similarly, $(l,f,+)$ can cancel out $(l,f,-)$.
Thus, after edge cancellation, the \emph{net} update of $\Delta G$, denoted as $\Delta G_{\textrm{net}}$, is
\[
\begin{split}
\Delta G_{\textrm{net}} =\{& (q,i,+),  \ (f,b,-), \ (p,f,+), \\
             &  (j,i,+), \ (r,f,+),  \ (k,i,+)\}  \qquad \qed
\end{split}
\]
\end{example}

Example~\ref{eg:09} suggests that a portion of redundancy in $\Delta G$ arises from the insertion and deletion of the same edge that may cancel out each other.
After cancellation, it is easy to verify that
\[
|\Delta G_{\textrm{net}}| \le |\Delta G| \ \textrm{ yet } \  G \oplus \Delta G_{\textrm{net}} = G \oplus \Delta G.
\]

To obtain $\Delta G_{\textrm{net}}$ from $\Delta G$,
we can readily use hashing techniques to count occurrences of updates in $\Delta G$.
More specifically, we use each edge of $\Delta G$ as a hash key,
and initialize each key with zero count.
Then, we scan each edge of $\Delta G$ once,
and increment (\Resp decrement) its count by one each time an edge insertion (\Resp deletion) appears in $\Delta G$.
After all edges in $\Delta G$ are scanned,
the edges whose counts are nonzeros make a net update $\Delta G_{\textrm{net}}$.
All edges in $\Delta G_{\textrm{net}}$ with $+1$ (\Resp $-1$) counts make a net insertion update $\Delta G_{\textrm{net}}^{+}$ (\Resp a net deletion update $\Delta G_{\textrm{net}}^{-}$).
Clearly, 
$
\Delta G_{\textrm{net}} = \Delta G_{\textrm{net}}^{+} \cup \Delta G_{\textrm{net}}^{-}.
$

Having reduced $\Delta G$ to the net edge updates $\Delta G_{\textrm{net}}$,
we next merge the updates of ``similar sink edges'' in $\Delta G_{\textrm{net}}$ to speedup the batch updates further.

We first introduce the notion of ``similar sink edges''.
\begin{definition}
 Two distinct edges $(a,c)$ and $(b,c)$ are called ``similar sink edges'' \wrt node $c$ if they have a common end node $c$ that both $a$ and $b$ point to. \qed
\end{definition}

``Similar sink edges'' is introduced to partition $\Delta G_{\textrm{net}}$.
To be specific,
we first sort all the edges $\{(i_p,j_p)\}$ of $\Delta G_{\textrm{net}}^{+}$ (\Resp $\Delta G_{\textrm{net}}^{-}$) according to its end node $j_p$.
Then, the ``similar sink edges'' \wrt node $j_p$ form a partition of $\Delta G_{\textrm{net}}^{+}$ (\Resp $\Delta G_{\textrm{net}}^{-}$).
For each block $\{(i_{p_k},j_p)\}$ in $\Delta G_{\textrm{net}}^{+}$, we next split it further into two sub-blocks according to whether its end node $i_{p_k}$ is in the old $V$.
Thus, after partitioning,
each block in $\Delta G_{\textrm{net}}^{+}$ (\Resp $\Delta G_{\textrm{net}}^{-}$), denoted as
$
\{(i_1,j), \ (i_2,j), \ \cdots, \ (i_{\delta},j)\},
$
falls into one of the following cases:
\begin{center}
   (C0) $i_1 \in V, \ i_2 \in V, \ \cdots, i_{\delta} \in V$ and $j \in V$; \\
   (C1) $i_1 \in V, \ i_2 \in V, \ \cdots, i_{\delta} \in V$ and $j \notin V$; \\
   (C2) $i_1 \notin V, \ i_2 \notin V, \ \cdots, i_{\delta} \notin V$ and $j \in V$;  \\
   (C3) $i_1 \notin V, \ i_2 \notin V, \ \cdots, i_{\delta} \notin V$ and $j \notin V$.  \\
\end{center}
\begin{example} \label{eg:10}
Let us recall $\Delta G_{\textrm{net}}$ derived by Example~\ref{eg:09},
in which $\Delta G_{\textrm{net}} = \Delta G_{\textrm{net}}^{+} \cup \Delta G_{\textrm{net}}^{-}$ with
\[ \small
\begin{split}
  \Delta G_{\textrm{net}}^{+} = \{& (q,i,+), \ (p,f,+), \ (j,i,+), \ (r,f,+),  \ (k,i,+)\} \\
  \Delta G_{\textrm{net}}^{-} = \{& (f,b,-)\}.
\end{split}
\]
We first partition $\Delta G_{\textrm{net}}^{+}$ by ``similar sink edges'' into
\[ \scalebox{0.9}{$
\Delta G_{\textrm{net}}^{+} = \{ (q,i,+), \ (j,i,+), \ (k,i,+) \}  \cup  \{ (p,f,+),  \ (r,f,+) \}
$}
\]

In the first block of $\Delta G_{\textrm{net}}^{+}$, since the nodes $q \notin V$, $j \in V$, and $k \in V$,
we will partition this block further into $\{ (q,i,+)\} \cup \{ (j,i,+),  (k,i,+) \}$.
Eventually,
\[ \scalebox{0.85}{$
\Delta G_{\textrm{net}}^{+} = \{ (q,i,+)\}  \cup \{ (j,i,+),  (k,i,+) \} \cup  \{ (p,f,+),   (r,f,+) \}
$} \quad \qed
\]
\end{example}

The main advantage of our partitioning approach is that, after partition, all the edge updates in each block can be processed simultaneously,
instead of one by one.
To elaborate on this, we use case (C0) as an example, \ie
the insertion of $\delta$ edges $\{(i_1,j), \ (i_2,j), \ \cdots, \ (i_{\delta},j)\}$ into $G=(V,E)$ when $i_1 \in V, \cdots, i_{\delta} \in V$, and $j \in V$.
Analogous to Theorem~\ref{thm:01}, one can readily prove that,
after such $\delta$ edges are inserted,
the changes $\mathbf{\Delta Q}$ to the old transition matrix is still a \emph{rank-one} matrix that can be decomposed as
$\mathbf{\tilde{Q}} = \mathbf{Q}+\mathbf{u} \cdot \mathbf{v}^T \  \textrm{ with } $
\[
\begin{split}
& \mathbf{u}:=\left\{ \begin{matrix}
               {{\mathbf{e}}_{j}} & \left( {{d}_{j}}=0 \right)  \\
               \tfrac{\delta}{{{d}_{j}}+\delta}{{\mathbf{e}}_{j}} & \left( {{d}_{j}}>0 \right)  \\
            \end{matrix} \right. ,
            \quad \mathbf{v}:=\left\{ \begin{matrix}
               \tfrac{1}{\delta}{{\mathbf{e}}_{I}} & \left( {{d}_{j}}=0 \right)  \\
               \tfrac{1}{\delta}{{\mathbf{e}}_{I}} - {{[\mathbf{Q}]}_{j,\star}^T} & \left( {{d}_{j}}>0 \right)  \\
            \end{matrix} \right.
\end{split}
\]
where ${\mathbf{e}}_{I}$ is an $n \times 1$ vector with its entry $[{\mathbf{e}}_{I}]_x=1$ if $x \in I\triangleq \{i_1,i_2, \cdots, i_{\delta}\}$, and $[{\mathbf{e}}_{I}]_x=0$ if $x \notin V$.
Since the rank-one structure of $\mathbf{\Delta Q}$ is preserved for updating $\delta$ edges,
Theorem~\ref{thm:02} still holds under the new settings of $\mathbf{u}$ and $\mathbf{v}$ for batch updates.
Therefore, the changes $\mathbf{\Delta S}$ to the SimRank matrix in response to $\delta$ edges insertion can be represented as a similar formulation to Theorem~\ref{thm:03},
as illustrated in the first row of Table~\ref{tab:02}.
Similarly,
we can also extend Theorems~\ref{thm:08}--\ref{thm:07} in Section~\ref{sec:06} to
support batch updates of $\delta$ edges for other cases (C1)--(C3) that accompany new node insertions.
Table~\ref{tab:02} summarizes the new $\mathbf{Q}$ and $\mathbf{S}$ in response to such batch edge updates of all the cases.
When $\delta=1$, these batch update results in Table~\ref{tab:02} can be reduced to the unit update results of Theorems~\ref{thm:01}--\ref{thm:07}.

\begin{algorithm}[t]
\small
\DontPrintSemicolon
\SetKwInOut{Input}{Input}
\SetKwInOut{Output}{Output}
\Input{a directed graph $G=(V,E)$, \\
       a sequence of edge updates $\Delta G=\{(i,j,\op) \}$, \\
       the old similarities $\mathbf{S}$ in $G$, \\
       the damping factor $C$.}
\Output{the new similarities ${\mathbf{\tilde{S}}}$ in $G \oplus \Delta G$.}
\nl \label{ln:a06-01}  obtain the net update $\Delta G_{\textrm{net}}$ from $\Delta G$ via hashing ; \;
\nl \label{ln:a06-02}  split $\Delta G_{\textrm{net}} = \Delta G_{\textrm{net}}^{+} \cup  \Delta G_{\textrm{net}}^{-}$ according to {\op} ; \;
\nl \label{ln:a06-03}  partition $\Delta G_{\textrm{net}}^{+}$ and $\Delta G_{\textrm{net}}^{-}$ by ``similar sink edges'' ; \;
\nl \label{ln:a06-04}  \For {each block of $\Delta G_{\textrm{net}}^{+}$} {
\nl \label{ln:a06-05}  split all edges $\{(i,j)\}$ of each block further into (at most) two sub-blocks based on whether $i \in V$ \ \; }
\nl \label{ln:a06-06}  \For {each block of $\Delta G_{\textrm{net}}^{-}$} {
\nl \label{ln:a06-07}  delete all edges of each block and update ${\mathbf{\tilde{S}}}$ via Table~\ref{tab:02} ;}
\nl \label{ln:a06-08}  remove all singleton nodes in the graph; \;
\nl \label{ln:a06-09}  \For {each sub-block of $\Delta G_{\textrm{net}}^{+}$} {
\nl \label{ln:a06-10}  insert all edges of each sub-block and update ${\mathbf{\tilde{S}}}$  via Table~\ref{tab:02} ;}
\nl \label{ln:a06-11}  \Return ${\mathbf{\tilde{S}}}$ ; \;
\caption{\IncBSR~($G, (i,j), \mathbf{S}, C$)}  \label{alg:06}
\end{algorithm}

Algorithm~\ref{alg:06} presents an efficient batch updates algorithm, \IncBSR, for dynamical SimRank computation.
The actual computational time of {\IncBSR} depends on the input parameter $\Delta G$ since different update types in Table~\ref{tab:02} would result in different computational time.
However, we can readily show that {\IncBSR} is superior to the $|\Delta G|$ executions of the unit update algorithm,
because {\IncBSR} can process the ``similar sink updates'' of each block simultaneously and can cancel out redundant updates.
To clarify this, let us assume that $|\Delta G_{\textrm{net}}|$ can be partitioned into $|B|$ blocks,
with $\delta_t$ denoting the number of edge updates in $t$-th block.
In the worst case, we assume that all edge updates happen to be the most time-consuming case (C0) or (C2). 
Then, the total time for handling $|\Delta G|$ updates is bounded by
\[ \small
\begin{split}
     & O\bigg(\sum\nolimits_{t=1}^{|B|} \big(n\delta_t + \delta_t^2 +  K(nd + \delta_t +|\AFF|) \big)\bigg) \\
 \le & O\bigg( n |\Delta G_{\textrm{net}}| + |\Delta G_{\textrm{net}}| \sum\nolimits_{t=1}^{|B|}\delta_t + K \sum\nolimits_{t=1}^{|B|} (nd+\delta_t+|\AFF|) \bigg) \\
 \le & O\big(  (n +|\Delta G_{\textrm{net}}|) |\Delta G_{\textrm{net}}| +  K(|B|nd+|\Delta G_{\textrm{net}}|+|B||\AFF|) \big) \\
\end{split}
\]

Note that $|B| \le |\Delta G_{\textrm{net}}|$, in general $|B| \ll |\Delta G_{\textrm{net}}|$.
Thus, {\IncBSR} is typically much faster than the $|\Delta G|$ executions of the unit update algorithm that is bounded by $O\big( |\Delta G| K(nd+\Delta G+|\AFF|) \big)$. 
\begin{example}
Recall from Example~\ref{eg:09} that a sequence of edge updates $\Delta G$ to the graph $G=(V,E)$ in Fig.~\ref{fig:07}.
We want to compute new SimRank scores in $G \oplus \Delta G$.

First, we can use hashing method to obtain the net update $\Delta G_{\textrm{net}}$ from $\Delta G$, as shown in Example~\ref{eg:09}.

Next, by Example~\ref{eg:10}, we can partition $\Delta G_{\textrm{net}}$ into
\[ \small
\begin{split}
    \Delta G_{\textrm{net}}^{+} = \{& (q,i,+)\}  \cup \{ (j,i,+),  (k,i,+) \} \cup  \{ (p,f,+),   (r,f,+) \} \\
    \Delta G_{\textrm{net}}^{-} = \{& (f,b,-)\}
\end{split}
\]

Then, for each block, we can apply the formulae in Table~\ref{tab:02} to update all edges simultaneously in a batch fashion.
The results are partially depicted as follows:
 \[\scalebox{.9}{$
\begin{tabular}{c|c|c|c|c|c}
  \hline
Node    & $\textsf{sim}_\textsf{old}$   & \multirow{2}{*}{$(f,b,-)$}         & \multirow{2}{*}{$(q,i,+)$}      & $(j,i,+)$         &       $(p,f,+)$ \\
Pairs   & in $G$                              &                   &                & $(k,i,+)$         &       $(r,f,+)$ \\ \hline
$(a,b)$ & 0.0745   & 0.0809         & 0.0809      & 0.0809         &       0.0809 \\
$(a,i)$ & 0        & 0              & 0           & 0.0340         &       0.0340 \\
$(b,i)$ & 0        & 0              & 0           & 0.0340         &       0.0340 \\
$(f,i)$ & 0.2464   & 0.2464         & 0.1232      & 0.1032         &       0.0516 \\
$(f,j)$ & 0.2064   & 0.2064         & 0.2064      & 0.2064         &       0.1032 \\
$(g,h)$ & 0.128    & 0.128          & 0.128       & 0.128          &       0.128  \\
$(g,k)$ & 0.128    & 0.128          & 0.128       & 0.128          &       0.128  \\
$(h,k)$ & 0.288    & 0.288          & 0.288       & 0.288          &       0.288  \\
$(i,j)$ & 0.3104   & 0.3104         & 0.1552      & 0.1552         &       0.1552 \\
$(l,m)$ & 0.16     & 0.16           & 0.16        & 0.16           &       0.16   \\
$(l,n)$ & 0.16     & 0.16           & 0.16        & 0.16           &       0.16   \\
$(m,n)$ & 0.16     & 0.16           & 0.16        & 0.16           &       0.16   \\
  \hline
\end{tabular}
$}\]
The column `$(q,i,+)$' represents the updated SimRank scores after the edge $(q,i)$ is added to $G \oplus \{(f,b,-)\} $.
The last column is the new SimRanks in $G \oplus \Delta G$. \qed
\end{example}
{
\section{Memory Efficiency} \label{sec:07} 
\begin{table*}[t]
\centering
\scalebox{.85}{$ \renewcommand\arraystretch{1.1}
\begin{tabular}{l|l|l}
\hline
\textbf{Line} & \textbf{Description}                                                                                                                                                                       & \textbf{Required Elements from old $\mathbf{S}$}       \\ \hline
3                & $\mathbf{w} \leftarrow \mathbf{Q}\cdot \mathcolor{red}{{{[\mathbf{S}]}_{\star,i}}}$                                                                                                                                 & $i$-th column of $\mathbf{S}$                    \\
4                & $\lambda \leftarrow \mathcolor{red}{{[\mathbf{S}]}_{i,i}}+\tfrac{1}{C} \cdot \mathcolor{red}{{[\mathbf{S}]}_{j,j}}-2\cdot {[\mathbf{w}]}_{j} - \tfrac{1}{C} +1$                                                            & $(i,i)$- and $(j,j)$-th elements of $\mathbf{S}$ \\
6                & ${\bm \gamma} \leftarrow \mathbf{w} +\frac{1}{2}\mathcolor{red}{{[\mathbf{S}]}_{i,i}}\cdot {{\mathbf{e}}_{j}}$                                                                                                    & $(i,i)$-th element of $\mathbf{S}$               \\
9                & ${\bm\gamma} \leftarrow \tfrac{1}{({{d}_{j}}+1)} \big( \mathbf{w}-\frac{1}{C}\mathcolor{red}{{[\mathbf{S}]}_{\star,j}}+( \frac{\lambda }{2\left( {{d}_{j}}+1 \right)}+ \frac{1}{C}-1 ) {{\mathbf{e}}_{j}} \big)$ & $j$-th column of $\mathbf{S}$ \\
15               & $\mathcolor{red}{\tilde{\mathbf{S}}} \leftarrow \mathcolor{red}{\mathbf{S}} + \mathbf{M}_{K} + \mathbf{M}_{K}^T$ & all elements of old $\mathbf{S}$ and new $\tilde{\mathbf{S}}$ \\ \hline
\end{tabular}$}
\caption{{Lines of \IncUSRone~(in Appendix~\ref{app:04a}) that require to get elements from old $\mathbf{S}$ (highlighted in red color)}} \label{tab:04}
\end{table*}

In previous sections, our main focus was devoted to speeding up the computational time of incremental SimRank.
However, for updating all pairs of SimRank scores,
the memory requirement for Algorithms~\ref{alg:03}--\ref{alg:06} remains at $O(n^2)$ since they need to store all $(n^2)$ pairs of old SimRank $\mathbf{S}$ into memory,
which hinders its scalability on large graphs.
We call Algorithms~\ref{alg:03}--\ref{alg:06} \emph{in-memory algorithms}.
\begin{table}[t]
\centering
\scalebox{.85}{$ \renewcommand\arraystretch{1.1}
\begin{tabular}{l|l|l}
\hline
\textbf{Line} & \textbf{Description}                                                                                                                                                                       & \textbf{Storage of $\mathbf{M}_k$}       \\ \hline
10                & $\mathcolor{red}{{\mathbf{M}}_{0}} \leftarrow C \cdot \mathbf{e}_j \cdot {\bm\gamma}^T$                                                                                                                                 & all elements of $\mathbf{M}_0$                    \\
14                & $\mathcolor{red}{{\mathbf{M}}_{k+1}} \leftarrow {{\bm{\xi }}_{k+1}}\cdot \bm{\eta }_{k+1}^{T}+ \mathcolor{red}{{\mathbf{M}}_{k}}$                                                            & all elements of $\mathbf{M}_k \quad (\forall k)$ \\
15                & ${\tilde{\mathbf{S}}} \leftarrow {\mathbf{S}} + \mathcolor{red}{\mathbf{M}_{K}} + (\mathcolor{red}{\mathbf{M}_{K}})^T$                                                                                                    & all elements of $\mathbf{M}_K$              \\ \hline
\end{tabular}$}
\caption{{Lines of \IncUSRone~(in Appendix~\ref{app:04a}) that require to store $\mathbf{M}_k$ (highlighted in red color)}} \label{tab:05}
\end{table}

In this section, we propose a novel scalable method based on Algorithms~\ref{alg:03}--\ref{alg:06} for dynamical SimRank search,
which updates all pairs of SimRanks column by column using only $O(Kn+m)$ memory,
with no need to store all $(n^2)$ pairs of old SimRank $\mathbf{S}$ into memory, and with no loss of accuracy.

Let us first analyze the $O(n^2)$ memory requirement for Algorithms~\ref{alg:03}--\ref{alg:06} in Sections~\ref{sec:04}--\ref{sec:06}.
We notice that there are two factors dominating the original $O(n^2)$ memory:
(1) the storage of the entire $n \times n$ old SimRank matrix $\mathbf{S}$,
and (2) the computation of $\mathbf{M}_k$ from one outer product.
For example, in \IncUSRone~(in Appendix~\ref{app:04a}),
Lines 3, 4, 6, 9, 15 need to get elements from old $\mathbf{S}$ (see Table~\ref{tab:04});
Lines 10, 14, 15 require to store $ n \times n$ entries of matrix $\mathbf{M}_k$ (see Table~\ref{tab:05}).
Indeed, the storage of $\mathbf{S}$ and $\mathbf{M}_k$ are the main obstacles to the scalability of our in-memory algorithms on large graphs,
resulting in $O(n^2)$ memory space.
Apart from these lines, the memory required for the remaining steps of {\IncUSRone} is $O(m)$,
dominated by (a) the storage of sparse matrix $\mathbf{Q}$ and (b) sparse matrix-vector products.

To overcome the bottleneck of the $O(n^2)$ memory,
our main idea is to update all pairs of $\mathbf{S}$ in a column-by-column style, with no need to store the entire $\mathbf{S}$ and $\mathbf{M}_k$. 
Specifically,
we update $\mathbf{S}$ by updating each column $[\mathbf{{S}}]_{\star,x} \ (\forall x=1,2,\cdots)$ of $\mathbf{S}$ individually.
Let us rewrite Line 15 of Table~\ref{tab:04} into the column-wise style:
\begin{equation} \label{eq:70}
{[\tilde{\mathbf{S}}]}_{\star, x} = {[\mathbf{S}]}_{\star,x} + {[\mathbf{M}_{K}]}_{\star,x} + {[(\mathbf{M}_{K})^T]}_{\star,x} \qquad (\forall x)
\end{equation}
Applying the following facts
\[{[{\mathbf{\Delta S}}]}_{\star, x} = {[\tilde{\mathbf{S}}]}_{\star, x} - {[\mathbf{S}]}_{\star,x}$ and ${[(\mathbf{M}_{K})^T]}_{\star,x} =  ({[\mathbf{M}_{K}]}_{x,\star})^T\]
into Eq.\eqref{eq:70} produces
\begin{equation} \label{eq:71}
{[{\mathbf{\Delta S}}]}_{\star, x} = {[\mathbf{M}_{K}]}_{\star,x} + ({[\mathbf{M}_{K}]}_{x,\star})^T \qquad (\forall x)
\end{equation}
This implies that, to compute one column of ${\mathbf{\Delta S}}$, we only need prepare one row and one column of $\mathbf{M}_{K}$.
To compute only the $x$-th row and $x$-th column of $\mathbf{M}_{K}$, there are two challenges:
(1) From Line 10 of Table~\ref{tab:04},
we notice that $\mathbf{M}_{K}$ is derived from the auxiliary vector $\bm \gamma$, and $\bm \gamma$ depends on the $i$-th and $j$-th column of old ${\mathbf{S}}$ according to Lines 3, 4, 6, 9 of Table~\ref{tab:04}.
Since the update edge $(i,j)$ can be arbitrary,
it is hard to determine which columns of old ${\mathbf{S}}$ will be used in future.
Thus, all our in-memory algorithms in Section \ref{sec:06} prepare $n \times n$ elements of ${\mathbf{S}}$ into memory, leading to $O(n^2)$ memory.
(2) According to Lines 10, 14, 15 of Table~\ref{tab:05},
it also requires $O(n^2)$ memory to iteratively compute $\mathbf{M}_{K}$.
It is not easy to use just linear memory for iteratively computing only one row and one column of $\mathbf{M}_{K}$.
In the next two subsections, we will address these two challenges, respectively.
\subsection{Avoid storing $n \times n$ elements of old $\mathbf{S}$}
\begin{figure*}[t]
\begin{minipage}{.495\textwidth}
\removelatexerror
\begin{algorithm}[H]
\small
\DontPrintSemicolon
\LinesNumbered
\SetKwInOut{Input}{Input}
\SetKwInOut{Output}{Output}
\Input{an old digraph $G=(V,E)$, \\
       a collection of edges $\Delta G$ inserted into $G$, \\
       $x$-th column $[\mathbf{S}]_{\star, x}$ of old SimRank in $G$, \\
       number of iterations $K$, \ \
       damping factor $C$.}
\Output{$x$-th column $[\tilde{\mathbf{S}}]_{\star, x}$ of new SimRank in $G \cup \Delta G$}
\SetKwBlock{Begin}{\textbf{foreach} {edge $(i,j) \in \Delta G$}}{...~(Continue~on~right~side)}
initialize the transition matrix $\mathbf{Q}$ in $G$ ; \;
\lForEach {$v \in V$} {${{d}_{v}} \leftarrow$ in-degree of node $v$ in $G$ ; }
\SetAlgoLined
\Begin
 {  \SetAlgoVlined
    \lIf {$i \in V$} {
        $[\mathbf{S}]_{\star,i} \leftarrow \PartialSim (\mathbf{Q}, i, K, C) $
        }
    \lIf {$j \in V$} {
        $[\mathbf{S}]_{\star,j} \leftarrow \PartialSim (\mathbf{Q}, j, K, C) $
    }
    \uIf(\tcp*[f]{Case (C0)}) {$i \in V$ and $j \in V$} {
        $\mathbf{w} \leftarrow \mathbf{Q}\cdot {{[\mathbf{S}]}_{\star,i}}$; \;
        $\lambda \leftarrow {{[\mathbf{S}]}_{i,i}}+\tfrac{1}{C} \cdot {[\mathbf{S}]}_{j,j}-2\cdot {[\mathbf{w}]}_{j} - \tfrac{1}{C} +1$ ; \;
        \uIf {${{d}_{j}}=0$} {
            $\mathbf{u} \leftarrow \mathbf{e}_j, \ \mathbf{v} := \mathbf{e}_i, \ {\bm \gamma} :=  \mathbf{w} +\frac{1}{2}{{[\mathbf{S}]}_{i,i}}\cdot {{\mathbf{e}}_{j}}$; \;}
        \Else {
            $\mathbf{u} \leftarrow \tfrac{1}{d_j+1} \mathbf{e}_j, \quad \mathbf{v} := \mathbf{e}_i-{[\mathbf{Q}]}_{j,\star}^T$ ; \;
            ${\bm\gamma} \leftarrow \tfrac{1}{({{d}_{j}}+1)} \big( \mathbf{w}-\frac{1}{C}  {{[\mathbf{S}]}_{\star,j}}+( \frac{\lambda }{2\left( {{d}_{j}}+1 \right)}+ \frac{1-C}{C} )  {{\mathbf{e}}_{j}} \big)$;
        }
        initialize ${{\bm{\xi }}_{0}} \leftarrow C \cdot \mathbf{e}_j,\quad {{\bm{\eta }}_{0}} \leftarrow {\bm\gamma}$; \;
        ${{\mathbf{m}}} \leftarrow C \cdot [{\bm\gamma}]_x \cdot \mathbf{e}_j, \ \ {{\mathbf{n}}} \leftarrow C \cdot [\mathbf{e}_j]_x \cdot {\bm\gamma} $; \;
        \For {$k=0,1,\cdots, K-1$} {
            ${{\bm{\xi }}_{k+1}} \leftarrow C \cdot \mathbf{Q}\cdot {{\bm{\xi }}_{k}} + C \cdot (\mathbf{v}^T\cdot {{\bm{\xi }}_{k}}) \cdot \mathbf{u}$ ; \;
            ${{\bm{\eta }}_{k+1}} \leftarrow  \mathbf{Q}\cdot {{\bm{\eta }}_{k}} + (\mathbf{v}^T\cdot {{\bm{\eta }}_{k}}) \cdot \mathbf{u}$ ;\;
            ${{\mathbf{m}}} \leftarrow [\bm{\eta }_{k+1}]_x \cdot {{\bm{\xi }}_{k+1}} +{{\mathbf{m}}}$ ; \;
            ${{\mathbf{n}}} \leftarrow {[{\bm{\xi }}_{k+1}]}_x \cdot \bm{\eta }_{k+1}+{{\mathbf{n}}}$ ; \;
        }
        $[{\mathbf{S}}]_{\star, x} \leftarrow [\mathbf{S}]_{\star, x} + \mathbf{m} + \mathbf{n}$ ; \;
        $d_j \leftarrow d_j+1, \quad \mathbf{Q} \leftarrow \mathbf{Q}  + \mathbf{u} \cdot \mathbf{v}^T $ ; \;
    }
    \uElseIf(\tcp*[f]{Case (C1)}) {$i \in V$ and $j \notin V$} {
        $\mathbf{y} \leftarrow C \cdot \mathbf{Q}\cdot {{[\mathbf{S}]}_{\star,i}}$ ; \;
        \If {$x=j$} {
            $z \leftarrow C \cdot {{[\mathbf{S}]}_{i,i}}+(1-C)$ ; \;
            $[\mathbf{S}]_{\star,x} \leftarrow  \left[ \renewcommand\arraystretch{1.2} \begin{array}{c}
               \mathbf{y}  \\ \hline
               z  \\
            \end{array} \right]$ ; \;
        }
        \Else {
            $[\mathbf{S}]_{\star,x} \leftarrow \left[ \renewcommand\arraystretch{1.2} \begin{array}{c}
               [\mathbf{S}]_{\star,x}   \\ \hline
               {[{\mathbf{y}}]}_x  \\
            \end{array} \right]$ ; \;
        }
$d_j \leftarrow 0, \quad V \leftarrow V \cup \{j\}, \quad \mathbf{Q} \leftarrow \left[ \begin{array}{c|c}
               \mathbf{Q} & \mathbf{0}  \\ \hline
               {{\mathbf{e}}_{i}^{T}} & 0  \\
            \end{array} \right]$; \;
}}
\caption{\IncSRAllP~($G, \Delta G, [\mathbf{S}]_{\star,x}, K, C$)}  \label{alg:07}
\end{algorithm}
\end{minipage}
\begin{minipage}{.495\textwidth}
\removelatexerror
\setcounter{algocf}{4}
\begin{algorithm}[H]
\small
\SetAlgoVlined
\DontPrintSemicolon
%
\SetKwBlock{Begin}{...~(Continued)}{end}
\Begin {
    \everypar={\nl}
    \setcounter{AlgoLine}{30}
    \uElseIf(\tcp*[f]{Case (C2)}) {$i \notin V$ and $j \in V$} {
    \everypar={\nl}
        \uIf {$x=i$} {
            $[\mathbf{S}]_{\star,x} \leftarrow  \left[ \renewcommand\arraystretch{1.2} \begin{array}{c}
               \mathbf{0}  \\ \hline
               1-C  \\
            \end{array} \right]$ ; \; }
        \Else {
                $\mathbf{z} \leftarrow \big( \tfrac{1}{2C ( {{d}_{j}}+1 )}\big( {{[\mathbf{S}]}_{j,j}}-{{(1-C)}^{2}} \big)+\tfrac{1-C}{C} \big){{\mathbf{e}}_{j}}-\tfrac{1}{C}{{[\mathbf{S}]}_{\star,j}}$ ;\;
                initialize ${{\bm{\xi }}_{0}} \leftarrow \mathbf{e}_j,\quad {{\bm{\eta }}_{0}} \leftarrow {\mathbf{z}}$ ; \;
                ${{\mathbf{m}}} \leftarrow {[\mathbf{z}]}_x \cdot \mathbf{e}_j, \quad {{\mathbf{n}}} \leftarrow {[\mathbf{e}_j]}_x \cdot \mathbf{z}$ ; \;
                \For {$k \leftarrow 0,1,\cdots, K-1$} {
                    ${{\bm{\xi }}_{k+1}} \leftarrow  C \cdot \mathbf{Q} \cdot {{\bm{\xi  }}_{k}} - \tfrac{C}{{{d}_{j}}+1}({{[\mathbf{Q}]}_{j,\star}} \cdot \bm{\xi }_{k}) \cdot  {{\mathbf{e}}_{j}}$;\;
                    ${{\bm{\eta }}_{k+1}} \leftarrow \mathbf{Q}\cdot {{\bm{\eta  }}_{k}} - \tfrac{1}{{{d}_{j}}+1}({{[\mathbf{Q}]}_{j,\star}} \cdot \bm{\eta  }_{k}) \cdot {{\mathbf{e}}_{j}}$ ;\;
                    ${{\mathbf{m}}} \leftarrow [\bm{\eta }_{k+1}]_x \cdot {{\bm{\xi }}_{k+1}} +{{\mathbf{m}}}$ ; \;
                    ${{\mathbf{n}}} \leftarrow {[{\bm{\xi }}_{k+1}]}_x \cdot \bm{\eta }_{k+1}+{{\mathbf{n}}}$ ; \;
                }
            $[\mathbf{S}]_{\star,x} \leftarrow  \left[ \renewcommand\arraystretch{1.2} \begin{array}{c}
               [\mathbf{S}]_{\star,x}  + \frac{C}{d_j+1} \cdot (\mathbf{m} + \mathbf{n})  \\ \hline
               0  \\
            \end{array} \right]$ ; \;
        }
        $d_i \leftarrow 0, \quad d_j \leftarrow d_j + 1, \quad V \leftarrow V \cup \{i\}$ ; \;
        $\mathbf{Q} \leftarrow \left[ \begin{array}{c|c}
           {\mathbf{Q}-\tfrac{1}{{{d}_{j}}+1}{{\mathbf{e}}_{j}} {{[\mathbf{Q}]}_{j,\star}}} & \tfrac{1}{{{d}_{j}}+1}{{\mathbf{e}}_{j}}  \\ \hline
           \mathbf{0} & 0  \\
        \end{array} \right] $; \;
    }
    \everypar={\nl} \ElseIf(\tcp*[f]{Case (C3)}) {$i \notin V$ and $j \notin V$} {
    \everypar={\nl}
        \uIf {$x=i$} {
            $[\mathbf{S}]_{\star,x} \leftarrow  \left[ \begin{array}{c}
               \mathbf{0} \\  \hline
               1-C  \\
               0 \\
            \end{array}  \right] \begin{array}{c}
               \\
               (i) \\
               (j) \\
            \end{array}$ ; \;
        }
        \uElseIf {$x=j$} {
            $[\mathbf{S}]_{\star,x} \leftarrow  \left[ \begin{array}{c}
               \mathbf{0} \\  \hline
               0  \\
               1-C^2 \\
            \end{array}  \right] \begin{array}{c}
               \\
               (i) \\
               (j) \\
            \end{array}$ ; \;
        }
        \Else {
            $[\mathbf{S}]_{\star,x} \leftarrow  \left[ \begin{array}{c}
               [\mathbf{S}]_{\star,x}  \\  \hline
               0  \\
               0 \\
            \end{array}  \right] \begin{array}{c}
               \\
               (i) \\
               (j) \\
            \end{array}$ ; \;
        }
        $\mathbf{{Q}} \leftarrow \left[ \begin{array}{c|c}
               \mathbf{Q} & \mathbf{0}  \\  \hline
               \mathbf{0} & \left[ \begin{matrix}
               0 & 0  \\
               1 & 0  \\
            \end{matrix} \right]  \\
            \end{array}  \right] \begin{array}{c}
               \\
               (i) \\
               (j) \\
            \end{array} $ ; \;
        $d_i \leftarrow 0, \quad d_j \leftarrow 0, \quad V \leftarrow V \cup \{i, j \}$ ; \;
    }
    \everypar={\nl} $G \leftarrow G \cup \{(i,j)\}$ ; \;
}
\everypar={\nl}  \Return $[\tilde{\mathbf{S}}]_{\star,x} \leftarrow [\mathbf{S}]_{\star,x} $ ; \;
\caption{\textbf{(Continued)}~~\IncSRAllP}
\end{algorithm}
\end{minipage}
\end{figure*}
Our above analysis imply that, to compute each column ${[{\mathbf{\Delta S}}]}_{\star, x}$, we only need prepare two columns information ($i$-th and $j$-th) from old $\mathbf{S}$.
Since the update edge $(i,j)$ can be arbitrary, there are no prior knowledge which $i$-th and $j$-th columns in old $\mathbf{S}$ will be used.
As opposed to Algorithms~\ref{alg:03}--\ref{alg:06} that memoize all $(n^2)$ pairs of old $\mathbf{S}$,
we use the following scalable method to compute only the $i$-th and $j$-th columns of old $\mathbf{S}$ on demand in linear memory.
Specifically, based on our previous work \cite{Yu2015} on partial-pairs SimRank retrieval,
we can readily verify that the following iterations will yield $[\mathbf{S}]_{\star,i}$ and $[\mathbf{S}]_{\star,j}$ in just $O(Kn+m)$ memory.
\[  \renewcommand\arraystretch{1}
\scalebox{0.85}{$ \begin{tabular}{|l|l|}
\hline
\textrm{initialize} $\mathbf{x}_{0} \leftarrow \mathbf{e}_i$                                    & \textrm{initialize} $\mathbf{x}_{0} \leftarrow \mathbf{e}_j$  \\
\textbf{for} $t \leftarrow 1, 2, \cdots, K$                                                     & \textbf{for} $t \leftarrow 1, 2, \cdots, K$ \\
\qquad  $\mathbf{x}_{t+1} \leftarrow \mathbf{Q}^T \cdot \mathbf{x}_{t} $                        & \qquad  $\mathbf{x}_{t+1} \leftarrow \mathbf{Q}^T \cdot \mathbf{x}_{t} $ \\
\textrm{initialize} $\mathbf{y} \leftarrow \mathbf{x}_{K+1}$                                    & \textrm{initialize} $\mathbf{y} \leftarrow \mathbf{x}_{K+1}$  \\
\textbf{for} $t \leftarrow 1, 2, \cdots, K$                                                     & \textbf{for} $t \leftarrow 1, 2, \cdots, K$ \\
\qquad  $\mathbf{y} \leftarrow \mathbf{x}_{K+1-t} + C \cdot \mathbf{Q} \cdot \mathbf{y} $       & \qquad  $\mathbf{y} \leftarrow \mathbf{x}_{K+1-t} + C \cdot \mathbf{Q} \cdot \mathbf{y}  $ \\
$[\mathbf{S}]_{\star,i} \leftarrow (1-C) \cdot \mathbf{y} $                                     & $[\mathbf{S}]_{\star,j} \leftarrow (1-C) \cdot \mathbf{y} $\\ \hline
\end{tabular} $}
\]
Next, $[\mathbf{S}]_{i,i}$ is obtained from the $i$-th element of $[\mathbf{S}]_{\star,i}$, and $[\mathbf{S}]_{j,j}$ from the $j$-th element of $[\mathbf{S}]_{\star,j}$.
Having prepared $[\mathbf{S}]_{\star,i}, [\mathbf{S}]_{\star,j}, [\mathbf{S}]_{i,i}$, and $ [\mathbf{S}]_{j,j}$,
we follow Lines 3, 4, 6, 9 of Table~\ref{tab:04} to derive the vector $\bm \gamma$ in linear memory.
In addition, since Line 15 of Table~\ref{tab:04} can be computed column-wisely via Eq.\eqref{eq:71}.
Throughout all lines in Table~\ref{tab:04}, we do not need store $n^2$ pairs of old $\mathbf{S}$ in memory.
However, $O(n^2)$ memory is still required to store $\mathbf{M}_k$.
In the next subsection, we will show how to avoid $O(n^2)$ memory to compute $\mathbf{M}_k$.

\subsection{Compute ${[\mathbf{M}_K]}_{\star, x}$ and ${[\mathbf{M}_K]}_{x, \star}$ in linear memory}
Using $\bm \gamma$, we next devise our method based on Table~\ref{tab:05},
aiming to use linear memory to compute each column ${[\mathbf{M}_K]}_{\star, x}$ and each row ${[\mathbf{M}_K]}_{x, \star}$ for Eq.\eqref{eq:71}.
In Table~\ref{tab:05}, our key observation is that $\mathbf{M}_k$ is the summation of the outer product of two vectors.
Due to this structure, instead of using $O(n^2)$ memory to store $\mathbf{M}_k$,
we can use only $O(n)$ memory to compute ${[\mathbf{M}_K]}_{\star, x}$  and ${[\mathbf{M}_K]}_{x,\star}$.
Specifically, we can compute Lines~10 and 14 of Table~\ref{tab:05} in a column-wise style for ${[\mathbf{M}_K]}_{\star, x}$ as follows:
\[  \renewcommand\arraystretch{1}
\scalebox{1}{$ \begin{tabular}{|l|l|}
\hline
${[{\mathbf{M}}_{0}]}_{\star, x } \leftarrow C \cdot {[\bm\gamma]}_x \cdot  \mathbf{e}_j $                                                      \\
\textbf{for~} $ k \leftarrow 0, \cdots,K-1$                                                                                                     \\
\qquad ${[{\mathbf{M}}_{k+1}]}_{\star, x } \leftarrow [\bm{\eta }_{k+1}]_{x} \cdot {{\bm{\xi }}_{k+1}} + {[{\mathbf{M}}_{k}]}_{\star, x }$      \\
\hline
\end{tabular} $}
\]
and in a row-wise style for ${[\mathbf{M}_K]}_{x,\star}$ as follows:
\[  \renewcommand\arraystretch{1}
\scalebox{1}{$ \begin{tabular}{|l|l|}
\hline
${[{\mathbf{M}}_{0}]}_{x, \star} \leftarrow C \cdot {[\mathbf{e}_j]}_x \cdot \bm\gamma  $ \\
\textbf{for~} $ k \leftarrow 0, \cdots,K-1$ \\
\qquad ${[{\mathbf{M}}_{k+1}]}_{x, \star} \leftarrow [{{\bm{\xi }}_{k+1}}]_{x} \cdot \bm{\eta }_{k+1} + {[{\mathbf{M}}_{k}]}_{x, \star}$ \\
\hline
\end{tabular} $}
\]
Fig.~\ref{fig:04} pictorially visualizes the column-wise computation of ${[{\mathbf{M}}_{K}]}_{\star, x }$.
Having computed ${[{\mathbf{M}}_{K}]}_{\star, x }$ and ${[{\mathbf{M}}_{K}]}_{x, \star}$,
we can use Eq.\eqref{eq:71} to derive the column ${[{\mathbf{\Delta S}}]}_{\star, x}$ of ${\mathbf{\Delta S}}$.

The main advantage of our method is that, throughout the entire updating process,
we need not store $n \times n$ pairs of $\mathbf{M}_k$ and $\mathbf{S}$,
and thereby, significantly reduce the memory usage from $O(n^2)$ to $O(Kn+m)$.
In addition to the insertion case (C0), our memory-efficient methods are applicable to other insertion cases in Subsection~\ref{sec:04c}.
The complete algorithm, denoted as \IncSRAllP, is described in Algorithm~\ref{alg:07}.
{\IncSRAllP} is a memory-efficient version of Algorithms~\ref{alg:03}--\ref{alg:06}.
It includes a procedure {\PartialSim} that allows us to compute two columns information of old $\mathbf{S}$ on demand in linear memory,
rather than store $n^2$ pairs of old $\mathbf{S}$ in memory.
In response to each edge update $(i,j)$,
once the two old columns $\mathbf{S}_{\star,i}$ and $\mathbf{S}_{\star,j}$ are computed via {\PartialSim} for updating the $x$-th column $[\mathbf{\Delta S}]_{\star,x}$,
they can be memoized in only $O(n)$ memory and reused later to compute another $y$-th column $[\mathbf{\Delta S}]_{\star,y}$ in response to the edge update $(i,j)$.

\vspace{5pt} \noindent \textbf{Correctness.} \
{\IncSRAllP} correctly returns similarity.
It consists of four update cases:
lines 6--22 for Case (C0),
lines 23--30 for Case (C1),
lines 31--45 for Case (C2), and
lines 46--54 for Case (C3).
The correctness of each case can be verified by Theorems~\ref{thm:03}, \ref{thm:08}, \ref{thm:06}, and \ref{thm:07}, respectively.
For instance, to verify the correctness for Case (C0),
we apply successive substitution to \textsf{{for}-loop} in lines~14--21, which produces the following result:
\[
[\tilde{\mathbf{S}}]_{u,v} = [\mathbf{{S}}]_{u,v}+ \sum_{k=1}^K {[\bm{\xi }_{k}]_{u} \cdot [\bm{\eta }_{k}]_{v}} +   \sum_{k=1}^K {[\bm{\xi }_{k}]_{v} \cdot [\bm{\eta }_{k}]_{u}}
\]
This is consistent with Eq.\eqref{eq:70}, implying that our memory-efficient method does not compromise any accuracy for scalability.
\setcounter{algocf}{0}
\begin{algorithm}[!t]
\small
\SetAlgorithmName{Procedure}{}
\DontPrintSemicolon
\LinesNumbered
\SetKwInOut{Input}{Input}
\SetKwInOut{Output}{Output}
\Input{transition matrix $\mathbf{Q}$ in $G$, \\
       query node $q$, \\
       number of iterations $K$, \\
       damping factor $C$.}
\Output{$q$-th column $[{\mathbf{{S}}}]_{\star,q}$ of SimRank scores in $G$.}

initialize $\mathbf{x}_{0} \leftarrow \mathbf{e}_q$ \;      
\For {$t \leftarrow 1, 2, \cdots, K$} {
    $\mathbf{x}_{t+1} \leftarrow \mathbf{Q}^T \cdot \mathbf{x}_{t} $ \;
}
initialize $\mathbf{y} \leftarrow \mathbf{x}_{K+1}$ \;
\For {$t \leftarrow 1, 2, \cdots, K$} {
    $\mathbf{y} \leftarrow \mathbf{x}_{K+1-t} + C \cdot \mathbf{Q} \cdot \mathbf{y} $ \;
}
\Return $[\mathbf{S}]_{\star,q} \leftarrow (1-C) \cdot \mathbf{y} $ \;
\caption{\PartialSim $(\mathbf{Q}, q, K, C)$}  \label{alg:08}
\end{algorithm}

\begin{figure}[!t] \centering
  \includegraphics[width=.9\linewidth]{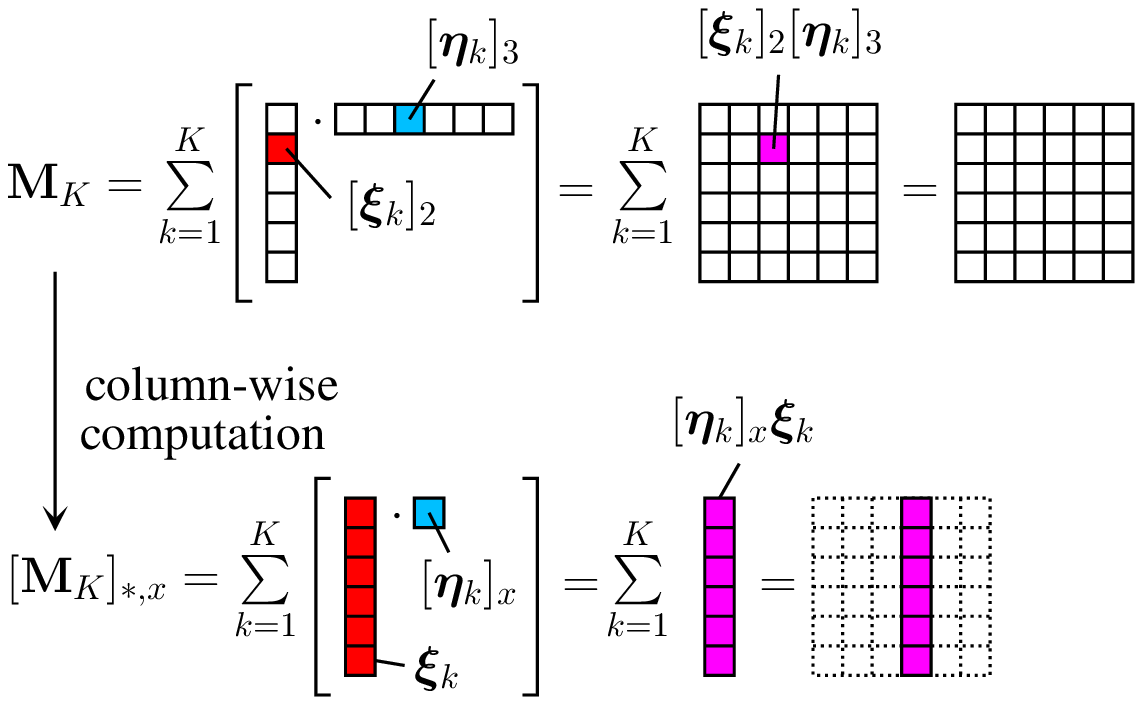}
  \caption{Memory usage reduction by partitioning $\mathbf{M}_K$ in a column-by-column style} \label{fig:04} 
\end{figure}
%
        It is worth mentioning that {\IncSRAllP} can be also combined with our batch updating method in Section~\ref{sec:08}.
        This will speed up the dynamical update of SimRank further, with $O(n(\max_{t=1}^{|B|}\delta_t) + m + Kn)$ memory.
        Here $O(n\delta_t)$ memory is needed to store $\delta_t$ columns of $\mathbf{S}$ when $[\mathbf{S}]_{\star,I}$ is required for processing the $t$-th block.
%
%
%
}
\section{Experimental Evaluation} \label{sec:09} 
In this section, we present a comprehensive experimental study on real and synthetic datasets,
to demonstrate
(i) the fast computational time of \IncSR~to incrementally update SimRanks on large time-varying networks, 
(ii) the pruning power of {\IncSR} that can discard unnecessary incremental updates outside ``affected areas'';
(iii) the exactness of \IncSR, as compared with \IncSVD;
(iv) the high efficiency of our complete scheme that integrates {\IncSR} with {\IncUSRtwo, \IncUSRthree, \IncUSRfour}
to support link updates that allow new node insertions;
(v) the fast computation time of our batch update algorithm {\IncBSR} against the unit update method {\IncSR};
(vi) the scalability of our memory-efficient algorithm {\IncSRAllP} in Section~\ref{sec:07} on million-scale large graphs for dynamical updates;
(vii) the performance comparison between {\IncSRAllP} and {\LTSF} in terms of computational time, memory space, and top-$k$ exactness;
(viii) the average updating time and memory usage of {\IncSRAllP} for each case of edge updates.
\begin{table*}[t]
\centering \scalebox{0.8}{$
\begin{tabular}{|l|ll|rr|rl|l|}
\hline
\multicolumn{3}{|c|}{\textbf{Datasets}}                & \textbf{$|V|$} & \textbf{$|E|$} & \multicolumn{2}{c|}{\textbf{\# of Pairs To Be Assessed}} & \multicolumn{1}{c|}{\textbf{Description}}   \\ \hline
\parbox[t]{2mm}{\multirow{2}{*}{\rotatebox[origin=c]{90}{Small}}}  & \DBLP  & (DBLP)            & 13,634         & 93,560         & 185,885,956                 & $(={|V|}^2) $                  & DBLP citation network                      \\
                        & \CITH  & (cit-HepPh)       & 34,546         & 421,578        & 1,193,426,116               & $(={|V|}^2) $                  & High Energy Physics citation network \\ \hline
\parbox[t]{2mm}{\multirow{3}{*}{\rotatebox[origin=c]{90}{Medium}}} & \YOUTU & (YouTube)         & 178,470        & 953,534        & 1,784,700,000               & $(= {10}^4 {|V|})$             & Social network of YouTube videos           \\
                        & \WEBB  & (web-BerkStan)    & 685,230        & 7,600,595      & 6,852,300,000               & $(= {10}^4 {|V|})$             & Web graph of Berkeley and Stanford         \\
                        & \WEBG  & (web-Google)      & 916,428        & 5,105,039      & 9,164,280,000               & $(= {10}^4 {|V|})$             & Web graph from Google                      \\ \hline
\parbox[t]{2mm}{\multirow{4}{*}{\rotatebox[origin=c]{90}{Large}}}  & \CITP  & (cit-Patents)     & 3,774,768      & 16,518,948     & 3,774,768,000               & $(= {10}^3 {|V|})$             & Citation network among US Patents          \\
                        & \SOCL  & (soc-LiveJournal) & 4,847,571      & 68,993,773     & 4,847,571,000               & $(= {10}^3 {|V|})$             & LiveJournal online social network          \\
                        & \UK    & (uk-2005)         & 39,459,925     & 936,364,282    & 39,459,925,000              & $(= {10}^3 {|V|})$             & Web graph from 2005 crawl of .uk domain    \\
                        & \IT    & (it-2004)         & 41,291,594     & 1,150,725,436  & 41,291,594,000              & $(= {10}^3 {|V|})$             & Web graph from 2004 crawl of .it domain \\ \hline
\end{tabular}
%
$}
  \caption{Description of Real-World Datasets} \label{tab:03}
\end{table*}
\subsection{Experimental Settings}
\noindent \textbf{Datasets.} \
We adopt both real and synthetic datasets.
The real datasets include small-scale ({\DBLP} and {\CITH}), medium-scale ({\YOUTU}, {\WEBB} and {\WEBG}), and large-scale graphs ({\CITP}, {\SOCL}, {\UK}, and {\IT}).
Table~\ref{tab:03} summarises the description of these datasets.

(Please refer to Appendix~\ref{app:05} for details.)

To generate synthetic graphs and updates, we adopt \textsf{GraphGen}\footnote{http://www.cse.ust.hk/graphgen/} generation engine.
The graphs are controlled by (a) the number of nodes $|V|$, and (b) the number of edges $|E|$.
We produce a sequence of graphs that follow the linkage generation model \cite{Garg2009}.
To control graph updates, we use two parameters simulating real evolution:
(a) update type (edge/node insertion or deletion), and (b) the size of updates $|\Delta G|$.

\noindent \textbf{Algorithms.} \
We implement the following algorithms:
(a) \IncSVD,
the SVD-based link-update algorithm \cite{Li2010}; 
(b) \IncUSR, our incremental method without pruning;
(c) \Batch, the batch SimRank method via fine-grained memoization \cite{Yu2013};
(d) \IncSR, our incremental method with pruning power but not supporting node insertions;
(e) \IncSRAll, our complete enhanced version of {\IncSR} that allows node insertions by incorporating {\IncUSRtwo, \IncUSRthree, and \IncUSRfour};
(f) \IncBSR, our batch incremental update version of {\IncSR};
(g) \IncSRAllP, our memory-efficient version of {\IncSRAll} that dynamically computes the SimRank matrix column by column without the need to store all pairs of old similarities;
(h) \LTSF, the log-based implementation of the existing competitor, TSF \cite{Shao2015}, which supports dynamic SimRank updates for top-$k$ querying.

\noindent \textbf{Parameters.} \
We set the damping factor $C=0.6$, as used in \cite{Jeh2002}.
By default, the total number of iterations is set to $K=15$ to guarantee accuracy ${C}^{K} \le 0.0005$~\cite{Lizorkin2008}.
On {\CITH} and {\YOUTU}, we set $K=10$;
On large graphs ({\CITP}, {\SOCL}, {\UK}, and {\IT}), we set $K=5$.
The target rank $r$ for {\IncSVD} is a speed-accuracy trade-off, we set $r=5$ in our time evaluation since,
as shown in the experiments of \cite{Li2010},
the highest speedup is achieved when $r=5$.
In our exactness evaluation, we shall tune this value.
For {\LTSF} algorithm, we set the number of one-way graphs $R_g = 100$, and the number of samples at query time $R_q=20$, as suggested in \cite{Shao2015}.

All the algorithms are implemented in Visual C++ and Matlab.
For small-scale graphs, we use a machine with an Intel Core 2.80GHz CPU and 8GB RAM.
For medium- and large-scale graphs, we use a processor with Intel Core i7-6700 3.40GHz CPU and 64GB RAM. 

%
\subsection{Experimental Results}
\subsubsection{Time Efficiency of {\IncSR} and {\IncUSR}}
\begin{figure*}
\centering
\begin{minipage}[t]{0.71\linewidth}
\centering
  \includegraphics[width=\linewidth]{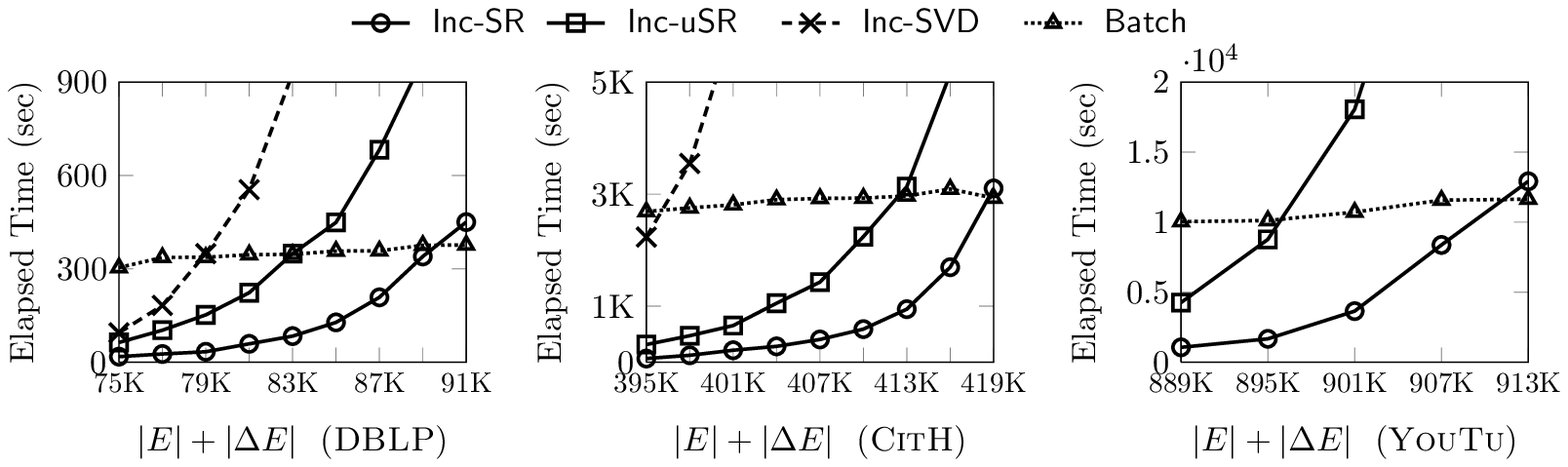} \\
  \caption{Time Efficiency on Real Data ($\Delta E$ does not accompany new nodes)}\label{fig:exp_01_02_03}
\end{minipage}
\begin{minipage}[t]{0.28\linewidth}
\centering
  \includegraphics[width=.95\linewidth]{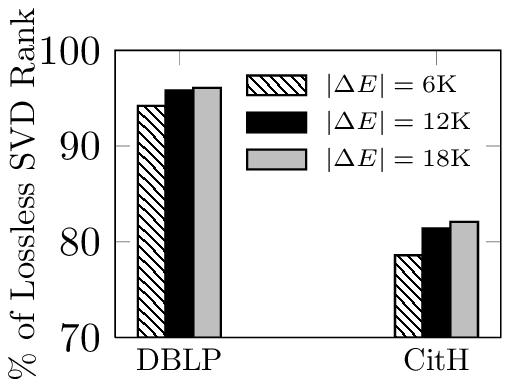}
  \caption{\% of Lossless SVD Rank}\label{fig:exp_04}
\end{minipage}
\end{figure*}
\begin{figure*}
\centering
\begin{minipage}[t]{.44\linewidth}
\centering
  \includegraphics[width=\linewidth]{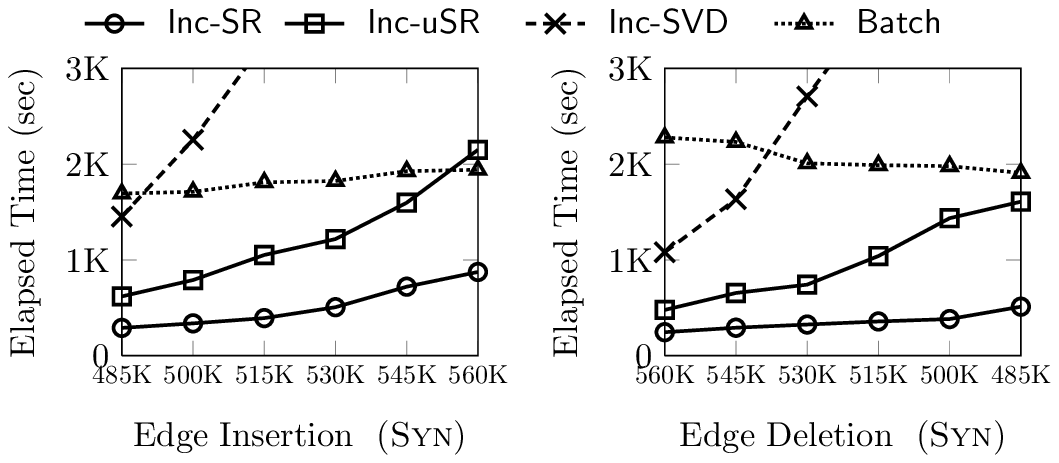} \\
  \caption{Time Efficiency on Synthetic Data}\label{fig:exp_05_06}
\end{minipage}
\begin{minipage}[t]{0.26\linewidth}
\centering
  \includegraphics[width=.95\linewidth]{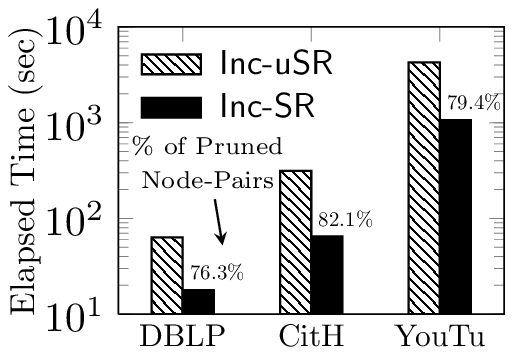}
  \caption{Pruning Power}\label{fig:exp_07}
\end{minipage}
\begin{minipage}[t]{0.26\linewidth}
\centering
  \includegraphics[width=.95\linewidth]{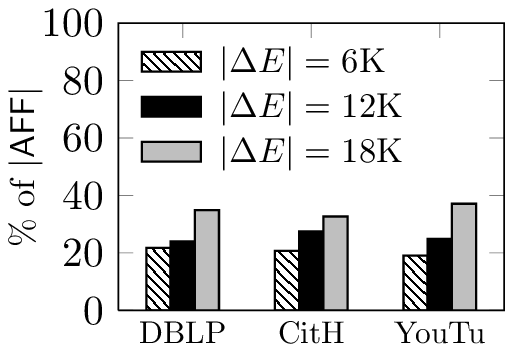}
  \caption{\% of Affected Areas}\label{fig:exp_08}
\end{minipage}
\end{figure*}
We first evaluate the computational time of {\IncSR} and {\IncUSR} against {\IncSVD} and {\Batch} on real datasets.

Note that, to favor {\IncSVD} that only works on small graphs (due to memory crash for high-dimension SVD $n>10^5$),
we just use {\IncSVD} on {\DBLP} and {\CITH}. 

Fig.\ref{fig:exp_01_02_03} depicts the results when edges are added to \DBLP, \CITH, \YOUTU, respectively.
For each dataset, we fix $|V|$ and increase $|E|$ by $|\Delta E|$, as shown in the $x$-axis.
Here, the edge updates are the differences between snapshots \wrt the ``year'' (\Resp ``video age'') attribute of \DBLP, \CITH~(\Resp \YOUTU),
reflecting their real-world evolution.
We observe the following.
(1) \IncSR~\emph{always} outperforms \IncSVD~and \IncUSR~when edges are increased.
For example, on \DBLP, when the edge changes are 10.7\%,
the time for \IncSR~(83.7s) is 11.2x faster than \IncSVD~(937.4s), and 4.2x faster than \IncUSR~(348.7s).
This is because \IncSR~employs a rank-one matrix method to update the similarities,
with an effective pruning strategy to skip unnecessary recomputations,
as opposed to \IncSVD~that entails rather expensive costs to incrementally update the SVD.
The results on \CITH~are more pronounced, \eg
\IncSR~is 30x better than \IncSVD~when $|E|$ is increased to 401K.
(2) \IncSR~is consistently better than \Batch~when the edge changes are fewer than 19.7\% on \DBLP, and 7.2\% on \CITH.
When link updates are 5.3\% on \DBLP~(\Resp 3.9\% on \CITH), \IncSR~improves \Batch~by 10.2x (\Resp 4.9x).
This is because (i) \IncSR~can exploit the sparse structure of $\mathbf{\Delta S}$ for incremental update,
and (ii) small link perturbations may keep $\mathbf{\Delta S}$ sparsity.
Hence, \IncSR~is highly efficient when link updates are small.
(3) The computational time of \IncSR, \IncUSR, \IncSVD, unlike \Batch, is sensitive to the edge updates $|\Delta E|$.
The reason is that \Batch~needs to reassess all similarities from scratch in response to link updates,
whereas \IncSR~and \IncUSR~can reuse the old information in SimRank for incremental updates.
In addition, \IncSVD~is too sensitive to $|\Delta E|$,
as it entails expensive tensor products to compute SimRank from the updated SVD matrices. 

Fig.\ref{fig:exp_04} shows the target rank $r$ required for the Li \etal\!\!'s {lossless} SVD approach \wrt the edge changes $|\Delta E|$ on \DBLP~and \CITH.
The $y$-axis is $\frac{r}{n} \times 100\%$. 
On each dataset, when increasing $|\Delta E|$ from 6K to 18K,
we see that $\frac{r}{n}$ is 95\% on \DBLP~(\Resp 80\% on \CITH),
Thus, $r$ is not negligibly smaller than $n$ in real graphs.
Due to the quartic time \wrt $r$,
\IncSVD~may be slow in practice to get a high accuracy.

On synthetic data, we fix $|V|=79,483$ and vary $|E|$ from 485K to 560K (\Resp 560K to 485K) in 15K increments (\Resp decrements).
The results are shown in Fig.\ref{fig:exp_05_06}. 
We can see that,
when 6.4\% edges are increased,
\IncSR~runs 8.4x faster than \IncSVD, 4.7x faster than \Batch, and 2.7x faster than \IncUSR.
When 8.8\% edges are deleted,
{\IncSR} outperforms {\IncSVD} by 10.4x, {\Batch} by 5.5x, and {\IncUSR} by 2.9x.
This justifies our complexity analysis of {\IncSR} and \IncUSR.
\subsubsection{Effectiveness of Pruning} 
Fig.\ref{fig:exp_07} shows the pruning power of {\IncSR} as compared with {\IncUSR} on \DBLP, \CITH, and \YOUTU,
in which
%
the percentage of the pruned node-pairs of each graph is depicted on the black bar.
The $y$-axis is in a log scale.
It can be discerned that, on every dataset,
\IncSR~constantly outperforms \IncUSR~by nearly 0.5 order of magnitude.
For instance, the running time of \IncSR~(64.9s) improves that of \IncUSR~(314.2s)~by 4.8x on \CITH,
with approximately 82.1\% node-pairs being pruned.
That is, our pruning strategy is powerful to discard unnecessary node-pairs on graphs with different link distributions.
\begin{figure*}
\centering
\begin{minipage}[b]{0.71\linewidth}
\centering
  \includegraphics[width=\linewidth]{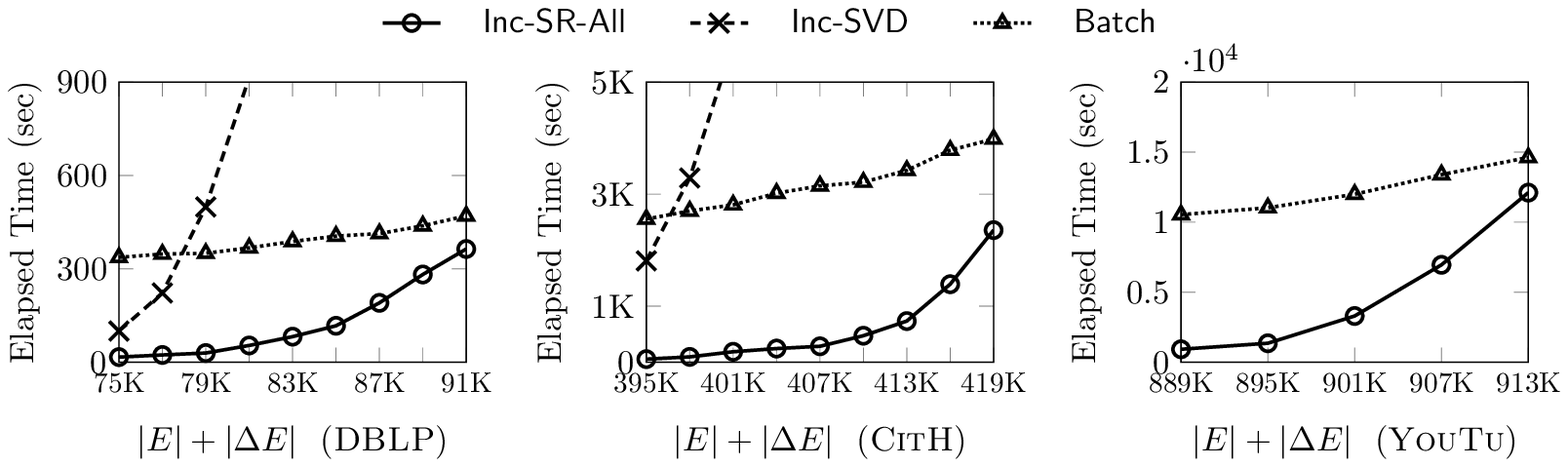} \\
  \caption{Time Efficiency on Real Data ($\Delta E$ accompanies new node insertions)}\label{fig:exp_11_12_13}
\end{minipage}
\begin{minipage}[b]{0.28\linewidth}
\centering
\scalebox{0.8}{$
  \begin{tabular}{c|r|rr|r}
\hline
\multicolumn{2}{c|}{Data ($|E|$)}    & \IncBSR  & \IncSRAll & (\%) \\ \hline
\parbox[t]{2mm}{\multirow{3}{*}{\rotatebox[origin=c]{90}{\DBLP}}}    & 75K  & 14.9    & 16.3     & 8.8  \\
                         & 83K  & 70.5    & 82.0     & 14.0 \\
                         & 91K  & 315.9   & 363.8    & 13.1 \\ \hline
\parbox[t]{2mm}{\multirow{3}{*}{\rotatebox[origin=c]{90}{\CITH}}}    & 395K & 50.5    & 54.5     & 7.3  \\
                         & 407K & 241.9   & 283.5    & 14.6 \\
                         & 419K & 1869.1  & 2357.4   & 20.7 \\ \hline
\parbox[t]{2mm}{\multirow{3}{*}{\rotatebox[origin=c]{90}{\YOUTU}}} & 889K & 876.6   & 921.9    & 4.9  \\
                         & 901K & 2756.8  & 3297.4   & 16.4 \\
                         & 913K & 10256.1 & 12109.2  & 15.3 \\ \hline
\end{tabular}$}
  \caption{Time for Batch Updates}\label{fig:exp_14}
\end{minipage}
\end{figure*}
%
%

Since our pruning strategy hinges mainly on the size of the ``affected areas'' of the SimRank update matrix,
Fig.\ref{fig:exp_08} illustrates the percentage of the ``affected areas'' of SimRank scores \wrt link updates $|\Delta E|$ on \DBLP, \CITH, and \YOUTU.
We find the following.
(1) When $|\Delta E|$ is varied from 6K to 18K on every real dataset,
the ``affected areas'' of SimRank scores are fairly small.
For instance, when $|\Delta E|=12$K,
the percentage of the ``affected areas'' is only 23.9\% on \DBLP, 27.5\% on \CITH, and 24.8\% on \YOUTU, respectively.
This highlights the effectiveness of our pruning method in real applications,
where a larger number of elements of the SimRank update matrix with zero scores can be discarded.
(2) For each dataset, the size of the ``affect areas'' mildly grows when $|\Delta E|$ is increased.
For example, on \YOUTU, the percentage of $|\AFF|$ increases from 19.0\% to 24.8\% when $|\Delta E|$ is changed from 6K to 12K.
This agrees with our time efficiency analysis,
where the speedup of {\IncSR} is more obvious for smaller $|\Delta E|$.
\begin{figure*}[t]
\centering
\begin{minipage}[b]{0.77\linewidth}
\centering
\scalebox{0.8}{$
\begin{tabular}{c|ccc|c|ccc}
\hline
\multirow{4}{*}{Datasets} & \multicolumn{3}{c|}{\IncSRAll}                                                                                                                                                                         & \IncBSR                                                                                  & \multicolumn{3}{c}{\IncSVD} \\ \cline{2-8}
                          & \begin{tabular}[c]{@{}c@{}}No \\ Optimization\end{tabular} & \begin{tabular}[c]{@{}c@{}}Turn on \\ Pruning\end{tabular} & \begin{tabular}[c]{@{}c@{}}Turn on Column-\\ wise Partitioning\end{tabular} & \begin{tabular}[c]{@{}c@{}}Turn on Pruning\\ \& Column-wise\\ Partitioning\end{tabular} & $r=5$  & $r=15$  & $r=25$  \\ \hline
\DBLP                     & 722.5M                                                     & 163.1M                                                     & 1.3M                                                                        & 15.0M                                                                                   & 1.36G  & 1.97G   & 3.86G   \\
\CITH                     & 1.64G                                                      & 413.9M                                                     & 4.2M                                                                        & 34.8M                                                                                   & 4.83G  & ---     & ---     \\
\YOUTU                    & ---                                                        & ---                                                        & 12.7M                                                                       & 186.2M                                                                                  & ---    & ---     & ---     \\ \hline
\end{tabular}$}
  \caption{Total Memory Efficiency on Real Data \ (``---'' means memory explosion)}\label{fig:exp_09}
\end{minipage}
\begin{minipage}[b]{0.22\linewidth}
\centering
  \includegraphics[height=2.6cm]{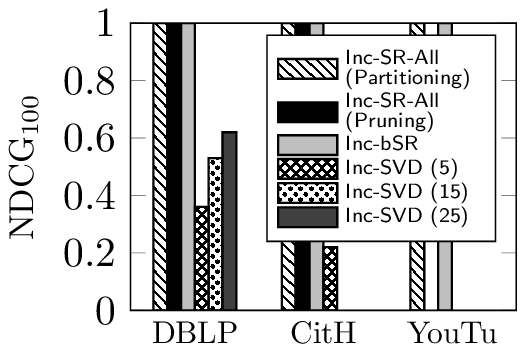}
  \caption{Exactness}\label{fig:exp_15}
\end{minipage}
\end{figure*}
\subsubsection{Time Efficiency of {\IncSRAll} and {\IncBSR}} 
We next compare the computational time of {\IncSRAll} with {\IncSVD} and {\Batch} on \DBLP, \CITH, and \YOUTU.
For each dataset, we increase $|E|$ by $|\Delta E|$ that might accompany new node insertions.
Note that {\IncSR} cannot deal with such incremental updates as $\mathbf{\Delta S}$ does not make any sense in such situations.
To enable {\IncSVD} to handle new node insertions, we view new inserted nodes as singleton nodes in the old graph $G$.
Fig.~\ref{fig:exp_11_12_13} depicts the results.
We can discern that
(1) on every dataset, {\IncSRAll} runs substantially faster than {\IncSVD} and {\Batch} when $|\Delta E|$ is small.
For example, as $|\Delta E|=6K$ on {\CITH}, {\IncSRAll} (186s) is 30.6x faster than {\IncSVD} (5692s) and 15.1x faster than {\Batch} (2809s).
The reason is that {\IncSRAll} can integrate the merits of {\IncSR} with {\IncUSRtwo, \IncUSRthree, \IncUSRfour}
to dynamically update SimRank scores in a rank-one style with no need to do costly matrix-matrix multiplications.
Moreover, the complete framework of {\IncSRAll} allows itself to support link updates that enables new node insertions.
(2) When $|\Delta E|$ grows larger on each dataset,
the time of {\IncSVD} increases significantly faster than {\IncSRAll}.
This larger increase is due to the SVD tensor products used by {\IncSVD}.
In contrast, {\IncSRAll} can effectively reuse the old SimRank scores to compute changes even if such changes may accompany new node insertions.

Fig.~\ref{fig:exp_14} compares the computational time of {\IncBSR} with {\IncSRAll}.
From the results, we can notice that, on each graph,
{\IncBSR} is consistently faster than {\IncSRAll}.
The last column ``(\%)'' denotes the percentage of {\IncBSR} improvement on speedup.
On each dataset,
the speedup of {\IncBSR} is more apparent when $|\Delta E|$ grows larger.
For example, on {\DBLP}, the improvement of {\IncBSR} over {\IncSRAll} is 8.8\% when $|E|=75$K, and 14.0\% when $|E|=83$K.
On {\CITH} (\Resp {\YOUTU}), the highest speedup of {\IncBSR} over {\IncSRAll} is 20.7\% for $|E|=419$K (\Resp 16.4\% for $|E|=901$K).
This is because the large size of $|\Delta E|$ may increase the number of the new inserted edges with one endpoint overlapped.
Hence, more edges can be handled simultaneously by {\IncBSR}, highlighting its high efficiency over {\IncSRAll}.
\subsubsection{Total Memory Usage}
%
Fig.~\ref{fig:exp_09} evaluates the total memory usage of {\IncSRAll} and {\IncBSR} against {\IncSVD} on real datasets.
Note that the total memory usage includes the storage of the old SimRanks required for all-pairs dynamical evaluation.
For {\IncSRAll}, we test its three versions:
(a) We first switch off our methods of ``pruning'' and ``column-wise partitioning'',
denoted as ``No Optimization''; (b) next turn on ``pruning'' only; and (c) finally turn on both.
For {\IncSVD}, we also tune the default target rank $r=5$ larger to see how the memory space is affected by $r$.

The results indicate that
(1) on each dataset when the memory of {\IncSVD} and {\IncBSR} does not explode,
the total spaces of {\IncSRAll} and {\IncBSR} are consistently much smaller {\IncSVD} whatever target rank $r$ is.
This is because, unlike {\IncSVD}, {\IncSRAll} and {\IncBSR} need not memorize the results of SVD tensor products.
(2) When the ``pruning'' switch is open, the space of {\IncSRAll} can be reduced by $\sim4$x further on real data,
due to our pruning method that discards many zeros in auxiliary vectors and final SimRanks.
(3) When the ``column-wise partitioning'' switch is open,
the space of {\IncSRAll} can be saved by $\sim100$x further.
The reason is that, as all pairs of SimRanks can be computed in a column-by-column style,
there is no need to memorize the entire old SimRank $\mathbf{S}$ and auxiliary $\mathbf{M}$.
This improvement agrees with our space analysis in Section~\ref{sec:07}.
(4) The space of {\IncBSR} is 8-11x larger than {\IncSRAll}, but is still acceptable.
This is because batch updates require more space to memoize several columns from the old $\mathbf{S}$ to handle a subset of edge updates simultaneously.
(5) For {\IncSVD}, when the target rank $r$ is varied from 5 to 25, its total space increases from 1.36G to 3.86G on \DBLP,
but crashes on \CITH~and \YOUTU.
This implies that $r$ has a huge impact on the space of \IncSVD,
and is not negligible in the big-$O$ analysis of \cite{Li2010}.
%
%
\begin{figure*}[t]
\centering
\begin{minipage}[b]{0.48\linewidth}
\centering
\scalebox{0.8}{$
  \begin{tabular}{l|r|r|rr}
  \hline
\multicolumn{1}{c|}{\multirow{2}{*}{Datasets}} & \multicolumn{1}{c|}{\multirow{2}{*}{\IncSRAllP}} & \multicolumn{3}{c}{\LTSF}                                                                 \\ \cline{3-5}
\multicolumn{1}{c|}{}                          & \multicolumn{1}{c|}{}                         & \multicolumn{1}{c|}{Total} & \multicolumn{1}{c}{Index (Merge)} & \multicolumn{1}{c}{Query} \\ \hline
\WEBB                                         & 0.453                                        & 4.764                     & 4.758                             & 0.006                     \\
\WEBG                                         & 1.440                                        & 6.883                     & 6.876                             & 0.007                     \\
\CITP                                         & 3.820                                        & 20.549                    & 20.536                            & 0.013                     \\
\SOCL                                         & 35.393                                       & 67.372                    & 67.322                            & 0.050                     \\
\UK                                           & 63.125                                       & 460.718                   & 460.360                           & 0.358                     \\
\IT                                           & 69.301                                       & 505.794                   & 505.400                           & 0.393 \\ \hline
\end{tabular}$}
  \caption{Avg Time (secs) for $\mathbf{S}_{\star,u}$ per Edge Update}\label{fig:exp_17}
\end{minipage} \qquad
\begin{minipage}[b]{0.42\linewidth}
\centering
  \includegraphics[height=2.9cm]{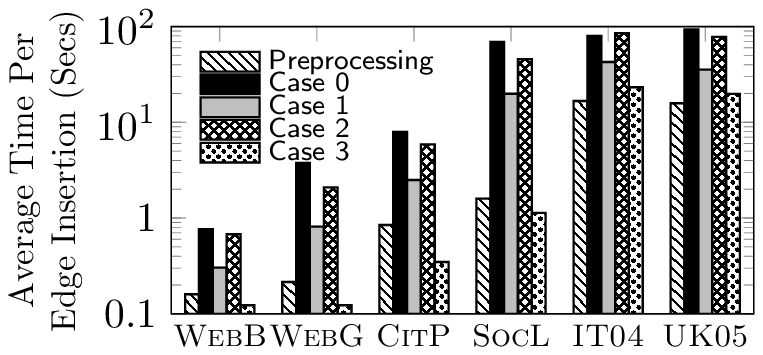}
  \caption{Avg Time for Each Insertion Case}\label{fig:exp_16}
\end{minipage}
\end{figure*}
\subsubsection{Exactness}
%
We next evaluate the exactness of {\IncSRAll}, {\IncBSR}, and {\IncSVD} on real datasets.
We leverage the NDCG metrics \cite{Li2010} to assess the top-100 most similar pairs. 
We adopt the results of the batch algorithm \cite{Fujiwara2013} on each dataset as the $\textrm{NDCG}_{100}$ baselines, due to its exactness.
For {\IncSRAll}, we evaluate its two enhanced versions: ``with column-wise partitioning'' and ``with pruning'';
for {\IncSVD}, we tune its target rank $r$ from 5 to 25.


Fig.~\ref{fig:exp_15} depicts the results, showing the following.
(1) On each dataset, the $\textrm{NDCG}_{100}$s of {\IncSRAll} and {\IncBSR} are 1, which are better than {\IncSVD} ($<0.62$).
This agrees with our observation that {\IncSVD} may loss eigen-information in real graphs.
In contrast, {\IncSRAll} and {\IncBSR} guarantee the exactness.
(2) The $\textrm{NDCG}_{100}$s for the two versions of {\IncSRAll} are exactly the same,
implying that both our pruning and column-wise partitioning methods are lossless while achieving high speedup.
%
%
%
%
%
\subsubsection{Scalability on Large Graphs}
To evaluate the scalability of our incremental techniques,
we run {\IncSRAllP}, a memory-efficient version of {\IncSR}, on six real graphs ({\WEBB}, {\WEBG}, {\CITP}, {\SOCL}, {\UK}, and {\IT}),
and compare its performance with {\LTSF}.
Both {\IncSRAllP} and {\LTSF} can compute any single column, $\mathbf{S}_{\star, u}$, of $\mathbf{S}$
with no need to memoize all $n^2$ pairs of the old $\mathbf{S}$.
To choose the query node $u$, we randomly pick up 10,000 queries from each medium-sized graph ({\WEBB} and {\WEBG}),
and 1,000 queries from each large-sized graph ({\CITP}, {\SOCL}, {\UK}, and {\IT}).
To ensure every selected $u$ can cover a board range of any possible queries,
for each dataset, we first sort all nodes in $V$ in descending order based on their importance that is measured by PageRank (PR),
and then split all nodes into 10 buckets: nodes with $\textrm{PR} \in [0.9, 1]$ are in the first bucket;
nodes with $\textrm{PR} \in [0.8, 0.9)$ the second, etc.
For every medium- (\Resp large-) sized graph,
we randomly select 1,000 (\Resp 100) queries from each bucket,
such that $u$ contains a wide range of various types of queries.
To generate dynamical updates, we follow the settings in \cite{Shao2015},
randomly choosing 1,000 edges, and considering 80\% of them as insertions and 20\% deletions.

Fig.~\ref{fig:exp_17} compares the average time of {\IncSRAllP} and {\LTSF} required to compute any column $\mathbf{S}_{\star, u}$ \wrt a given query $u$ for each edge update on six real graphs.
It can be discerned that, on each dataset, {\IncSRAllP} is scalable well over large graphs, and runs consistently 4--7x faster than log-based {\LTSF} per edge update.
On one-billion edge graphs (\IT), for every edge update,
the updating time of {\IncSRAllP} (69.301s) is 7.3x faster than that of {\LTSF} (505.794s).
This is because the time of {\LTSF} is dominated by its cost of merging $R_g=100$ one-way graphs' log buffers for updating the index.
For example, on large {\IT}, almost 99.92\% time required by {\LTSF} is due to its merge operations.
In comparison, our memory-efficient method for {\IncSRAllP} takes advantage of the rank-one Sylvester equation which computes the updates to $\mathbf{S}_{\star, u}$ in a column-by-column style on demand,
without the need to merge one-way graphs and memoize all pairs of old $\mathbf{S}$ in advance.

Fig.~\ref{fig:exp_16} shows the time complexities of {\IncSRAllP} for four cases of edge insertions on each real dataset.
For every graph, we randomly select 1,000 edges $\{(i,j)\}$ for insertion updates,
with nodes $i$ and $j$ respectively having the probability ${1}/{2}$ to be picked up from the old vertex set $V$.
Hence, each case of edge insertion occurs at ${1}/{4}$ probability.
For each insertion case, we sum all the time spent in this case, and divide it by the total number of edge insertions counted for this case.
Fig.~\ref{fig:exp_16} reports the average time per edge update for each case, together with the preprocessing time over each dataset (including the cost of loading the graph and preparing its transition matrix $\mathbf{Q}$).
From the results, we see that, on each dataset, the time spent for Cases (C0) and (C2) is moderately higher than that for Case (C1);
Case (C0) is slightly slower than Case (C2); Case (C3) entails the lowest time cost.
These results are consistent with our intuition and mathematical formulation of $\mathbf{\Delta S}$ for each case.
Case (C0) has the most expensive time cost as it needs to iteratively prepare vectors $\bm\xi_k$ and $\bm\eta_k$, and old similarities $\mathbf{S}_{\star,i}$ and $\mathbf{S}_{\star,j}$ via matrix-vector products.
In contrast, Case (C2) only requires to iteratively prepare $\bm\xi_k, \bm\eta_k$ and $\mathbf{S}_{\star,i}$;
Case (C1) just requires to perform one matrix-vector product to prepare one vector $\mathbf{y}$.
For Case (C3), the new inserted edge forms a new component of the graph.
There is no precomputation of any auxiliary vectors, and thus Case (C3) is the fastest.
\begin{figure*}[t]
\centering
\begin{minipage}[b]{0.25\linewidth}
\centering
  \includegraphics[height=2.9cm]{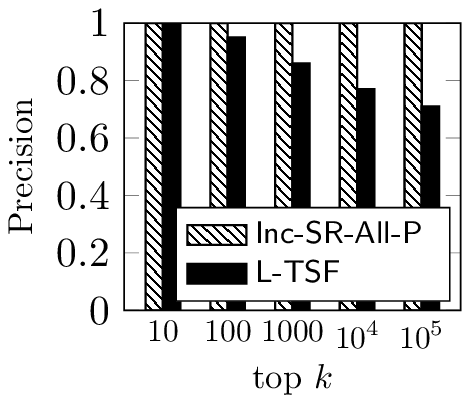}
  \caption{Precision on {\YOUTU}}\label{fig:exp_18}
\end{minipage}
\begin{minipage}[b]{0.38\linewidth}
\centering
  \includegraphics[height=2.9cm]{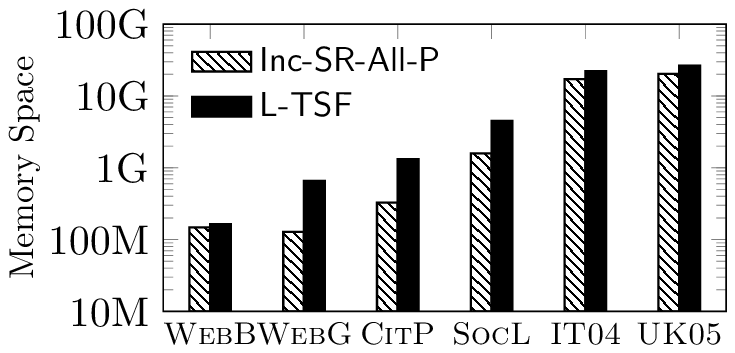}
  \caption{Memory of {\IncSRAllP} \& {\LTSF}}\label{fig:exp_19}
\end{minipage}
\begin{minipage}[b]{0.35\linewidth}
\centering
  \includegraphics[height=2.9cm]{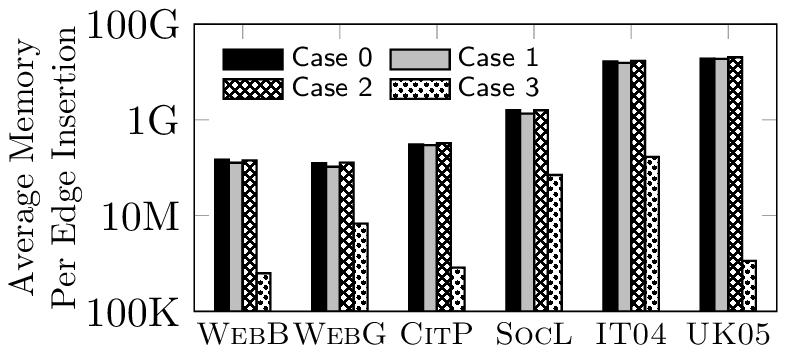}
  \caption{Memory for Each Insertion Case}\label{fig:exp_20}
\end{minipage}
\end{figure*}
\subsubsection{Precision}
To compare the precision of {\IncSRAllP} and {\LTSF},
we define the \emph{precision} measure  \cite{Jiang2017} for top-$k$ querying:
\[
\textrm{Precision}= \frac{|\textrm{approximate top-$k$ set} \cap \textrm{exact top-$k$ set}|}{k}
\]
The original batch algorithm in \cite{Jeh2002} (\Resp \cite{Li2010}) serves as the exact solution to obtain SimRank results for {\LTSF} (\Resp {\IncSRAllP}).
We evaluate the precision of both algorithms on several real datasets.
Fig.~\ref{fig:exp_18} reports the results on {\YOUTU};
the tendencies on other datasets are similar.
We see that, when top-$k$ varies from 10 to $10^5$,
the precision of {\LTSF} remains high $(>84\%)$ for small top-$k$ $(<1000)$,
but is lower (68\%--75\%) for large top-$k$ $(>{10}^4)$.
This is because the probabilistic guarantee for the error bound of {\LTSF} is based on the assumption that no cycle in the given graph has a length shorter than $K$ (the total number of steps).
Hence, {\LTSF} is highly efficient for top-$k$ single-source querying, where $k$ is not large.
In contrast, the precision of {\IncSRAllP} is stable at 1,
meaning that it produces the exact SimRank results of \cite{Li2010}, regardless of its top-$k$ size.
Thus, {\IncSRAllP} is better for non top-$k$ query.
\subsubsection{Memory of {\IncSRAllP}}
Fig.~\ref{fig:exp_19} evaluates the memory usage of {\IncSRAllP} and {\LTSF} over six real datasets.
We observe that both algorithms scale well on large graphs.
On {\WEBB}, {\IT}, and {\UK}, the memory space of {\IncSRAllP} is almost the same as {\LTSF};
On {\WEBG}, {\CITP}, and {\SOCL}, the memory usage of {\IncSRAllP} is 5--8x less than {\LTSF}.
This is because, unlike {\LTSF} that need load a one-way graph to memory,
{\IncSRAllP} only requires to prepare the vector information of $\bm\xi_k, \bm\eta_k$, old $\mathbf{S}_{\star,i}$, and old $\mathbf{S}_{\star,j}$ to assess the changes to each column of $\mathbf{S}$ in response to edge update $(i,j)$.
The memory space of these auxiliary vectors can sometimes be comparable to the size of the one-way graph, and sometimes be much smaller.
However, such memory space is linear to $n$ as we do not need $n^2$ space to store the entire old $\mathbf{S}$.
Note that the old $\mathbf{S}_{\star,j}$ and $\mathbf{S}_{\star,i}$ can be computed on demand with only linear memory by our partial-pairs SimRank approach \cite{Yu2015}.
Moreover, we see that, with the growing scale of the real datasets,
the memory space of {\IncSRAllP} is increasing linearly, highlighting its scalability on large graphs.

Fig.~\ref{fig:exp_20} depicts further the average memory usage of {\IncSRAllP} for each case of edge insertion.
We randomly pick up 1,000 edges $\{(i,j)\}$ for insertion updates on each dataset,
with nodes $i$ and $j$ respectively having the probability ${1}/{2}$ to be chosen from the old vertex set $V$.
The average memory space of {\IncSRAllP} for each case is reported in Fig.~\ref{fig:exp_20}.
We see that, on each dataset, the memory required for Cases (C0), (C1), and (C2) are similar,
whereas the memory space of Case (C3) is much smaller than the other cases.
The reason is that, for Cases (C0), (C1), and (C2),
{\IncSRAllP} needs linear memory to store some auxiliary vectors
 (\eg $\bm\xi_k, \bm\eta_k$, $\mathbf{y}$, old $\mathbf{S}_{\star,i}$, and old $\mathbf{S}_{\star,j}$) for updating SimRank scores,
whereas for Case (C3), no auxiliary vectors are required for precomputation, thus saving much memory space.
\section{Related Work} \label{sec:02}
%
%
        Recent results on SimRank can be distinguished into two categories:
        (i) dynamical SimRank \cite{He2010,Li2010,Yu2014,Shao2015,Jiang2017}, and
        (ii) static SimRank \cite{Kusumoto2014,Fujiwara2013,Li2010a,Li2015,Fogaras2007,Lee2012,Lizorkin2008,Yu2013}.

\subsection{Incremental SimRank}

%
Li \etal \cite{Li2010} devised an interesting matrix representation of SimRank,
and was the first to show a SVD method for incrementally updating all-pairs SimRanks, which requires $O(r^4n^2)$ time
and $O({r}^{2}n^2)$ memory.
However, their incremental techniques are \emph{inherently} inexact, with no guaranteed accuracy.

{Recently, Shao \etal \cite{Shao2015} provided an excellent exposition of a two-stage random sampling framework, TSF, for top-$k$ SimRank dynamic search \wrt one query $u$.
In the preprocessing stage, they sampled a collection of one-way graphs to index random walks in a scalable manner.
In the query stage, they retrieved similar nodes by pruning unqualified nodes based on the connectivity of one-way graph.
To retrieve \emph{top-$k$ nodes} with highest SimRank scores in \emph{a single column} $\mathbf{S}_{\star,u}$,
\cite{Shao2015} requires $O(K^2 R_q R_g)$ \emph{average} query time to retrieve $\mathbf{S}_{\star,u}$ along with $O(n \log k)$ time to return top-$k$ results from $\mathbf{S}_{\star,u}$.
The recent work of Jiang \etal \cite{Jiang2017} has argued that, to retrieve $\mathbf{S}_{\star,u}$, the querying time of \cite{Shao2015} is $O(K n R_q R_g)$.
The $n$ factor is due to the time to traverse the reversed one-way graph; in the worst case, all $n$ nodes are visited.
Moreover, Jiang \etal \cite{Jiang2017} observed that the probabilistic error guarantee of Shao \etal\!\!'s method is based on the assumption that no cycle in the given graph has a length shorter than $K$,
and they proposed READS, a new efficient indexing scheme that improves precision and indexing space for dynamic SimRank search.
The query time of READS is $O(r n)$ to retrieve one column $\mathbf{S}_{\star,u}$, where $r$ is the number of sets of random walks.
Hence, TSF and READS are highly efficient for \emph{top-$k$ single-source} SimRank search.
In comparison, our dynamical method focuses on \emph{all $(n^2)$-pairs} SimRank search in $O(K(m+|\AFF|))$ time.
Optimization methods in this work are based on a rank-one Sylvester matrix equation characterising changes to $n^2$ pairs of SimRank scores,
which is fundamentally different from \cite{Shao2015,Jiang2017}'s methods that maintain one-way graphs (or SA forests) updating.
It is important to note that, for large-scale graphs, our incremental methods do not need to memoize all $(n^2)$ pairs of old SimRank scores,
and can dynamically update $\mathbf{S}$ column-wisely in only $O(Kn+m)$ memory.
For updating each column of $\mathbf{S}$,
our experiments in Section~\ref{sec:09} verify that
our memory-efficient incremental method is scalable on large real graphs while running 4--7 times faster than the dynamical TSF \cite{Shao2015} per edge update,
due to the high cost of \cite{Shao2015} merging one-way graph's log buffers for TSF indexing.
}

There has also been a body of work on incremental computation of other graph-based relevance measures. 
Banhmani \etal \cite{Bahmani2010} utilized the Monte Carlo method for incrementally computing Personalized PageRank.
Desikan \etal \cite{Desikan2005} proposed an excellent incremental PageRank algorithm for node updating.
Their central idea revolves around the first-order Markov chain.
Sarma \etal \cite{Sarma2011} presented an excellent exposition of randomly sampling random walks of short length, and merging them together to estimate PageRank on graph streams.
%
%
\subsection{Batch SimRank}
        In comparison to incremental algorithms,
        the batch SimRank computation has been well-studied on static graphs.

For deterministic methods,
Jeh and Widom \cite{Jeh2002} were the first to propose an iterative paradigm for SimRank,
entailing $O(Kd^2n^2)$ time for $K$ iterations,
where $d$ is the average in-degree.
Later,
Lizorkin \etal \cite{Lizorkin2008} utilized the partial sums memoization to speed up SimRank computation to $O(Kdn^2)$.
Yu \etal \cite{Yu2013} have also improved SimRank computation to $O(Kd'n^2)$ time (with $d' \le d$)
via a fine-grained memoization to share the common parts among different partial sums.
Fujiwara \etal \cite{Fujiwara2013} exploited the matrix form of \cite{Li2010} to find the top-$k$ similar nodes in $O(n)$ time \wrt a given query node.
All these methods require $O(n^2)$ memory to output all pairs of SimRanks.
Recently, Kusumoto \etal \cite{Kusumoto2014} proposed a linearized method that requires only $O(dn)$ memory and $O(K^2 d n^2)$ time to compute all pairs of SimRanks.
The recent work of \cite{Yu2015} proposes an efficient aggregate method for computing partial pairs of SimRank scores.
The main ideas of partial-pairs SimRank search are also incorporated into the incremental model of our work,
achieving linear memory to update $n^2$-pairs similarities.

For parallel SimRank computing,
Li \etal~\cite{Li2015} proposed a highly parallelizable algorithm, called CloudWalker, for large-scale SimRank search on Spark with ten machines.
Their method consists of offline and online phases.
For offline processing, an indexing vector is derived by solving a linear system in parallel.
For online querying, similarities are computed instantly from the index vector.
Throughout, the Monte Carlo method is used to maximally reduce time and space.

The recent work of Zhang \etal \cite{Zhang2017} conducted extensive experiments and discussed in depth many existing SimRank algorithms in a unified environment using different metrics,
encompassing efficiency, effectiveness, robustness, and scalability.
The empirical study for 10 algorithms from 2002 to 2015 shows that, despite many recent research efforts,
the running time and precision of known algorithms have still much space for improvement.
This work makes a further step towards this goal.

Fogaras and R{\'a}cz \cite{Fogaras2007} proposed P-SimRank in linear time to estimate a single-pair SimRank $s(a,b)$ from the probability that two random surfers, starting from $a$ and $b$, will finally meet at a node.
Li \etal \cite{Li2010a} harnessed the random walks to compute local SimRank for a single node-pair.
Later, Lee \etal \cite{Lee2012} employed the Monte Carlo method to find top-$k$ SimRank node-pairs.
Tao \etal \cite{Tao2014} proposed an excellent two-stage way for the top-$k$ SimRank-based similarity join.

Recently, Tian and Xiao \cite{Tian2016} proposed SLING, an efficient index structure for static SimRank computation.
SLING requires $O(n/\epsilon)$ space and $O(m/\epsilon+n\log \tfrac{n}{\delta} /\epsilon)$ pre-computation time, and answers any single-pair (\textit{resp.} single-source) query in $O(1/\epsilon)$ (\Resp $O(n/\epsilon)$) time.

%
%
\section{Conclusions} \label{sec:11} 
In this article,
we study the problem of incrementally updating SimRank scores on time-varying graphs.
Our complete scheme, {\IncSRAll}, consists of five ingredients:
(1) For edge updates that do not accompany new node insertions,
we characterize the SimRank update matrix $\mathbf{\Delta S}$ via a rank-one Sylvester equation.
Based on this, a novel efficient algorithm is devised,
which reduces the incremental computation of SimRank from $O(r^4n^2)$ to $O(Kn^2)$ for each link update.
(2) To eliminate unnecessary SimRank updates further,
we also devise an effective pruning strategy that can improve the incremental computation of SimRank to $O(K(m+|\AFF|))$,
where $|\AFF|\ (\ll n^2)$ is the size of the ``affected areas'' in the SimRank update matrix.
(3) For edge updates that accompany new node insertions,
we consider three insertion cases, according to which end of the inserted edge is a new node.
For each case, we devise an efficient incremental SimRank algorithm that can support new node insertions and accurately update the affected similarities.
(4) For batch updates, we also propose efficient batch incremental methods that can handle ``similar sink edges'' simultaneously and eliminate redundant edge updates.
(5) To optimize the memory for all-pairs SimRank updates,
we also devise a column-wise memory-efficient technique that significantly reduces the storage from $O(n^2)$ to $O(Kn+m)$, without the need to memoize $n^2$ pairs of SimRank scores.
Our experimental evaluations on real and synthetic datasets demonstrate that
(a) our incremental scheme is consistently 5--10 times faster than Li \etal\!\!'s SVD based method;
(b) our pruning strategy can speed up the incremental SimRank further by 3--6 times;
(c) the batch update algorithm enables an extra 5--15\% speedup, with just a little compromise in memory;
(d) our memory-efficient incremental method is scalable on billion-edge graphs;
for every edge update, its computational time can be 4--7 times faster than {\LTSF} and its memory space can be 5--8 times less than {\LTSF};
(e) for different cases of edge updates, Cases (C0) and (C2) entail more time than Case (C1), and Case (C3) runs the fastest.
\bibliographystyle{abbrv} 
\bibliography{ref}

\begin{thebibliography}{10}

\bibitem{Bahmani2010}
B.~Bahmani, A.~Chowdhury, and A.~Goel.
\newblock Fast incremental and personalized {PageRank}.
\newblock {\em PVLDB}, 4(3), 2010.

\bibitem{Berkhin2005}
P.~Berkhin.
\newblock Survey: {A} survey on {PageRank} computing.
\newblock {\em Internet Mathematics}, 2, 2005.

\bibitem{Desikan2005}
P.~K. Desikan, N.~Pathak, J.~Srivastava, and V.~Kumar.
\newblock Incremental {PageRank} computation on evolving graphs.
\newblock In {\em WWW}, 2005.

\bibitem{Fogaras2005}
D.~Fogaras and B.~R{\'a}cz.
\newblock Scaling link-based similarity search.
\newblock In {\em WWW}, 2005.

\bibitem{Fogaras2007}
D.~Fogaras and B.~R{\'a}cz.
\newblock Practical algorithms and lower bounds for similarity search in
  massive graphs.
\newblock {\em IEEE Trans. Knowl. Data Eng.}, 19, 2007.

\bibitem{Fujiwara2013}
Y.~Fujiwara, M.~Nakatsuji, H.~Shiokawa, and M.~Onizuka.
\newblock Efficient search algorithm for {SimRank}.
\newblock In {\em ICDE}, 2013.

\bibitem{Garg2009}
S.~Garg, T.~Gupta, N.~Carlsson, and A.~Mahanti.
\newblock Evolution of an online social aggregation network: {A}n empirical
  study.
\newblock In {\em Internet Measurement Conference}, 2009.

\bibitem{He2010}
G.~He, H.~Feng, C.~Li, and H.~Chen.
\newblock Parallel {SimRank} computation on large graphs with iterative
  aggregation.
\newblock In {\em KDD}, 2010.

\bibitem{Jeh2002}
G.~Jeh and J.~Widom.
\newblock {SimRank}: {A} measure of structural-context similarity.
\newblock In {\em KDD}, 2002.

\bibitem{Jiang2017}
M.~Jiang, A.~W. Fu, R.~C. Wong, and K.~Wang.
\newblock {READS:} {A} random walk approach for efficient and accurate dynamic
  simrank.
\newblock {\em {PVLDB}}, 10(9):937--948, 2017.

\bibitem{Kusumoto2014}
M.~Kusumoto, T.~Maehara, and K.~Kawarabayashi.
\newblock Scalable similarity search for {SimRank}.
\newblock In {\em {SIGMOD}}, 2014.

\bibitem{Lee2012}
P.~Lee, L.~V. Lakshmanan, and J.~X. Yu.
\newblock On top-$k$ structural similarity search.
\newblock In {\em ICDE}, 2012.

\bibitem{Li2010}
C.~Li, J.~Han, G.~He, X.~Jin, Y.~Sun, Y.~Yu, and T.~Wu.
\newblock Fast computation of {SimRank} for static and dynamic information
  networks.
\newblock In {\em EDBT}, 2010.

\bibitem{Li2010a}
P.~Li, H.~Liu, J.~X. Yu, J.~He, and X.~Du.
\newblock Fast single-pair {SimRank} computation.
\newblock In {\em SDM}, 2010.

\bibitem{Li2015}
Z.~Li, Y.~Fang, Q.~Liu, J.~Cheng, R.~Cheng, and J.~C.~S. Lui.
\newblock Walking in the cloud: Parallel {SimRank} at scale.
\newblock {\em {PVLDB}}, 9(1):24--35, 2015.

\bibitem{Lizorkin2008}
D.~Lizorkin, P.~Velikhov, M.~N. Grinev, and D.~Turdakov.
\newblock Accuracy estimate and optimization techniques for {SimRank}
  computation.
\newblock {\em PVLDB}, 1, 2008.

\bibitem{Ntoulas2004}
A.~Ntoulas, J.~Cho, and C.~Olston.
\newblock What's new on the web?: {T}he evolution of the web from a search
  engine perspective.
\newblock In {\em WWW}, 2004.

\bibitem{Rothe2014}
S.~Rothe and H.~Sch{\"{u}}tze.
\newblock {CoSimRank:} {A} flexible {\&} efficient graph-theoretic similarity
  measure.
\newblock In {\em {ACL}}, pages 1392--1402, 2014.

\bibitem{Sarma2011}
A.~D. Sarma, S.~Gollapudi, and R.~Panigrahy.
\newblock Estimating {PageRank} on graph streams.
\newblock {\em J. ACM}, 58:13, 2011.

\bibitem{Shao2015}
Y.~Shao, B.~Cui, L.~Chen, M.~Liu, and X.~Xie.
\newblock An efficient similarity search framework for {SimRank} over large
  dynamic graphs.
\newblock {\em {PVLDB}}, 8(8):838--849, 2015.

\bibitem{Sun2011}
Y.~Sun, J.~Han, X.~Yan, P.~S. Yu, and T.~Wu.
\newblock {PathSim}: {M}eta path-based top-$k$ similarity search in
  heterogeneous information networks.
\newblock {\em PVLDB}, 4, 2011.

\bibitem{Tao2014}
W.~Tao, M.~Yu, and G.~Li.
\newblock Efficient top-$k$ {SimRank}-based similarity join.
\newblock {\em {PVLDB}}, 8(3):317--328, 2014.

\bibitem{Tian2016}
B.~Tian and X.~Xiao.
\newblock {SLING:} {A} near-optimal index structure for simrank.
\newblock In {\em {SIGMOD}}, pages 1859--1874, 2016.

\bibitem{Yu2013}
W.~Yu, X.~Lin, and W.~Zhang.
\newblock Towards efficient {SimRank} computation on large networks.
\newblock In {\em ICDE}, 2013.

\bibitem{Yu2014}
W.~Yu, X.~Lin, and W.~Zhang.
\newblock Fast incremental {SimRank} on link-evolving graphs.
\newblock In {\em ICDE}, pages 304--315, 2014.

\bibitem{Yu2015}
W.~Yu and J.~A. McCann.
\newblock Efficient partial-pairs {SimRank} search for large networks.
\newblock {\em {PVLDB}}, 8(5):569--580, 2015.

\bibitem{Yu2015a}
W.~Yu and J.~A. McCann.
\newblock High quality graph-based similarity retrieval.
\newblock In {\em SIGIR}, 2015.

\bibitem{Zhang2017}
Z.~Zhang, Y.~Shao, B.~Cui, and C.~Zhang.
\newblock An experimental evaluation of {SimRank}-based similarity search
  algorithms.
\newblock {\em {PVLDB}}, 10(5):601--612, 2017.

\end{thebibliography}

\begin{appendices}
\section{Limitation of Li \etal\!\!'s SVD \cite{Li2010}} \label{app:01}
%
%
We rigorously explain the reason why Li \etal\!\!'s incremental method may miss some eigen-information even if a lossless SVD is utilized for SimRank computation.


Let us first revisit the main idea of Li \etal\!\!'s incremental method~\cite{Li2010}.
Briefly, \cite{Li2010} characterizes SimRank matrix $\mathbf{S}$ in Eq.\eqref{eq:03a} in terms of three matrices $\mathbf{U},\mathbf{\Sigma},{\mathbf{V}}$,
%
%
where $\mathbf{U}, \mathbf{\Sigma},{\mathbf{V}}$ are derived by the SVD of $\mathbf{Q}$, \ie
\begin{equation} \label{eq:01}
  \mathbf{Q}=\mathbf{U} \cdot \mathbf{\Sigma} \cdot {\mathbf{V}}^T.
\end{equation}
Then, when links are changed,
\cite{Li2010} incrementally computes the new SimRank matrix $\tilde{\mathbf{S}}$
by updating the old matrices $\mathbf{U}, \mathbf{\Sigma}, {\mathbf{V}}$ respectively as
\begin{equation} \label{eq:02}
\tilde{\mathbf{U}}=\mathbf{U}\cdot \mathbf{U_C}, \quad
\tilde{\mathbf{\Sigma}}=\mathbf{\Sigma_C}, \quad
\tilde{{\mathbf{V}}}= {\mathbf{V}} \cdot {\mathbf{V_C}},
\footnote{In the sequel, we abuse a tilde to denote the updated version of a matrix,
\eg $\tilde{\mathbf{U}}$ is the updated matrix of old $\mathbf{U}$ after link updates.}
\end{equation}
where $\mathbf{U_C},\mathbf{\Sigma_C},{\mathbf{V_C}}$ are derived from the SVD of the auxiliary matrix $\mathbf{C} \triangleq \mathbf{\Sigma} + \mathbf{U}^T \cdot \mathbf{\Delta Q} \cdot \mathbf{V}$, \ie
\begin{equation} \label{eq:03}
\mathbf{C}= \mathbf{U_C} \cdot \mathbf{\Sigma_C} \cdot{\mathbf{V_C}}^T,
\end{equation}
and $\mathbf{\Delta Q}$ is the changes to $\mathbf{Q}$ in response to link updates.

However,
the main problem is that
the derivation of Eq.\eqref{eq:02} rests on the assumption that
\begin{equation} \label{eq:08}
\mathbf{U} \cdot \mathbf{U}^T = \mathbf{V} \cdot \mathbf{V}^T = \mathbf{I}_n.
\end{equation}
Unfortunately, Eq.\eqref{eq:08} does \emph{not} hold (unless $\mathbf{Q}$ is a full-rank matrix, \ie $\textrm{rank}(\mathbf{Q})=n$)
because in the case of $\textrm{rank}(\mathbf{Q})<n$,
even a ``perfect'' (lossless) SVD of $\mathbf{Q}$ via Eq.\eqref{eq:01} would produce $n \times \alpha$ \emph{rectangular} matrices $\mathbf{U}$ and $\mathbf{V}$ with $\alpha=\textrm{rank}(\mathbf{Q})<n$.
Thus,
\[\textrm{rank}(\mathbf{U} \cdot \mathbf{U}^T)=\alpha<n=\textrm{rank}(\mathbf{I}_n),\]
which implies that $\mathbf{U} \cdot \mathbf{U}^T \neq \mathbf{I}_n$.
Similarly, $\mathbf{V} \cdot \mathbf{V}^T \neq \mathbf{I}_n$ when $\textrm{rank}(\mathbf{Q})<n$.
Hence, Eq.\eqref{eq:08} is not always true,
as visualized in Fig.~\ref{fig:05}.
\begin{example} \label{eg:02}
Consider a graph with the matrix $\mathbf{Q}=\left[\begin{smallmatrix}0 & 1 \\0 & 0 \end{smallmatrix}\right]$,
and its lossless SVD:
\[\mathbf{Q}=\mathbf{U} \cdot \mathbf{\Sigma} \cdot {\mathbf{V}}^T \ \textrm{with} \ \mathbf{U}=\left[\begin{smallmatrix}1 \\0 \end{smallmatrix}\right], \ \mathbf{\Sigma}=[1], \ {\mathbf{V}}=\left[\begin{smallmatrix}0 \\1 \end{smallmatrix}\right].\]
One can readily verify that
\[
\mathbf{U} \cdot \mathbf{U}^T = \left[\begin{smallmatrix}1 \\0 \end{smallmatrix}\right] \cdot [\begin{smallmatrix}1 & 0 \end{smallmatrix}] = \left[\begin{smallmatrix}1 & 0\\0 & 0 \end{smallmatrix}\right] \neq \left[\begin{smallmatrix}1 & 0\\0 & 1 \end{smallmatrix}\right] = \mathbf{I}_n
\quad (n=2),
\]
whereas
\[
\mathbf{U}^T \cdot \mathbf{U} = [\begin{smallmatrix}1 & 0 \end{smallmatrix}] \cdot  \left[\begin{smallmatrix}1 \\0 \end{smallmatrix}\right]  = 1 = \mathbf{I}_\alpha
\footnote{The notation $\mathbf{I}_{\alpha}$ denotes the $\alpha \times \alpha$ identity matrix.}
\quad  (\alpha=\textrm{rank}(\mathbf{Q})=1).
\]
Thus, Eq.\eqref{eq:08} does not hold when $\mathbf{Q}$ is not full-rank.
\qed
\end{example}
To clarify why Eq.\eqref{eq:08} gets involved in the derivation of Eq.\eqref{eq:02},
let us briefly recall from \cite{Li2010} the four steps of obtaining Eq.\eqref{eq:02},
and the problem lies in the last step.
\begin{figure}[t] \centering
  \includegraphics[width=\linewidth]{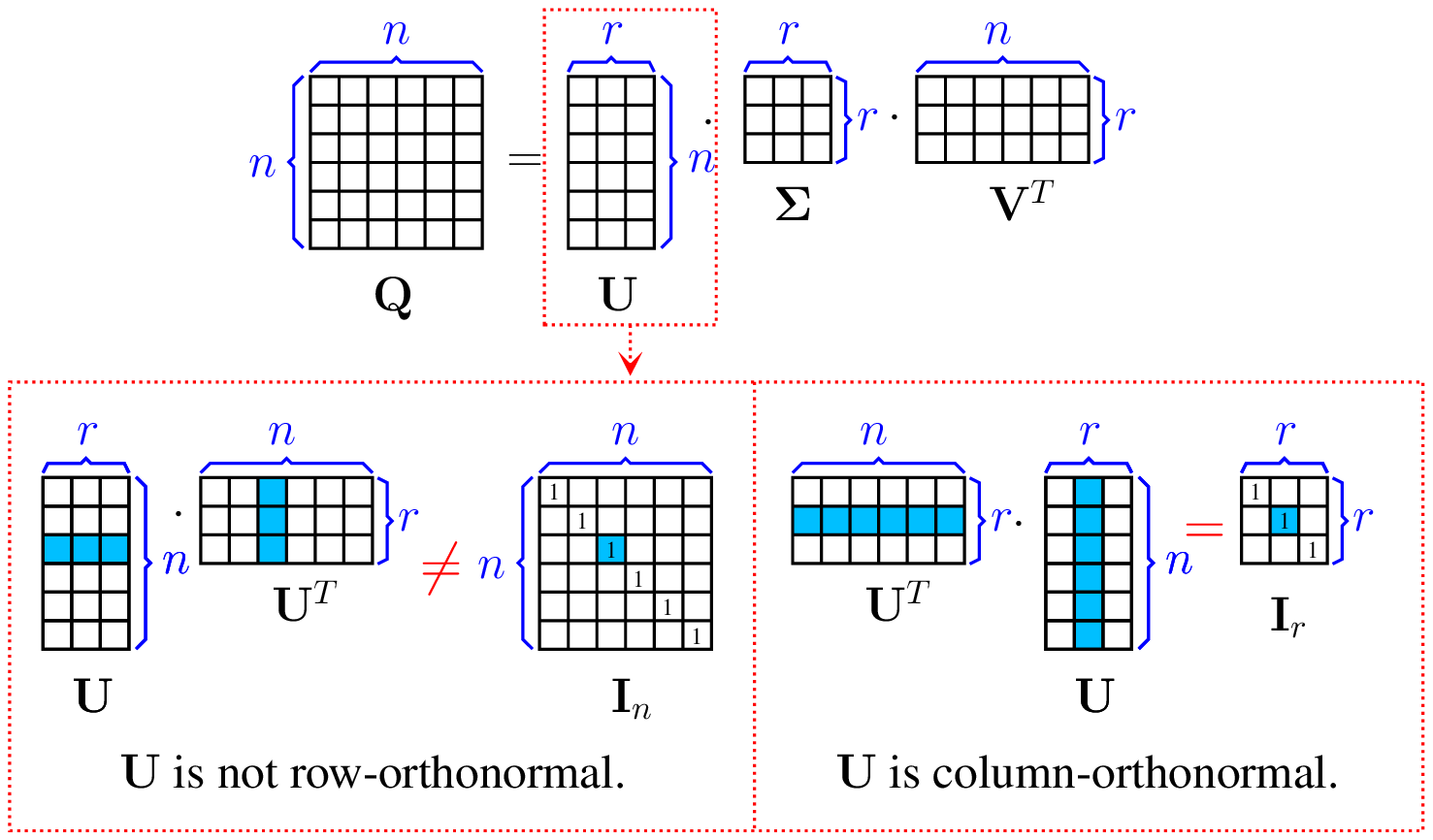}
  \caption{$\mathbf{U}\cdot \mathbf{U}^T \neq \mathbf{I}_n$ whenever $\textrm{rank}(\mathbf{Q})=r<n$} \label{fig:05} 
\end{figure}

\textsc{Step 1}. \
Initially, when links are changed,
the old $\mathbf{Q}$ is updated to new $\tilde{\mathbf{Q}}  = \mathbf{Q} + \mathbf{\Delta Q}$.
By replacing $\mathbf{Q}$ with Eq.\eqref{eq:01}, it follows that
\begin{equation} \label{eq:05}
\tilde{\mathbf{Q}} = \mathbf{U} \cdot \mathbf{\Sigma} \cdot {\mathbf{V}}^T + \mathbf{\Delta Q}.
\end{equation}

\textsc{Step 2}. \
Premultiply by $\mathbf{U}^T$ and postmultiply by $\mathbf{V}$ on both sides of Eq.\eqref{eq:05},
and then apply the property $\mathbf{U}^T \cdot \mathbf{U} = \mathbf{V}^T \cdot \mathbf{V} = \mathbf{I}_{\alpha}$.
It follows that
\begin{equation} \label{eq:06}
\mathbf{U}^T \cdot \tilde{\mathbf{Q}} \cdot {\mathbf{V}}=  \mathbf{\Sigma} +\mathbf{U}^T \cdot \mathbf{\Delta Q} \cdot {\mathbf{V}}.
\end{equation}

\textsc{Step 3}. \
Let $\mathbf{C}$ be the right-hand side of Eq.\eqref{eq:06}. Applying Eq.\eqref{eq:03} to Eq.\eqref{eq:06} yields
\begin{equation} \label{eq:07}
\mathbf{U}^T \cdot \tilde{\mathbf{Q}} \cdot {\mathbf{V}}= \mathbf{U_C} \cdot \mathbf{\Sigma_C} \cdot{\mathbf{V_C}}^T.
\end{equation}

\textsc{Step 4}. \
Li \etal \cite{Li2010} attempted to premultiply by $\mathbf{U}$ and postmultiply by $\mathbf{V}^T$ on both sides of Eq.\eqref{eq:07} first,
and then rested on the assumption of Eq.\eqref{eq:08} to obtain
\begin{equation} \label{eq:09}
\underbrace{\mathbf{U} \cdot \mathbf{U}^T}_{\hspace{10pt}\scalebox{1.2}{$? \hspace{-10pt}=\mathbf{I}_n$}} \cdot \tilde{\mathbf{Q}} \cdot \underbrace{\mathbf{V} \cdot \mathbf{V}^T}_{\hspace{10pt}\scalebox{1.2}{$? \hspace{-10pt}=\mathbf{I}_n$}}  =\underbrace{(\mathbf{U} \cdot  \mathbf{U_C})}_{\triangleq \tilde{\mathbf{U}}} \cdot \underbrace{\mathbf{\Sigma_C}}_{\triangleq \tilde{\mathbf{\Sigma}}} \cdot \underbrace{({\mathbf{V_C}} \cdot {\mathbf{V}})^T}_{\triangleq {\tilde{\mathbf{V}}}^T},
\end{equation}
which is the result of Eq.\eqref{eq:02}.

However, the problem lies in \textsc{Step} 4.
As mentioned before, Eq.\eqref{eq:08} does not hold when $\textrm{rank}(\mathbf{Q})<n$,
which means that $\tilde{\mathbf{Q}} \neq \tilde{\mathbf{U}} \cdot \tilde{\mathbf{\Sigma}} \cdot \tilde{\mathbf{V}}^T$ in Eq.\eqref{eq:09}.
Consequently, updating the old ${\mathbf{U}},{\mathbf{\Sigma}},{\mathbf{V}}$ via Eq.\eqref{eq:02}
may produce an error (up to $\|\mathbf{I}_n - \mathbf{U} \cdot \mathbf{U}^T  \|_2=1$, which is not practically small) for incrementally ``approximating'' ${\mathbf{S}}$.

\begin{example} \label{eg:03}
Recall the old $\mathbf{Q}$ and its SVD in Example \ref{eg:02}.
Suppose there is a new edge insertion,
associated with $\mathbf{\Delta Q}=\left[\begin{smallmatrix}0 & 0 \\ 1 & 0 \end{smallmatrix}\right]$.
\cite{Li2010} first computes auxiliary matrix $\mathbf{C}$ as
\setlength\arraycolsep{1pt}
\[
\mathbf{C} \triangleq  \mathbf{\Sigma} + \mathbf{U}^T \cdot \mathbf{\Delta Q} \cdot \mathbf{V} = [1] + \left[\begin{smallmatrix}1 & 0 \end{smallmatrix}\right] \cdot \left[\begin{smallmatrix}0 & 0 \\ 1 & 0 \end{smallmatrix}\right] \cdot \left[\begin{smallmatrix}0 \\1 \end{smallmatrix}\right] =[1].
\]
Then, the matrix $\mathbf{C}$ is decomposed via Eq.\eqref{eq:03} into
\[
\mathbf{C}= \mathbf{U_C} \cdot \mathbf{\Sigma_C} \cdot{\mathbf{V_C}}^T \textrm{ with } \mathbf{U_C}=\mathbf{\Sigma_C}=\mathbf{V_C}=[1].
\]
Finally, \cite{Li2010} updates the new SVD of $\tilde{\mathbf{Q}}$ via Eq.\eqref{eq:02} as
\[
\tilde{\mathbf{U}}=\mathbf{U}\cdot \mathbf{U_C}=\left[\begin{smallmatrix}1 \\0 \end{smallmatrix}\right], \quad
\tilde{\mathbf{\Sigma}}=\mathbf{\Sigma_C}=[1], \quad
\tilde{{\mathbf{V}}}= {\mathbf{V}} \cdot {\mathbf{V_C}}=\left[\begin{smallmatrix}0 \\1 \end{smallmatrix}\right].
\]

However, one can readily verify that
\[
\tilde{\mathbf{U}} \cdot \tilde{\mathbf{\Sigma}} \cdot {\tilde{{\mathbf{V}}}}^T = \left[\begin{smallmatrix}0 & 1 \\0 & 0 \end{smallmatrix}\right] \neq  \left[\begin{smallmatrix}0 & 1 \\ 1 & 0 \end{smallmatrix}\right] = \mathbf{Q} + \mathbf{\Delta} \mathbf{Q} = \tilde{\mathbf{Q}}.
\]
In comparison, a ``true'' SVD of $\tilde{\mathbf{Q}}$ should be
\[
\tilde{\mathbf{Q}} = \hat{\mathbf{U}} \cdot \hat{\mathbf{\Sigma}} \cdot{\hat{\mathbf{V}}}^T  \textrm{ with } \hat{\mathbf{U}} = \left[\begin{smallmatrix}0 & 1 \\1 & 0\end{smallmatrix}\right], \ \hat{\mathbf{\Sigma}}=\hat{\mathbf{V}}= \left[\begin{smallmatrix}1 & 0 \\0 & 1\end{smallmatrix}\right].
\]
%
Besides, the approximation error is not small in practice
\[{\|\tilde{\mathbf{Q}} - \tilde{\mathbf{U}} \cdot \tilde{\mathbf{\Sigma}} \cdot {\tilde{{\mathbf{V}}}}^T\|}_{2} = {\| \left[\begin{smallmatrix}0 & 1 \\ 1 & 0 \end{smallmatrix}\right] - \left[\begin{smallmatrix}0 & 1 \\0 & 0 \end{smallmatrix}\right] \|}_{2}=1. \quad \qed \]
\end{example}

%
Our analysis suggests that, only when
(i) $\mathbf{Q}$ is full-rank, and
(ii) the SVD of $\mathbf{Q}$ is lossless $(n=\textrm{rank}(\mathbf{Q})=\alpha)$,
Li \etal\!\!'s incremental way \cite{Li2010} can produce the \emph{exact} $\mathbf{S}$, 
but the time complexity of \cite{Li2010}, $O(r^4n^2)$, would become $O(n^6)$,
which is prohibitively expensive.
In practice,
as evidenced by our statistical experiments in Fig.\ref{fig:exp_04} on Stanford Large Network Datasets (SNAP), 
most real graphs are not full-rank,
highlighting our need to devise an efficient method for dynamic SimRank computation.
%
%
\section{Proofs \& Intuitions of Theorems} \label{app:02}
\subsection{Proof of Theorem~\ref{thm:01}} \label{app:02a}
%
%
%
\begin{proof}
We show this by considering the two cases below:
%

(i) If ${{d}_{j}}=0$, then ${{[\mathbf{Q}]}_{j,\star}}=\mathbf{0}$, and the inserted edge $(i,j)$ will update ${{[\mathbf{Q}]}_{j,i}}$ from 0 to 1,
\ie $\mathbf{\Delta Q}={{\mathbf{e}}_{j}}\mathbf{e}_{i}^{T}$.

(ii) If ${{d}_{j}}>0$, then all nonzeros in old ${{[\mathbf{Q}]}_{j,\star}}$ are $\tfrac{1}{d_j}$.
The inserted edge $(i,j)$ will update ${{[\mathbf{Q}]}_{j,\star}}$ via 2 steps:
first, all nonzeros in ${{[\mathbf{Q}]}_{j,\star}}$ are changed from $\tfrac{1}{d_j}$ to $\tfrac{1}{d_j+1}$;
then, the entry ${{[\mathbf{Q}]}_{j,i}}$ is changed from 0 to $\tfrac{1}{d_j+1}$.
\[
{{[\mathbf{\tilde{Q}}]}_{j,\star}}
= \tfrac{{{d}_{j}}}{{{d}_{j}}+1}  {{[\mathbf{Q}]}_{j,\star}}  + \tfrac{1}{{{d}_{j}}+1}  \mathbf{e}_{i}^{T}
=  {{[\mathbf{Q}]}_{j,\star}} + \tfrac{1}{{{d}_{j}}+1} ( \mathbf{e}_{i}^{T} - {{[\mathbf{Q}]}_{j,\star}} )
\]
Since only the $j$-th row of $\mathbf{Q}$ is affected, it follows that
\[
\mathbf{\tilde{Q}} - \mathbf{Q}
= \underbrace{\tfrac{1}{{{d}_{j}}+1} \mathbf{e}_{j}}_{:= \mathbf{u}} \underbrace{( \mathbf{e}_{i}^{T} - {{[\mathbf{Q}]}_{j,\star}} )}_{:= \mathbf{v}^T}
= \mathbf{u}\cdot {\mathbf{v}}^T 
\]

Finally, combining (i) and (ii), Eq.\eqref{eq:18} holds.  \qed
\end{proof}
\subsection{Proof of Theorem~\ref{thm:02}} \label{app:02b}
\begin{proof}
We show this by following the two steps:

(a) We first formulate $\mathbf{\Delta S}$ recursively.
%
%
To describe $\mathbf{\Delta S}$ in terms of the old $\mathbf{Q}$ and $\mathbf{S}$,
we subtract Eq.\eqref{eq:03a} from Eq.\eqref{eq:15},
and apply $\mathbf{\Delta S}=\mathbf{\tilde{S}}-\mathbf{S}$, yielding
\begin{equation}  \label{eq:23}
\scalebox{0.88}{$
\mathbf{\Delta S} = C\cdot \mathbf{\tilde{Q}}\cdot \mathbf{S}\cdot {{\mathbf{\tilde{Q}}}^{T}} + C\cdot \mathbf{\tilde{Q}}\cdot \mathbf{\Delta S}\cdot {{\mathbf{\tilde{Q}}}^{T}}
-C\cdot \mathbf{Q}\cdot \mathbf{S}\cdot {{\mathbf{Q}}^{T}}. $}
\end{equation}
If there are two vectors $\mathbf{u}$ and $\mathbf{v}$ such that
\begin{equation}\label{eq:24}
\mathbf{\tilde{Q}}=\mathbf{Q}+\mathbf{\Delta Q}=\mathbf{Q}+\mathbf{u} \cdot {{\mathbf{v}}^{T}},
\end{equation}
then we can plug Eq.\eqref{eq:24} into the term $C\cdot \mathbf{\tilde{Q}}\cdot \mathbf{S}\cdot {{\mathbf{\tilde{Q}}}^{T}}$ of Eq.\eqref{eq:23},
and simplify the result into
\begin{equation}\label{eq:25}
\mathbf{\Delta S}=C\cdot \mathbf{\tilde{Q}}\cdot \mathbf{\Delta S}\cdot {{\mathbf{\tilde{Q}}}^{T}}+C\cdot \mathbf{T}
\end{equation}
\begin{equation}\label{eq:26}
\textrm { with }\mathbf{T}=\mathbf{u}{{(\mathbf{QSv})}^{T}}+(\mathbf{QSv}){{\mathbf{u}}^{T}}+({{\mathbf{v}}^{{T}}}\mathbf{Sv})\mathbf{u}{{\mathbf{u}}^{T}}.
\end{equation}
%
We can verify that $\mathbf{T}$ is a symmetric matrix ($\mathbf{T}={{\mathbf{T}}^{T}}$).
Moreover, we note that $\mathbf{T}$ is the sum of two rank-one matrices.
This can be verified by letting
\[\mathbf{z}\triangleq \mathbf{S}\cdot \mathbf{v},\ \mathbf{y}\triangleq \mathbf{Q}\cdot \mathbf{z},\ \lambda \triangleq {{\mathbf{v}}^{T}}\cdot \mathbf{z}.\]
Then, using the auxiliary vectors $\mathbf{z}, \mathbf{y}$ and the scalar $\lambda$,
we can simplify Eq.\eqref{eq:26} into
\begin{eqnarray}
  \mathbf{T} 
 &=&  \mathbf{u}\cdot {{\mathbf{w}}^{T}}+\mathbf{w}\cdot {{\mathbf{u}}^{T}}, \ \ \textrm{ with } \mathbf{w}= \mathbf{y}+\tfrac{\lambda }{2}\mathbf{u}.  \label{eq:27}
\end{eqnarray}

(b) We next convert the recursive form of $\mathbf{\Delta S}$ into the series form.
One can readily verify that 
\begin{equation}\label{eq:28}
\mathbf{X}=\mathbf{A}\cdot \mathbf{X}\cdot \mathbf{B}+\mathbf{C}\quad \Leftrightarrow \quad \mathbf{X}=\sum\nolimits_{k=0}^{\infty }{{{\mathbf{A}}^{k}}\cdot \mathbf{C}\cdot {{\mathbf{B}}^{k}}}
\end{equation}
Thus, based on Eq.\eqref{eq:28},
the recursive definition of $\mathbf{\Delta S}$ in Eq.\eqref{eq:25} naturally leads itself to the series form:
\[\mathbf{\Delta S}=\sum\nolimits_{k=0}^{\infty }{{{C}^{k+1}}\cdot {{{\mathbf{\tilde{Q}}}}^{k}}\cdot \mathbf{T}\cdot {{({{{\mathbf{\tilde{Q}}}}^{T}})}^{k}}}.\]
Combining this with Eq.\eqref{eq:27} yields
\begin{eqnarray*}
\mathbf{\Delta S} &=& \sum\nolimits_{k=0}^{\infty }{{{C}^{k+1}}\cdot {{{\mathbf{\tilde{Q}}}}^{k}}\cdot \left( \mathbf{u}\cdot {{\mathbf{w}}^{T}}+\mathbf{w}\cdot {{\mathbf{u}}^{T}} \right)\cdot {{({{{\mathbf{\tilde{Q}}}}^{T}})}^{k}}} \\
 &=&  \mathbf{M}+{{\mathbf{M}}^{T}}  \ \textrm{ with } \mathbf{M}\textrm{ being defined in Eq.}\eqref{eq:16}.
\end{eqnarray*}
By Eq.\eqref{eq:28},
the series form of $\mathbf{M}$ in Eq.\eqref{eq:16} satisfies the rank-one Sylvester recursive form of Eq.\eqref{eq:14}. \qed
\end{proof}
\subsection{Proof of Theorem~\ref{thm:03}} \label{app:02c}
\begin{proof}
We divide the proof into the following two cases:

(i) When ${{d}_{j}}=0$,
according to Eq.\eqref{eq:18} in Theorem~\ref{thm:01},
$\mathbf{v}={{\mathbf{e}}_{i}}, \ \mathbf{u}={{\mathbf{e}}_{j}}$.
Plugging them into Eq.\eqref{eq:20} gets
\[
\mathbf{z} ={{[\mathbf{S}]}_{\star,i}}, \quad
\mathbf{y} =\mathbf{Q}\cdot {{[\mathbf{S}]}_{\star,i}}, \quad
\lambda = {{[\mathbf{S}]}_{i,i}}.
\]
Thus, applying $\mathbf{w}=\mathbf{y}+\tfrac{\lambda }{2}\mathbf{u}$ in Theorem \ref{thm:02},
we have
\[
\mathbf{w}=\mathbf{Q}\cdot {{[\mathbf{S}]}_{\star,i}}+\tfrac{1}{2}{{[\mathbf{S}]}_{i,i}}\cdot {{\mathbf{e}}_{j}}.
\]
Coupling this with Eq.\eqref{eq:16}, $\mathbf{u}={{\mathbf{e}}_{j}}$, and Theorem~\ref{thm:02} completes the proof of the case ${{d}_{j}}=0$ for Eq.\eqref{eq:29aa}.

(ii) When ${{d}_{j}}>0$, Eq.\eqref{eq:18} in Theorem \ref{thm:01} implies that
\begin{equation}\label{eq:32}
\mathbf{v}={{\mathbf{e}}_{i}}-[\mathbf{Q}]_{j,\star}^{T}, \quad \mathbf{u}=\tfrac{1}{{{d}_{j}}+1}\cdot {{\mathbf{e}}_{j}}.
\end{equation}
Substituting these back into Eq.\eqref{eq:20} yields
\begin{eqnarray*}
\mathbf{z}&=&
{{[\mathbf{S}]}_{\star,i}}-\mathbf{S}\cdot [\mathbf{Q}]_{j,\star}^{T}, \quad
\mathbf{y}  =  \mathbf{Q}\cdot {{[\mathbf{S}]}_{\star,i}}-\mathbf{Q}\cdot \mathbf{S}\cdot [\mathbf{Q}]_{j,\star}^{T}, \\[3pt]
\lambda  
 &=& {{[\mathbf{S}]}_{i,i}}-2\cdot {{[\mathbf{Q}]}_{j,\star}}\cdot {{[\mathbf{S}]}_{\star,i}}+{{[\mathbf{Q}]}_{j,\star}}\cdot \mathbf{S}\cdot [\mathbf{Q}]_{j,\star}^{T}.
\end{eqnarray*}
To simplify $\mathbf{Q}\cdot \mathbf{S}\cdot [\mathbf{Q}]_{j,\star}^{T}$ in $\mathbf{y}$,
and ${{[\mathbf{Q}]}_{j,\star}}\cdot \mathbf{S}\cdot [\mathbf{Q}]_{j,\star}^{T}$ in $\lambda$,
we postmultiply both sides of Eq.\eqref{eq:03a} by ${{\mathbf{e}}_{j}}$ to obtain
\begin{equation} \label{eq:33}
\mathbf{Q}\cdot \mathbf{S}\cdot [\mathbf{Q}]_{j,\star}^{T}=\tfrac{1}{C}\cdot ( {{[\mathbf{S}]}_{\star,j}}-(1-C)\cdot {{\mathbf{e}}_{j}} ).
\end{equation}
We also premultiply both sides of Eq.\eqref{eq:33} by $\mathbf{e}_j^T$ to get
\begin{equation} \label{eq:34}
{[\mathbf{Q}]}_{j,\star} \cdot \mathbf{S}\cdot [\mathbf{Q}]_{j,\star}^{T}=\tfrac{1}{C}\cdot ( {{[\mathbf{S}]}_{j,j}}-1) + 1.
\end{equation}
Plugging Eqs.\eqref{eq:33} and \eqref{eq:34} into $\mathbf{y}$ and $\lambda$, respectively,
and then putting $\mathbf{y}$ and $\lambda$ into $\mathbf{w}=\mathbf{y}+\tfrac{\lambda }{2}\mathbf{u}$ produce
%
\begin{equation*}
\mathbf{w}=  \mathbf{Q}\cdot {{[\mathbf{S}]}_{\star,i}}-\tfrac{1}{C}\cdot {{[\mathbf{S}]}_{\star,j}}+ ( \tfrac{1}{C}+\tfrac{\lambda }{2 ( {{d}_{j}}+1 )}-1 )\cdot {{\mathbf{e}}_{j}},
\end{equation*}
where
$\lambda = {{[\mathbf{S}]}_{i,i}}+\tfrac{1}{C} \cdot {[\mathbf{S}]}_{j,j}-2\cdot {{[\mathbf{Q}]}_{j,\star}}\cdot {{[\mathbf{S}]}_{\star,i}} - \tfrac{1}{C} +1$.

Combining this with Eqs.\eqref{eq:16} and \eqref{eq:32} shows the case ${{d}_{j}}>0$ for Eq.\eqref{eq:29bb}.

Finally, taking (i) and (ii)
together with Theorem~\ref{thm:02}
completes the entire proof. \qed
\end{proof}
\subsection{Proof of Theorem~\ref{thm:09}} \label{app:02d}
\begin{proof}
We prove this by considering two cases:

  (i) If ${{d}_{j}}=1$, then after the edge $(i,j)$ is deleted, ${{[\mathbf{Q}]}_{j,i}}$ will change from 1 to 0,
\ie
\[
\mathbf{\Delta Q}=\mathbf{u} \cdot \mathbf{v}^T \quad \textrm{ with }  \mathbf{u} = {{\mathbf{e}}_{j}} \textrm{ and } \mathbf{v} = - \mathbf{e}_{i}.
\]
According to Eq.\eqref{eq:20} in Theorem~\ref{thm:02}, we have
\[
\mathbf{z} =-{{[\mathbf{S}]}_{\star,i}}, \quad
\mathbf{y} =-\mathbf{Q}\cdot {{[\mathbf{S}]}_{\star,i}}, \quad
\lambda = {{[\mathbf{S}]}_{i,i}}.
\]
Thus, plugging them into $\mathbf{w}=\mathbf{y}+\tfrac{\lambda }{2}\mathbf{u}$ produces
\[
\mathbf{w}=-\mathbf{Q}\cdot {{[\mathbf{S}]}_{\star,i}}+\tfrac{1}{2}{{[\mathbf{S}]}_{i,i}}\cdot {{\mathbf{e}}_{j}}.
\]
Combining this with Theorem~\ref{thm:02} completes the proof of the case when $d_j= 1$.

(ii) If ${{d}_{j}}>1$, then all nonzeros in old ${{[\mathbf{Q}]}_{j,\star}}$ are $\tfrac{1}{d_j}$.
The deleted edge $(i,j)$ will update ${{[\mathbf{Q}]}_{j,\star}}$ via 2 steps:
first, all nonzeros in ${{[\mathbf{Q}]}_{j,\star}}$ are changed from $\tfrac{1}{d_j}$ to $\tfrac{1}{d_j-1}$;
then, the entry ${{[\mathbf{Q}]}_{j,i}}$ is changed from $\tfrac{1}{d_j}$ to 0.
\[
{{[\mathbf{\tilde{Q}}]}_{j,\star}}
= \tfrac{{{d}_{j}}}{{{d}_{j}}-1} ( {{[\mathbf{Q}]}_{j,\star}}  - \tfrac{1}{{{d}_{j}}}  \mathbf{e}_{i}^{T} )
=  {{[\mathbf{Q}]}_{j,\star}} + \tfrac{1}{{{d}_{j}}-1} ( {{[\mathbf{Q}]}_{j,\star}} - \mathbf{e}_{i}^{T} )
\]
Since only the $j$-th row of $\mathbf{Q}$ is affected, it follows that
\[
\mathbf{\tilde{Q}} - \mathbf{Q}
= \underbrace{\tfrac{1}{{{d}_{j}}-1} \mathbf{e}_{j}}_{:= \mathbf{u}} \underbrace{( {{[\mathbf{Q}]}_{j,\star}} - \mathbf{e}_{i}^{T} )}_{:= \mathbf{v}^T}
= \mathbf{u}\cdot {\mathbf{v}}^T 
\]
By virtue of Eq.\eqref{eq:20} in Theorem~\ref{thm:02}, we have
\begin{eqnarray*}
\mathbf{z} &=& \mathbf{S}\cdot \mathbf{v}=\mathbf{S}\cdot \left( [\mathbf{Q}]_{j,\star}^{T} - {{\mathbf{e}}_{i}} \right)
= \mathbf{S}\cdot [\mathbf{Q}]_{j,\star}^{T} - {{[\mathbf{S}]}_{\star,i}}, \\
\mathbf{y} &=& \mathbf{Q}\cdot \mathbf{z} = \mathbf{Q}\cdot \mathbf{S}\cdot [\mathbf{Q}]_{j,\star}^{T} - \mathbf{Q}\cdot {{[\mathbf{S}]}_{\star,i}}
= \textrm{\{using Eq.\eqref{eq:33}\}} \\
&=& \tfrac{1}{C}\cdot ( {{[\mathbf{S}]}_{\star,j}}-(1-C)\cdot {{\mathbf{e}}_{j}} ) - \mathbf{Q}\cdot {{[\mathbf{S}]}_{\star,i}}. \\
\lambda &=& {{\mathbf{v}}^{T}}\cdot \mathbf{z} = {{( [\mathbf{Q}]_{j,\star} - {{\mathbf{e}}_{i}^{T}} )} }\cdot ( \mathbf{S}\cdot [\mathbf{Q}]_{j,\star}^{T} - {{[\mathbf{S}]}_{\star,i}} ) \\
 &=& {{[\mathbf{S}]}_{i,i}}-2\cdot {{[\mathbf{Q}]}_{j,\star}}\cdot {{[\mathbf{S}]}_{\star,i}}+{{[\mathbf{Q}]}_{j,\star}}\cdot \mathbf{S}\cdot [\mathbf{Q}]_{j,\star}^{T}. \\
 &=& \textrm{\{using Eq.\eqref{eq:34}\}}  \\
 &=& {{[\mathbf{S}]}_{i,i}}-2\cdot {{[\mathbf{Q}]}_{j,\star}}\cdot {{[\mathbf{S}]}_{\star,i}}+\tfrac{1}{C}\cdot ( {{[\mathbf{S}]}_{j,j}}-1) + 1.
\end{eqnarray*}
Hence, substituting them into $\mathbf{w}=\mathbf{y}+\tfrac{\lambda }{2}\mathbf{u}$ yields
\[
\mathbf{w}=  -\mathbf{Q}\cdot {{[\mathbf{S}]}_{\star,i}}+\tfrac{1}{C}\cdot {{[\mathbf{S}]}_{\star,j}}+ ( 1-\tfrac{1}{C}+\tfrac{\lambda }{2 ( {{d}_{j}}-1 )} )\cdot {{\mathbf{e}}_{j}}.
\]
Combining this with Theorem~\ref{thm:02} completes the proof of the case when $d_j> 1$.

Finally, coupling (i) and (ii) proves Theorem~\ref{thm:09}.  \qed
\end{proof}
\subsection{Proof and Intuition of Theorem~\ref{thm:04}} \label{app:02e}
\begin{proof}
We only show the edge insertion case ${{d}_{j}}>0$, due to space limits.
The proofs of other cases are similar.

For $k=0$,
it follows from Eq.\eqref{eq:29c} that ${{[{{\mathbf{M}}_{0}}]}_{a,b}}={{[{{\mathbf{e}}_{j}}]}_{a}} {{[\bm{\gamma }]}_{b}}$.
Thus, $\forall (a,b)\notin {{\mathsf{\mathcal{A}}}_{0}}\times {{\mathsf{\mathcal{B}}}_{0}}$,
there are two cases:
(i) $a\ne j$, or
(ii) $a=j$, $b\in {{\mathsf{\mathcal{F}}}_{1}}^{C}\cap {{\mathsf{\mathcal{F}}}_{2}}^{C}$, and $b\ne j$.

For case (i), ${{[{{\mathbf{e}}_{j}}]}_{a}}=0$ for $a\ne j$.
Thus, ${{[{{\mathbf{M}}_{0}}]}_{a,b}}=0$.
For case (ii),
${{[{{\mathbf{e}}_{j}}]}_{a}}=1$ for $a=j$.
Thus, ${{[{{\mathbf{M}}_{0}}]}_{a,b}}={{[\bm{\gamma }]}_{b}}$,
where ${{[\bm{\gamma }]}_{b}}$ is the linear combinations of the 3 terms:
${{[\mathbf{Q}]}_{b,\star}}\cdot {{[\mathbf{S}]}_{\star,i}}$, ${{[\mathbf{S}]}_{b,j}}$, and ${{[{{\mathbf{e}}_{j}}]}_{b}}$,
according to the case of ${{d}_{j}}>0$ in Eq.\eqref{eq:29bb}.

Next, our goal is to show the 3 terms are all 0s.
(a) For $b\notin {{\mathsf{\mathcal{F}}}_{1}}$,
by definition in Eq.\eqref{eq:39},
$b\in \mathsf{\mathcal{O}}(y)$ for $\forall y$,
we have ${{[\mathbf{S}]}_{i,y}}=0$.
Due to symmetry,
$b\in \mathsf{\mathcal{O}}(y)\Leftrightarrow y\in \mathsf{\mathcal{I}}(b)$,
which implies that ${{[\mathbf{S}]}_{i,y}}=0$ for $\forall y\in \mathsf{\mathcal{I}}(b)$.
\footnote{Herein, we denote by $\mathcal{I}(a)$ the in-neighbor set of node $a$.}
Thus, ${{[\mathbf{Q}]}_{b,\star}}\cdot {{[\mathbf{S}]}_{\star,i}}=\frac{1}{\mathsf{\mathcal{I}}(b)}\sum_{x\in \mathsf{\mathcal{I}}(b)}^{{}}{{{[\mathbf{S}]}_{x,i}}}=0$.
(b) For $b\notin {{\mathsf{\mathcal{F}}}_{2}}$,
it follows from the case ${{d}_{j}}>0$ in Eq.\eqref{eq:40} that ${{[\mathbf{S}]}_{j,b}}=0$.
Hence, by $\mathbf{S}$ symmetry,
${{[\mathbf{S}]}_{b,j}}={{[\mathbf{S}]}_{j,b}}=0$.
(c) ${{[{{\mathbf{e}}_{j}}]}_{b}}=0$ since $b\ne j$.

Taking (a)--(c) together,
it follows that ${{[{{\mathbf{M}}_{0}}]}_{a,b}}=0$,
which completes the proof for the case $k=0$.

For $k>0$, one can readily prove that the $k$-th iterative ${{\mathbf{M}}_{k}}$ in Line~\ref{ln:a01-17} of Algorithm~\ref{alg:01} is the first $k$-th partial sum of $\mathbf{M}$ in Eq.\eqref{eq:29c}.
Thus, ${{\mathbf{M}}_{k+1}}$ can be derived from ${{\mathbf{M}}_{k}}$ as follows:
\[
\mathbf{M}_{k} = C \cdot \tilde{\mathbf{Q}} \cdot \mathbf{M}_{k-1} \cdot \tilde{\mathbf{Q}}^T + C \cdot \mathbf{e}_j \cdot \bm{\gamma}^T.
\]
Thus, the $(a,b)$-entry form of the above equation is
\[ \scalebox{0.92}{$
{[\mathbf{M}_{k}]}_{a,b} = \tfrac{C}{|\tilde{{\cal I}}(a)||\tilde{{\cal I}}(b)|}   \sum\nolimits_{x \in \tilde{{\cal I}}(a)}  \sum\nolimits_{y \in \tilde{{\cal I}}(b)} {[\mathbf{M}_{k-1}]}_{x,y} + C \cdot {[\mathbf{e}_j]}_{a} \cdot {[\bm{\gamma}]}_{b}.
$}\]
To show that ${[\mathbf{M}_{k}]}_{a,b}=0$ for $(a,b)\notin {{\mathsf{\mathcal{A}}}_{0}}\times {{\mathsf{\mathcal{B}}}_{0}} \cup {{\mathsf{\mathcal{A}}}_{k}}\times {{\mathsf{\mathcal{B}}}_{k}}$,
we follow the 2 steps:
(i) For $(a,b)\notin {{\mathsf{\mathcal{A}}}_{0}}\times {{\mathsf{\mathcal{B}}}_{0}}$,
as proved in the case $k=0$, the term $C \cdot {{[{{\mathbf{e}}_{j}}]}_{a}} {{[\bm{\gamma }]}_{b}}$ in the above equation is obviously 0.
(ii) For $(a,b)\notin {{\mathsf{\mathcal{A}}}_{k}}\times {{\mathsf{\mathcal{B}}}_{k}}$,
by virtue of Eq.\eqref{eq:41},
$a \in \tilde{\cal O}(x), b \in \tilde{\cal O}(y)$, for $\forall x,y$,
we have ${[\mathbf{M}_{k-1}]}_{x,y}=0$.
Hence, by symmetry, it follows that
$x \in \tilde{\cal I}(a), y \in \tilde{\cal I}(b)$, ${[\mathbf{M}_{k-1}]}_{x,y}=0$.

Taking (i) and (ii) together,
we can conclude that
${[\mathbf{M}_{k}]}_{a,b} = 0$
for $(a,b)\notin {{\mathsf{\mathcal{A}}}_{0}}\times {{\mathsf{\mathcal{B}}}_{0}} \cup {{\mathsf{\mathcal{A}}}_{k}}\times {{\mathsf{\mathcal{B}}}_{k}}$. \qed
\end{proof}

Intuitively, ${\cal F}_1$ in Eq.\eqref{eq:39} captures the nodes ``$\blacktriangle$'' in \eqref{eq:39a}.
To be specific,
${\cal F}_1$ can be obtained via 2 phases:
(i) For the given node $i$,
we first build an intermediate set ${\cal T}:=\{y  |  {{[\mathbf{S}]}_{i,y}}\ne 0\}$,
which consists of nodes ``$\star$'' in \eqref{eq:39a}.
(ii) For each node $x \in {\cal T}$,
we then find all out-neighbors of $x$ in $G$, which produces ${\cal F}_1$,
\ie ${\cal F}_1 = \bigcup_{x\in {\cal T}}{{\cal O}(x)}$.
Analogously, the set ${\cal F}_2$ in Eq.\eqref{eq:40},
in the case of $d_j>0$, consists of the nodes ``$\star$'' depicted in~\eqref{eq:39b}.
When $d_j=0$, ${\cal F}_2 = \varnothing $ as the term ${{[\mathbf{S}]}_{\star,i}}$ is not in the expression of $\bm{\gamma }$ in Eq.\eqref{eq:29aa}, 
in contrast to the case $d_j>0$.

After obtaining ${\cal F}_1$ and ${\cal F}_2$,
we can readily find ${{\mathsf{\mathcal{A}}}_{0}}\times {{\mathsf{\mathcal{B}}}_{0}}$, according to Eq.\eqref{eq:41}.
For $k>0$, to iteratively derive the node-pair set ${{\mathsf{\mathcal{A}}}_{k}}\times {{\mathsf{\mathcal{B}}}_{k}}$,
we take the following two steps:
(i) we first construct a node-pair set ${\cal T}_1 \times {\cal T}_2 :=\{(x,y) | {[\mathbf{M}_{k-1}]}_{x,y} \neq 0\}$.
(ii) For every node $x \in {\cal T}_1$ (\Resp $y \in {\cal T}_2$),
we then find all out-neighbors of $x$ (\Resp $y$) in $G \cup \{(i,j)\}$, which yields ${\cal A}_k$ (\Resp ${\cal B}_k$),
\ie ${\cal A}_k = \bigcup_{x\in {\cal T}_1}{\tilde{{\cal O}}(x)}$ and $ {\cal B}_k = \bigcup_{y\in {\cal T}_2}{\tilde{{\cal O}}(y)}$.

The node selectivity of Theorem \ref{thm:04} hinges on $\mathbf{\Delta S}$ sparsity.
Since real graphs are constantly updated with \emph{minor} changes,
$\mathbf{\Delta S}$ is often \emph{sparse} in general.
Hence, many node-pairs with zero scores in $\mathbf{\Delta S}$ can be discarded. 
As demonstrated by our experiments in Fig.\ref{fig:exp_07},
76.3\% paper-pairs on \DBLP~can be pruned,
significantly reducing unnecessary similarity recomputations. 
\section{Examples} \label{app:03}
\subsection{Li \etal\!\!'s SVD incremental approach} \label{app:03a}
\begin{example} \label{eg:01}
Figure~\ref{fig:01} depicts a citation graph $G$, a tiny fraction of \DBLP,
where each node is a paper, and an edge represents a reference from one paper to another.
Suppose $G$ is updated by adding an edge $(i,j)$, denoted by $\Delta G$ (see the dash arrow).
Using the damping factor $C=0.8$,
we would like to compute SimRank scores in the new graph $G \cup \Delta G$.

The results are compared in the table of Figure~\ref{fig:01},
where Column `$\textsf{sim}_\textsf{Li et al.}$' denotes the approximation of SimRank scores returned by Li \etal\!\!'s Algorithm~3~\cite{Li2010},
and Column `$\textsf{sim}_\textsf{true}$' denotes the ``true'' SimRank scores returned by a batch algorithm \cite{Fujiwara2013} that runs in $G \cup \Delta G$ from scratch.
%
%
It can be noticed that for some node-pairs (not highlighted in gray),
the similarities obtained by Li \etal\!\!'s incremental method are different from the ``true'' SimRank scores
even if lossless SVD is used
\footnote{A \emph{rank-$\alpha$ SVD} of the matrix $\mathbf{X} \in \mathbb{R}^{n \times n}$ is a factorization of the form
$\mathbf{X}_{\alpha} = \mathbf{U} \cdot \mathbf{\Sigma} \cdot {\mathbf{V}}^T$,
where $\mathbf{U},\mathbf{V} \in \mathbb{R}^{n \times \alpha}$ are column-orthonormal matrices, and
$\mathbf{\Sigma} \in \mathbb{R}^{\alpha \times \alpha}$ is a diagonal matrix,
$\alpha$ is called \emph{the target rank} of the SVD, as specified by the user.

If $\alpha=\textrm{rank}(\mathbf{X})$,
then $\mathbf{X}_{\alpha}=\mathbf{X}$,
and we call it the \emph{lossless SVD}.

If $\alpha<\textrm{rank}(\mathbf{X})$,
then ${\|\mathbf{X}-\mathbf{X}_{\alpha}\|}_{2}$ gives the least square estimate error,
and we call it the \emph{low-rank SVD}.
}
during the process of updating ${\mathbf{U}},{\mathbf{\Sigma}},{{\mathbf{V}}}^T$.
This suggests that Li~\etal\!\!'s incremental approach~\cite{Li2010} is inherently \emph{approximate}.
In fact, 
their incremental strategy would neglect some useful eigen-information whenever $\textrm{rank}(\mathbf{Q})<n$.

We also notice that the target rank $r$ for the SVD of the matrix $\mathbf{C}$
\footnote{As defined in \cite{Li2010}, $r$ is the target rank for the SVD of the auxiliary matrix $\mathbf{C} \triangleq \mathbf{\Sigma} + \mathbf{U}^T \cdot \mathbf{\Delta Q} \cdot \mathbf{V}$,
where $\mathbf{\Delta Q}$ is the changes to $\mathbf{Q}$ for link updates.}
is not always negligibly smaller than~$n$.
For example, in Column `$\textsf{sim}_\textsf{Li et al.}$' of Figure~\ref{fig:01},
$r$ is chosen to be $ \textrm{rank}(\mathbf{C})=9$ to get a \emph{lossless} SVD of $\mathbf{C}$.
Although $r=9$ appears not negligibly smaller than $n=15$,
the accuracy of `$\textsf{sim}_\textsf{Li et al.}$' is still undesirable as compared with `$\textsf{sim}_\textsf{true}$',
not to mention using $r<9$. \qed
\end{example}

Example~\ref{eg:01} implies that Li \etal\!\!'s incremental approach~\cite{Li2010} is approximate and may produce high computational overheads
since $r$ is not always much smaller.
%
%
\subsection{Example of Theorem~\ref{thm:01}} \label{app:03b}
\begin{example} \label{eg:04a}
Recall the digraph $G$ in Fig.~\ref{fig:01},
and the edge $(i,j)$ to be inserted into $G$.
Notice that, in the old $G$, $d_j=2>0$ and
\vspace{-15pt} \[{[\mathbf{Q}]}_{j,\star}=\kbordermatrix{
 \hspace*{-0.5em} &   &        &   &         (h)   &     &   &  (k)            &   &        &   \\
 \hspace*{-0.5em} & 0 & \cdots & 0 & \tfrac{1}{2} &   0    & 0    &  \tfrac{1}{2}  & 0 & \cdots & 0 \\
} \in \mathbb{R}^{1 \times 15}.\]
According to Theorem \ref{thm:01},
the change $\mathbf{\Delta Q}$ is a $15 \times 15$ rank-one matrix,
and can be decomposed as
$\mathbf{u}\cdot \mathbf{v}^T$  with

\scalebox{0.85}{$  \mathbf{u}= \tfrac{1}{{{d}_{j}}+1}{{\mathbf{e}}_{j}}=\tfrac{1}{3}{{\mathbf{e}}_{j}}
=\kbordermatrix{
 \hspace*{-0.5em} &   &        &   &     (j)      &   &        &   \\
 \hspace*{-0.5em} & 0 & \cdots & 0 & \tfrac{1}{3} & 0 & \cdots & 0 \\
}{}^T \in \mathbb{R}^{15 \times 1}, $} \\[3pt]
\scalebox{0.85}{$  \mathbf{v}= {{\mathbf{e}}_{i}}-{{[\mathbf{Q}]}_{j,\star}^T}
=\kbordermatrix{
 \hspace*{-0.5em} &   &        &   &         (h)   &  (i)   & (j)  &  (k)            &   &        &   \\
 \hspace*{-0.5em} & 0 & \cdots & 0 & -\tfrac{1}{2} &   1    & 0    &  -\tfrac{1}{2}  & 0 & \cdots & 0 \\
}{}^T \in \mathbb{R}^{15 \times 1}. $}  \qed
\end{example}
\subsection{Example of Algorithm~\ref{alg:01}} \label{app:03c}
\begin{example} \label{eg:04}
Consider the old digraph $G$ and $\mathbf{S}$ in Fig.~\ref{fig:01}.
When the new edge $(i,j)$ is inserted to $G$,
{\IncUSRone} computes the new $\tilde{\mathbf{S}}$ as follows,
whose results are partially depicted in Column `$\textsf{sim}_\textsf{true}$' of Fig.~\ref{fig:01}.

Given the following information from the old $\mathbf{S}$:

\scalebox{0.8}{$  {[\mathbf{S}]}_{\star,i}
=\kbordermatrix{
 \hspace*{-0.5em} &   &        &   &  (f)   &  (g)   & (h)  &  (i)    &   (j)  &   &        &   \\
 \hspace*{-0.5em} & 0, & \cdots, & 0, &  0.246, &   0,    & 0,    &  0.590,  &  0.310, & 0, & \cdots, & 0 \\
}{}^T \in \mathbb{R}^{15 \times 1}, $}

\scalebox{0.8}{$  {[\mathbf{S}]}_{\star,j}
=\kbordermatrix{
 \hspace*{-0.5em} &   &        &   &  (f)   &  (g)   & (h)  &  (i)    &   (j)  &   &        &   \\
 \hspace*{-0.5em} & 0, & \cdots, & 0, &  0.246, &   0,    & 0,    &  0.310,  &  0.510, & 0, & \cdots, & 0 \\
}{}^T \in \mathbb{R}^{15 \times 1}, $} \\[1pt]

\IncUSRone~first computes $\mathbf{w}$ and $\lambda$ via lines \ref{ln:a01-03}--\ref{ln:a01-04}:\\[-15pt]
\begin{eqnarray*}
 {\mathbf{w}}
&=& \kbordermatrix{
 \hspace*{-0.5em} &  (a)   &  (b)   &    &         &   \\
 \hspace*{-0.5em} & 0.104, & 0.139, & 0, & \cdots, & 0 \\
}{}^T \in \mathbb{R}^{15 \times 1}, \\
\lambda &=& 0.590+\tfrac{1}{0.8} \times 0.510 -2 \times 0 - \tfrac{1}{0.8} +1 = 0.978.
\end{eqnarray*}

Since $d_j=2$,
the vectors $\mathbf{u}$ and $\mathbf{v}$ for the rank-one decomposition of $\mathbf{\Delta Q}$ can be computed via line~\ref{ln:a01-07a}.
Their results are depicted in Example \ref{eg:04a}.

Next, $\bm \gamma$ can be obtained from $\mathbf{w}$ and $\lambda$ via line~\ref{ln:a01-08}:
\begin{eqnarray*}
&&{\bm\gamma} = \tfrac{1}{(2+1)} \big( \mathbf{w}-\tfrac{1}{0.8}  {{[\mathbf{S}]}_{\star,j}}+( \tfrac{\lambda }{2 \times ( 2+1 )}+ \tfrac{1}{0.8}-1 )  {{\mathbf{e}}_{j}} \big) \\
&=& \scalebox{0.72}{$ \kbordermatrix{
 \hspace*{-0.5em} &  (a)   &  (b)   &    &    &    &   (f)   &    &     &  (i)    &  (j)    &    &         &   \\
 \hspace*{-0.5em} & 0.035, & 0.046, & 0, & 0, & 0, & -0.086  & 0, &  0, & -0.129, & -0.075, & 0, & \cdots, & 0 \\
}{}^T \in \mathbb{R}^{15 \times 1}$}
\end{eqnarray*}

In light of $\bm \gamma$, $\mathbf{M}_k$ can be computed via lines \ref{ln:a01-13}--\ref{ln:a01-17}.
After $K=10$ iterations, $\mathbf{M}_{K}$ can be derived as follows:
\[
\scalebox{0.6}{$
\begin{blockarray}{ccccc|cccc|c|c}
       & (a)     & (b)      & (c)& (d)      & (e)       & (f)    & \cdots      & (i)          & (j)     & (k)\cdots(o)  \\[3pt]
 \begin{block}{c(cccc|cccc|c|c)}
   (a) & -0.005  & -0.009   &  0 & 0.009    &           &        &             &              & -0.009  &               \\
   (b) & -0.004  & -0.006   &  0 & 0.006    &           &        & \Big{$0$}   &     & -0.007  & \Big{$0$}     \\
   (c) & 0       &  0       &  0 & 0        &           &        &             &              &     0   &               \\
   (d) & -0.002  & -0.002   &  0 &  -0.005  &           &        &             &              &     0   &               \\\cline{1-11}
\vdots &         & \Big{$0$}&    &          &           &        & \Big{$0$}   &              &\Big{$0$}& \Big{$0$}     \\
   (i) &         &          &    &          &           &        &             &              &         &               \\\cline{1-11}
   (j) &  0.028  &  0.037   &  0 &  0       &           & -0.068 &             & -0.104       & -0.060  &               \\\cline{1-11}
\vdots &         & \Big{$0$}&    &          &           &        & \Big{$0$}   &              &\Big{$0$}& \Big{$0$}     \\
   (o) &         &          &    &          &           &        &             &              &         &               \\
 \end{block}
\end{blockarray}$}\]

\vspace{-12pt} Finally, using $\mathbf{M}_{K}$ and the old $\mathbf{S}$, the new $\tilde{\mathbf{S}}$ is obtained via line \ref{ln:a01-18},
as partly shown in Column `$\textsf{sim}_\textsf{true}$' of Fig.~\ref{fig:01}. \qed
\end{example}
\subsection{Example of Theorem~\ref{thm:04}} \label{app:03d}
\begin{example} \label{eg:05}
Recall Example~\ref{eg:04} and the old graph $G$ in Fig.~\ref{fig:01}.
When edge $(i,j)$ is inserted to $G$,
according to Theorem~\ref{thm:04},
${\cal F}_1 =\{a,b\}, \ {\cal F}_2=\{f,i,j\}, \ {\cal A}_0 \times {\cal B}_0=\{j\} \times \{a,b,f,i,j\}$.
Hence, instead of computing the entire vector $\bm \gamma$ in Eqs.\eqref{eq:29aa} and \eqref{eq:29bb},
we only need to compute part of its entries ${[\bm \gamma]}_{x}$ for $\forall x \in {\cal B}_0$.

For the first iteration, since ${\cal A}_1 \times {\cal B}_1=\{a,b\} \times \{a,b,d,j\}$,
then we only need to compute $18 \ (=3 \times 6)$ entries ${[\mathbf{M}_1]}_{x,y}$ for $\forall (x,y)\in \{a,b,j\}\times \{a,b,d,f,i,j\}$,
skipping the computations of $207 \ (={15}^{2} - 18)$ remaining entries in $\mathbf{M}_1$.
After $K=10$ iterations, many unnecessary node-pairs are pruned,
as in part highlighted in the gray rows of the table in Fig.~\ref{fig:01}. \qed
\end{example}
\section{Algorithms \& Analysis} \label{app:04}
\subsection{{\IncUSRone} Algorithm} \label{app:04a}
\setcounter{algocf}{5}
\begin{algorithm}[t]
\small
\DontPrintSemicolon
\SetKwInOut{Input}{Input}
\SetKwInOut{Output}{Output}
\Input{a directed graph $G=(V,E)$, \\
       a new edge $(i,j)_{i \in V, \ j \in V}$ inserted to $G$, \\
       the old similarities $\mathbf{S}$ in $G$, \\
       the number of iterations $K$, \\
       the damping factor $C$.}
\Output{the new similarities ${\mathbf{\tilde{S}}}$ in $G \cup\{(i,j)\}$.}
\nl \label{ln:a01-01} initialize the transition matrix $\mathbf{Q}$ in $G$ ; \;
\nl \label{ln:a01-02}  ${{d}_{j}}:=$ in-degree of node $j$ in $G$ ; \;
\nl \label{ln:a01-03} memoize $\mathbf{w} := \mathbf{Q}\cdot {{[\mathbf{S}]}_{\star,i}}$ ; \;
\nl \label{ln:a01-04} compute $\lambda := {{[\mathbf{S}]}_{i,i}}+\tfrac{1}{C} \cdot {[\mathbf{S}]}_{j,j}-2\cdot {[\mathbf{w}]}_{j} - \tfrac{1}{C} +1$ ; \;
\nl \label{ln:a01-06}         \uIf {${{d}_{j}}=0$} {
\nl \label{ln:a01-06a}              $\mathbf{u} := \mathbf{e}_j, \ \mathbf{v} := \mathbf{e}_i, \ {\bm \gamma} :=  \mathbf{w} +\frac{1}{2}{{[\mathbf{S}]}_{i,i}}\cdot {{\mathbf{e}}_{j}}$; \;}
\nl \label{ln:a01-07}         \Else {
\nl \label{ln:a01-07a}           $\mathbf{u}:= \tfrac{1}{d_j+1} \mathbf{e}_j, \quad \mathbf{v} := \mathbf{e}_i-{[\mathbf{Q}]}_{j,\star}^T$ ; \;
\nl \label{ln:a01-08}  ${\bm\gamma} := \tfrac{1}{({{d}_{j}}+1)} \big( \mathbf{w}-\frac{1}{C}  {{[\mathbf{S}]}_{\star,j}}+( \frac{\lambda }{2\left( {{d}_{j}}+1 \right)}+ \frac{1}{C}-1 )  {{\mathbf{e}}_{j}} \big)$; }
\nl \label{ln:a01-13}      initialize ${{\bm{\xi }}_{0}} := C \cdot \mathbf{e}_j,\quad {{\bm{\eta }}_{0}} := {\bm\gamma},\quad {{\mathbf{M}}_{0}} := C \cdot \mathbf{e}_j \cdot {\bm\gamma}^T$ ; \;
\nl \label{ln:a01-14}  \For {$k=0,1,\cdots, K-1$} {
\nl \label{ln:a01-15}   ${{\bm{\xi }}_{k+1}} := C \cdot \mathbf{Q}\cdot {{\bm{\xi }}_{k}} + C \cdot (\mathbf{v}^T\cdot {{\bm{\xi }}_{k}}) \cdot \mathbf{u}$ ; \;
\nl \label{ln:a01-16}   ${{\bm{\eta }}_{k+1}} := \mathbf{Q}\cdot {{\bm{\eta }}_{k}} + (\mathbf{v}^T\cdot {{\bm{\eta }}_{k}}) \cdot \mathbf{u}$ ;\;
\nl \label{ln:a01-17}   ${{\mathbf{M}}_{k+1}} := {{\bm{\xi }}_{k+1}}\cdot \bm{\eta }_{k+1}^{T}+{{\mathbf{M}}_{k}}$ ; \;
}
\nl \label{ln:a01-18} \Return $\tilde{\mathbf{S}} := \mathbf{S} + \mathbf{M}_{K} + \mathbf{M}_{K}^T$ ; \;
\caption{\IncUSRone~($G, (i,j), \mathbf{S}, K, C$)}  \label{alg:01}
\end{algorithm}

Algorithm~\ref{alg:01} illustrates the pseudo code of \IncUSRone.

Given an old graph $G=(V,E)$, a new edge $(i,j)$ with $i \in V$ and $j \in V$ to be inserted to $G$, the old similarities $\mathbf{S}$ in $G$,
and the damping factor $C$,
{\IncUSRone} incrementally computes $\tilde{\mathbf{S}}$ in $G \cup \{(i,j)\}$ as follows:

First, it initializes the transition matrix $\mathbf{Q}$ and in-degree $d_j$ of node $j$ in $G$ (lines \ref{ln:a01-01}--\ref{ln:a01-02}).
Using $\mathbf{Q}$ and $\mathbf{S}$,
it precomputes the auxiliary vector $\mathbf{w}$ and scalar $\lambda$ (lines \ref{ln:a01-03}--\ref{ln:a01-04}).
Once computed, both $\mathbf{w}$ and $\lambda$ are memoized for precomputing
(i) the vectors $\mathbf{u}$ and $\mathbf{v}$ for a rank-one factorization of $\mathbf{\Delta Q}$, and
(ii) the initial vector ${\bm \gamma}$  for subsequent $\mathbf{M}_{k}$ iterations (lines \ref{ln:a01-06}--\ref{ln:a01-08}).
Then, the algorithm maintains two auxiliary vectors ${\bm{\xi }}_{k}$ and ${\bm{\eta }}_{k}$ to iteratively compute matrix $\mathbf{M}_{k}$ (lines \ref{ln:a01-13}--\ref{ln:a01-17}).
The process continues until the number of iterations reaches a given $K$.
Finally, the new $\tilde{\mathbf{S}}$ is obtained by $\mathbf{M}_{K}$\footnote{We can show ${\|\mathbf{M}_{K}-\mathbf{M}\|}_{\max}\le C^{K+1}$ with $\mathbf{M}$ in Eq.\eqref{eq:29c}.} (line \ref{ln:a01-18}).

\noindent \textbf{Correctness.} \
\IncUSRone~can \emph{correctly} compute new SimRanks for edge update that does not accompany new node insertions,
as verified by Theorems \ref{thm:01}--\ref{thm:03}. \

\noindent \textbf{Complexity.} \
The total complexity of \IncUSRone~is bounded by $O(Kn^2)$ time and $O(n^2)$ memory in the worst case for updating \emph{all} similarities of $n^2$ node-pairs.
Precisely,
\IncUSRone~runs in two phases:
preprocessing (lines \ref{ln:a01-01}--\ref{ln:a01-08}),
and incremental iterations (lines \ref{ln:a01-13}--\ref{ln:a01-18}):

(a) For the preprocessing,
it requires $O(m)$ time in total ($m$ is the number of edges in the old $G$),
which is dominated by computing $\mathbf{w}$ (lines \ref{ln:a01-03}), involving the matrix-vector multiplication $\mathbf{Q}\cdot {{[\mathbf{S}]}_{\star,i}}$.
The time for computing vectors $\mathbf{u}, \mathbf{v}, {\bm \gamma}$ is bounded by $O(n)$,
which includes only vector scaling and additions, \ie SAXPY.

(b) For the incremental iterative phase,
computing ${{\bm{\xi }}_{k+1}}$ and ${{\bm{\eta }}_{k+1}}$ needs $O(m+n)$ time for each iteration (lines \ref{ln:a01-15}--\ref{ln:a01-16}).
Computing ${{\mathbf{M}}_{k+1}}$ entails $O(n^2)$ time for performing one outer product of two vectors and one matrix addition (lines \ref{ln:a01-17}).
Thus, the cost of this phase is $O(Kn^2)$ time for $K$ iterations.

Collecting (a) and (b), all $n^2$ node-pair similarities can be incrementally computed in $O(Kn^2)$ total time.
%
\subsection{{\IncSR} Algorithm with Pruning} \label{app:04b}
\begin{algorithm}[t]
\small
\DontPrintSemicolon
\SetKwInOut{Input}{Input / Output}
\Input{the same as Algorithm~\ref{alg:01}.}
\lnlset{ln:a02-01}{1-2} the same as Algorithm~\ref{alg:01} ;\;
\lnlset{ln:a02-02}{3} find ${{\mathsf{\mathcal{B}}}_{0}}$ via Eq.\eqref{eq:41} ; \;
\lnlset{ln:a02-03}{} memoize ${[\mathbf{w}]}_{b} := {[\mathbf{Q}]}_{b,\star}\cdot {{[\mathbf{S}]}_{\star,i}}$,  for all $b \in {{\mathsf{\mathcal{B}}}_{0}}$ ; \;
\lnlset{ln:a02-04}{4-12} almost the same as Algorithm~\ref{alg:01} except that the computations of the entire vector $\bm \gamma$ in Lines $6,8,10,12$ are replaced by the computations of only parts of entries in $\bm \gamma$, respectively,
\eg  in Line~6 of Algorithm~\ref{alg:01}, ``${\bm \gamma}:=\mathbf{w}+\tfrac{1}{2}{[\mathbf{S}]}_{i,i}\cdot \mathbf{e}_j$''
are replaced by ``${[{\bm \gamma}]}_{b}:={[\mathbf{w}]}_{b}+\tfrac{1}{2}{[\mathbf{S}]}_{i,i}\cdot {[\mathbf{e}_j]}_{b}$, for all $b \in {\cal B}_{0}$'' ; \;
\lnlset{ln:a02-13}{13} ${[{{\bm{\xi }}_{0}}]}_j := C ,\ {[{\bm{\eta }}_{0}]}_{b} := {[{\bm\gamma}]}_{b}, \ {[{\mathbf{M}}_{0}]}_{j,b} := C \cdot {[{\bm\gamma}]}_{b}, \forall b \in {\cal B}_0$;\;
\lnlset{ln:a02-14}{14} \For {$k=1,\cdots, K$} {
\lnlset{ln:a02-15}{15} find ${{\mathsf{\mathcal{A}}}_{k}}\times {{\mathsf{\mathcal{B}}}_{k}}$ via Eq.\eqref{eq:41} ; \;
\lnlset{ln:a02-16}{16}  memoize $\sigma_1 := C \cdot (\mathbf{v}^T\cdot {{\bm{\xi }}_{k-1}}), \ \sigma_2 :=\mathbf{v}^T\cdot {{\bm{\eta }}_{k-1}} $ ; \;
\lnlset{ln:a02-17}{17} ${[{\bm{\xi }}_{k}]}_{a} := C \cdot {[\mathbf{Q}]}_{a, \star}\cdot {{\bm{\xi }}_{k-1}} +  \sigma_1 \cdot {[\mathbf{u}]}_{a}$, for all $a \in {\cal A}_{k}$ ; \;
\lnlset{ln:a02-18}{18} ${[{\bm{\eta }}_{k}]}_{b} := {[\mathbf{Q}]}_{b,\star}\cdot {{\bm{\eta }}_{k-1}} + \sigma_2 \cdot {[\mathbf{u}]}_{b}$, for all $b \in {\cal B}_{k}$ ;\;
\lnlset{ln:a02-19}{19} \scalebox{0.92}{${[{\mathbf{M}}_{k}]}_{a,b} := {[{\bm{\xi }}_{k}]}_{a}\cdot {[\bm{\eta }_{k}]}_{b}+{[{\mathbf{M}}_{k-1}]}_{a,b}, \ \forall (a,b) \in {{\mathsf{\mathcal{A}}}_{k}}\times {{\mathsf{\mathcal{B}}}_{k}}$;}\;
}
\lnlset{ln:a02-20}{20} \scalebox{0.95}{${[\tilde{\mathbf{S}}]}_{a,b} := {[\mathbf{S}]}_{a,b} + {[\mathbf{M}_{K}]}_{a,b} + {[\mathbf{M}_{K}]}_{b,a}, \ \forall (a,b) \in {{\mathsf{\mathcal{A}}}_{K}}\times {{\mathsf{\mathcal{B}}}_{K}}$;}\;
\lnlset{ln:a02-21}{21} \Return $\tilde{\mathbf{S}}$ ; \;
\caption{\IncSR~($G, \mathbf{S}, K, (i,j), C$)}  \label{alg:02}
\end{algorithm}

Algorithm~\ref{alg:02} illustrates the pseudo code of {\IncSR}.

\noindent \textbf{Correctness.} \
\IncSR~can \emph{correctly} prune the node-pairs with a-priori zero scores in $\mathbf{\Delta S}$,
which is verified by Theorem~\ref{thm:04}.
It also \emph{correctly} returns the new similarities,
as evidenced by Theorems \ref{thm:01}--\ref{thm:03}.

\noindent \textbf{Complexity.} \
The total time of \IncSR~is $O(K(m+|\AFF|))$ for $K$ iterations,
where $|\AFF|:= \textrm{avg}_{k \in[0,K]} ( |{\cal A}_k| \cdot |{\cal B}_k|)$
with ${\cal A}_k, {\cal B}_k$ in Eq.\eqref{eq:41},
being the average size of ``affected areas'' in $\mathbf{M}_k$ for $K$ iterations.
More concretely, (a) for the preprocessing,
finding ${{\mathsf{\mathcal{B}}}_{0}}$ (line~\ref{ln:a02-02}) needs $O(dn)$ time.
Utilizing ${{\mathsf{\mathcal{B}}}_{0}}$,
computing ${[\mathbf{w}]}_{b}$ reduces from $O(m)$ to $O(d|{\cal B}_0|)$ time,
with $|{\cal B}_0| \ll n$.
Analogously, $\bm \gamma$ in lines 6,8,10,12 of Algorithm~\ref{alg:01} needs only $O(|{\cal B}_0|)$ time.
(b) For each iteration,
finding ${\cal A}_k \times {{\mathsf{\mathcal{B}}}_{k}}$ (line \ref{ln:a02-15}) entails $O(dn)$ time.
Memoizing $\sigma_1, \sigma_2$ needs $O(n)$ time (line \ref{ln:a02-16}).
Computing ${\bm \xi}$ (\Resp $\bm \eta$) reduces from $O(m)$ to $O(d|{\cal A}_k|)$ (\Resp $O(d|{\cal B}_k|)$) time (lines \ref{ln:a02-17}--\ref{ln:a02-18}).
Computing ${[{\mathbf{M}}_{k}]}_{a,b}$ reduces from $O(n^2)$ to $O(|{\cal A}_k||{\cal B}_k|)$ time (line \ref{ln:a02-19}).
Thus,
the total time complexity can be bounded by $O(K(m+|\AFF|))$ for $K$ iterations.

It is worth mentioning that \IncSR, in the worst case, has the same complexity bound as \IncUSR.
However, in practice, $|\AFF| \ll n^2$, as demonstrated by our experimental study in Fig.\ref{fig:exp_08}.

%
\section{Description of Real Datasets} \label{app:05}
The description of the real datasets is as follows:

(1) \underline{\DBLP}\footnote{http://dblp.uni-trier.de/\~{}ley/db/} is a co-citation graph,
where each node is a paper with attributes (\eg publication year), and edges are citations.
By virtue of the publication year,
we extract several snapshots. 

(2) \underline{\CITH}\footnote{http://snap.stanford.edu/data/} is a reference network (cit-HepPh) from e-Arxiv.
If a paper $u$ references $v$,
there is a link $u \to v$.

(3) \underline{\YOUTU}\footnote{http://netsg.cs.sfu.ca/youtubedata/} is a YouTube network,  
where a video $u$ (node) is linked to $v$ if $v$ is in the relevant video list of $u$.
We extract snapshots based on the age of videos.

(4) \underline{{\WEBB}} is a Berkeley-Stanford web graph, where nodes are pages from \textsf{berkely.edu} and \textsf{stanford.edu} domains,
and edges are hyperlinks.

(5) \underline{{\WEBG}} is a Google web graph, where nodes are web pages, and edges are hyperlinks.

(6) \underline{{\CITP}} is a patent citation network among US, where a node is a granted patent, and a link a citation.

(7) \underline{{\SOCL}} is a LiveJournal friendship social network, where a node is a user, and a link denotes friendship.

(8) \underline{{\UK}}\footnote{http://law.di.unimi.it/datasets.php} is a web graph obtained from a 2005 crawl of the \textsf{.uk} domain, where an edge is a link from one web page (node) to another.

(9) \underline{{\IT}} is a web graph crawled from the \textsf{.it} domain, where an edge is a link from a page (node) to another.

\end{appendices}
\end{document}